\documentclass[a4paper,UKenglish]{lipics}
 
\usepackage{microtype}%
\usepackage{alltt,paralist,bcprules,amsmath,amssymb,proof,myproof,graphicx,url,color}
\theoremstyle{plain}

\newenvironment{defn}{\vspace*{1ex}\begin{definition}}{\end{definition}\vspace*{1ex}}
\newif\iffull\fulltrue
\newif\ifdraft\draftfalse

\newenvironment{pfof}[1]{\paragraph*{Proof of {#1}.}}{}

\def\qed{\(\Box\)}

\newcommand\mylongseq[2]{{\overrightarrow{#1}}^{#2}}
\newcommand\balanced{\mathbf{bal}}
\newcommand\Wlang{\mathcal{L}_{\mathtt{w}}}
\newcommand\uty{\delta}

\newcommand\rname[1]{(\rn{#1})}
\newcommand\subterm{\preceq}

\newcommand\ttytosty[1]{\mathbin{[\![}#1\mathbin{]\!]}}
\newcommand\utytosty[1]{\mathbin{[\![}#1\mathbin{]\!]}}
\newcommand\remeps[1]{#1{\uparrow_{\Te}}}

\newcommand\TTE{\Xi}
\newcommand\tty{\xi}
\newcommand\Tempty{\T_{\epsilon}}
\newcommand\Tplus{\T_{+}}
\newcommand\comment[1]{}

\newcommand\tropt[1]{}
\newcommand\TREE[1]{\mathbf{Tree}_{#1}}

\newcommand\notred{\mathrel{\;\not\!\!\red}}

\newcommand\emphd[1]{\emph{#1}}
\newcommand\pK{\p_{\mathtt{ST}}}
\newcommand\SE{\mathcal{K}}

\newcommand\Lang{\mathcal{L}}

\newcommand\LLang{\mathcal{L}_{\mathtt{leaf}}}
\newcommand\leaves{\textbf{leaves}}

\newcommand\DSubst[1]{#1}
\newcommand\mvsub{\theta}

\newcommand\NONTERMS{\mathcal{N}}

\newcommand\tyTOsty[1]{\mathbin{[\![}#1\mathbin{]\!]}}
\newcommand\tr{\Rightarrow}

\newcommand\order{\mathtt{order}}

\newcommand\app[1]{{#1}\,}

\newcommand\tset[1]{\set{#1}}

\newcommand\FV{\mathbf{FV}}

\newcommand\nk[1]{{\footnotesize \color{blue}{[#1 -nk]}}}
\newcommand\T{\mathtt{o}}

\newcommand\DCOL{\mathbin{::}}
\newcommand\IT{\bigwedge}
\newcommand\TE{\Gamma}

\newcommand\sty{\kappa}

\newcommand\ity{\sigma}

\newcommand\COL{\mathbin{:}}
\newcommand\p{\vdash}

\newcommand\GRAM{\mathcal{G}}
\newcommand\TERMS{\Sigma}

\newcommand\Ta{\mathtt{a}}
\newcommand\Tb{\mathtt{b}}

\newcommand\Te{\mathtt{e}}
\newcommand\Teps{\Te}

\newcommand\TT[1]{\mathtt{#1}}

\newcommand\RULES{\mathcal{R}}

\newcommand\red{\longrightarrow}
\newcommand\reds{\longrightarrow^*}
\newcommand\redwith[1]{\longrightarrow_{#1}}

\newcommand\redswith[1]{\longrightarrow^*_{#1}}

\newcommand\arity{\mathtt{ar}}

\newcommand\set[1]{\{#1\}}
\newcommand\Hra{\rightarrow}
\newcommand\ra{\rightarrow}

\newcommand\dom{\mathit{dom}}

\newcommand\REL{\mathcal{R}}

\newcommand\InfruleS[3]{\begin{minipage}{#1\textwidth}\infrule{#2}{#3}\end{minipage}}
\newcommand\InfruleSR[4]{\begin{minipage}{#1\textwidth}\infrule[#2]{#3}{#4}\end{minipage}}

\newcommand\UE{\Gamma}

\newcommand\myemph[1]{\emph{#1}}

\newcommand{\ot}{\mathtt{o}}

\newcommand{\lap}{\langle}
\newcommand{\rap}{\rangle}

\newcommand{\urr}[2]{#1 \DCOL_{\mathrm{u}} #2}
\newcommand{\brr}[2]{#1 \DCOL_{\mathrm{b}} #2}
\newcommand{\envdom}[1]{\dom(#1)}

\newcommand{\ibr}{\mathbin{*}}

\newcommand{\des}[1]{#1^{\bullet}} %
\newcommand{\sbs}[1]{#1^{\bullet}} %

\newcommand{\lse}{\lesssim}
\newcommand{\gse}{\gtrsim}
\newcommand{\se}{\sim}
\newcommand{\vsim}{\sim_{\mathrm{v}}}
\newcommand{\las}{\sqsubseteq}

\newcommand{\Y}{\ensuremath{Y}}
\newcommand{\lambday}{{\ensuremath{\lambda}\Y}}
\newcommand{\lambdayg}[1]{\ensuremath{\lambday_{\!#1}}}

\newcommand{\defe}{:=}

\newcommand{\odr}[1]{\order(#1)}
\newcommand{\rt}{\mathtt{oR}}
\newcommand{\rf}[1]{(#1)^{\sharp}}
\renewcommand{\st}[1]{(#1)^{\sharp}_{\ot}}
\newcommand{\oty}{\rho}
\newcommand{\RE}{\Phi}
\newcommand{\fg}[1]{(#1)^{\flat}}

\newcommand{\appseq}[2]{{\overrightarrow{#1}}^{#2}}
\newcommand{\ind}[2]{#1 \le #2}

\newcommand{\appi}[4]{\app{#1}{\appseq{#2}{\ind{#3}{#4}}}}

\newcommand{\reason}[1]{\text{(#1)}}

\newenvironment{myitemize}{%
\begin{itemize}\setlength{\parindent}{15pt}
}{%
\end{itemize}
}
\newcommand{\myitem}{\item}

\newcommand{\mn}[2]{#1 \setminus #2}
\newcommand{\res}[1]{\lap #1 \rap}

\newcommand{\eps}{\Te} %

\newcommand{\onm}{All}
\newcommand{\rnm}{Right}
\newcommand{\rtVar}{RTy-Var}
\newcommand{\rtBrO}{RTy-Br\onm}
\newcommand{\rtBrR}{RTy-Br\rnm}
\newcommand{\rtAlph}{RTy-Alph}
\newcommand{\rtEpsO}{RTy-Eps\onm}
\newcommand{\rtEpsR}{RTy-Eps\rnm}
\newcommand{\rtNtO}{RTy-Nt\onm}
\newcommand{\rtNtR}{RTy-Nt\rnm}
\newcommand{\rtApp}{RTy-App}
\newcommand{\rtSet}{RTy-Set}
\newcommand{\rtAbs}{RTy-Abs}

\newcommand{\VT}{\mathbb{V}}
\newcommand{\ip}[1]{[\![#1]\!]}
\newcommand{\coset}[1]{[#1]_{\vsim}}
\newcommand{\cterm}[1]{\mathrm{Term}_{#1}}

\newcommand{\br}{\TT{br}}
\newcommand{\Br}{\mathit{Br}}
\newcommand{\pp}{\p_{\mathrm{s}}}

\newcommand{\TEpz}{\TE'}

\newcommand{\empword}{\varepsilon}

\newcommand\wtt[1]{#1^{\star}} %

\newcommand{\LLangE}{\LLang^{\empword}}

\newcommand{\conv}[1]{{#1}^{\circ}}
\newcommand\ipc[1]{\ip{#1}^{\circ}}
\newcommand\cfun{\to}

\newcommand\id{\mathit{id}}
\newcommand\concat{\cdot}
\newcommand\sem[1]{\mathbf{Tr}_{#1}}
\newcommand\vtg[1]{\sem{#1}}
\newcommand\R{R'}

\newcommand\upsymbol{\triangleleft}
\newcommand\up[1]{{#1}^{\upsymbol}}
\newcommand\pup[1]{\up{(#1)}}
\newcommand\eup[2]{{(#1)}^{\upsymbol}_{#2}}
\newcommand\nSE{[\,]}
\newcommand\cSE[2]{ #1 , #2 }
\newcommand\K{K}

\newcommand\asd[1]{{\footnotesize \color[RGB]{105,10,130}[#1 -asd]}}

\ifdraft\else
\renewcommand\nk[1]{}
\renewcommand\asd[1]{}

\fi

\iffull
\title{On Word and Frontier Languages of Unsafe Higher-Order Grammars}
\else
\title{On Word and Frontier Languages of Unsafe Higher-Order Grammars\footnote{%
A full version~\cite{1604.01595} of the paper is available at
\url{http://arxiv.org/abs/1604.01595}.}}
\fi
\titlerunning{On Unsafe Path and Frontier Languages}
\author[1]{Kazuyuki Asada}
\author[1]{Naoki Kobayashi}
\affil[1]{The University of Tokyo}
\Copyright{Kazuyuki Asada and Naoki Kobayashi}
\serieslogo{}%
\volumeinfo%
  {Billy Editor and Bill Editors}%
  {2}%
  {Conference title on which this volume is based on}%
  {1}%
  {1}%
  {1}%
\EventShortName{}
\DOI{10.4230/LIPIcs.xxx.yyy.p}%
\begin{document}
\maketitle
\ifdraft
version: \today
\fi
\setcounter{footnote}{0}
\iffull
\begin{abstract}
Higher-order grammars are an extension of regular and context-free grammars,
where non-terminals may take parameters. They have been extensively studied in
1980's, and restudied recently in the context of model checking and program verification.
We show that the class of unsafe order-(\(n+1\)) word languages coincides with 
the class of frontier languages of unsafe order-\(n\) tree languages. 
We use intersection types for transforming an order-(\(n+1\)) word grammar
to a corresponding order-\(n\) tree grammar.
The result has been proved for safe languages by Damm in 1982, but 
it has been open for unsafe languages, to our knowledge.
Various known results on higher-order grammars can be obtained as almost immediate
corollaries of our result.
\end{abstract}
\fi %
\section{Introduction}
\label{sec:intro}
Higher-order grammars are an extension of regular and context-free grammars,
where %
non-terminals may take trees or (higher-order)
functions on trees as parameters.
They were extensively studied in the 1980's~\cite{DammTCS,Engelfriet91,EngelfrietHTT},
and recently reinvestigated in the context of model checking~\cite{Knapik01TLCA,Ong06LICS}
and applied to program verification~\cite{Kobayashi13JACM}.

The present paper shows that the class of unsafe order-(\(n+1\)) word languages coincides with 
the class of ``frontier languages'' of unsafe order-\(n\) tree languages.
Here, the frontier of a tree is the sequence of symbols that occur in the leaves of the tree
from left to right,
and the frontier language of a tree language consists of the frontiers of elements of
the tree language.
The special case where \(n=0\) corresponds to the well-known fact that
the frontier language of a regular tree language is a context-free language.
The result has been proved by Damm~\cite{DammTCS} for grammars with the 
safety restriction (see~\cite{Salvati15OI} for a nice historical account of
the safety restriction), but it has been open for unsafe grammars, to our 
knowledge.\footnote{Kobayashi et al.~\cite{KIT14FOSSACS} mentioned the result, referring
to the paper under preparation: ``On Unsafe Tree and Leaf Languages,''  which is actually
the present paper.}

Damm's proof relied on the safety restriction (in particular, the fact
that variable renaming is not required for safe
grammars~\cite{Ong09safety}) and does not apply (at least directly) to
the case of unsafe grammars. We instead use intersection types to
transform an order-(\(n+1\)) word grammar \(\GRAM\) to an order-\(n\)
tree grammar \(\GRAM'\) such that the frontier language of \(\GRAM'\)
coincides with the language generated by \(\GRAM\). Intersection types
have been used for recent other studies of higher-order grammars and
model
checking~\cite{Kobayashi13JACM,KIT14FOSSACS,Kobayashi13LICS,KO09LICS,SalvatiW15,Parys14,KMS13HOSC,Tsukada14LICS};
our proof in the present paper provides even more evidence that
intersection types are a versatile tool for studies of higher-order
grammars. Compared with the previous work on intersection types for higher-order grammars,
the technical novelties include: (i) our intersection types (used in Section~\ref{sec:leaf})
are mixtures of non-linear and linear intersection types and (ii) our type-based transformation involves
global restructuring of terms. These points have made the correctness 
of the transformations non-trivial and delicate.

As stressed by Damm~\cite{DammTCS} at the beginning of his paper, 
the result will be useful for analyzing properties of
higher-order languages by induction on the order of grammars. Our result allows
properties on (unsafe) order-\(n\) languages to be 
reduced to those on order-\((n-1)\) tree languages, and then the latter
may be studied by investigating those on the path languages of order-\((n-1)\) tree languages,
which are order-\((n-1)\) word languages. 
As a demonstration of this, we discuss an application to (a special case of) 
the diagonal problem 
for unsafe languages~\cite{Parys16diagonality} in Section~\ref{sec:app}, along with other applications.

The rest of this paper is structured as follows. Section~\ref{sec:pre} reviews the definition of
higher-order grammars, and states the main result.
Sections~\ref{sec:leaf} and \ref{sec:leaf2} prove the result by providing
 the (two-step) transformations from order-\((n+1)\) word grammars to order-\(n\) tree grammars.
Section~\ref{sec:app} discusses applications of the result. Section~\ref{sec:related} discusses
related work and Section~\ref{sec:conc} concludes the paper.
\section{Preliminaries}
\label{sec:pre}

This section defines higher-order grammars and the languages generated by them,
and then explains the main result.
Most of the following definitions follow those in \cite{KIT14FOSSACS}.

A higher-order grammar consists of non-deterministic rewriting rules of the form \(A \to t\), where
\(A\) is a non-terminal and \(t\) is a simply-typed \(\lambda\)-term that may contain non-terminals
and terminals (tree constructors).
\begin{defn}[types and terms]
The set of \emphd{simple types},\footnote{We sometimes call simple types 
\emph{sorts} in this paper, to avoid confusion
with intersection types introduced later for grammar transformations.} ranged over by \(\sty\), is given by:
\(
\sty %
::= \T \mid \sty_1\ra \sty_2
\).
The order and arity of a simple type \(\sty\), written
\(\order(\sty)\) and \(\arity(\sty)\), are defined respectively by:
\[
\begin{array}{l}
\order(\T)=0\qquad \order(\sty_1\ra\sty_2)=\max(\order(\sty_1)+1,\order(\sty_2))\\
\arity(\T)=0\qquad \arity(\sty_1\ra\sty_2)=1+\arity(\sty_2)\\
\end{array}
\]
The type \(\T\) describes trees, and \(\sty_1\to\sty_2\) describes functions
from \(\sty_1\) to \(\sty_2\).
The set of \(\lambda\)-\emphd{terms}, ranged over by \(t\), is defined by:
\(
t %
::= x \mid A \mid a \mid \app{t_1}{t_2}\mid \lambda x\COL\sty.t
\).
Here, \(x\) ranges over variables, \(A\) over symbols called non-terminals,
and \(a\) over symbols called terminals. We assume that each terminal \(a\) has
a fixed arity; we write \(\TERMS\) for the map from terminals to their arities.
A term \(t\) is called an \emphd{applicative term} (or simply
a \emphd{term}) if it does not contain \(\lambda\)-abstractions.
A (simple) type environment \(\SE\) is a
sequence of type bindings of the form \(x\COL\sty\)
such that if \(\SE\) contains \(x\COL\sty\) and \(x'\COL\sty'\) in different positions
then \(x \neq x'\).
In type environments, {non-terminals} are also treated as variables. 
A \(\lambda\)-term \(t\) has type \(\sty\) under \(\SE\) if \(\SE\pK t:\sty\)
is derivable from the following typing rules.\\
\InfruleS{0.45}{}{\SE,\,x\COL\sty,\,\SE' \pK x:\sty}
\InfruleS{0.45}{}{\SE \pK a:\underbrace{\T\ra\cdots\ra\T}_{\TERMS(a)}\ra\T}

\InfruleS{0.45}{\SE\pK t_1:\sty_2\ra\sty\andalso \SE\pK t_2:\sty_2}
  {\SE\pK \app{t_1}{t_2}:\sty}
\InfruleS{0.45}{\SE,\,x\COL\sty_1 \pK t:\sty_2}
  {\SE\pK \lambda x\COL\sty_1.t:\sty_1\ra\sty_2}\\
We call \(t\) a (finite, \(\TERMS\)-ranked) \emphd{tree} if \(t\) is an applicative
term consisting of only terminals, and \(\pK t:\T\) holds.
We write \(\TREE{\TERMS}\) for the set of \(\TERMS\)-ranked trees,
and use the meta-variable \(\pi\) for a tree.
\end{defn}

We often omit type annotations and just write 
\(\lambda x.t\) for \(\lambda x\COL\sty.t\).
We consider below only well-typed \(\lambda\)-terms of the form
\(\lambda x_1.\cdots\lambda x_k.t\), where \(t\) is an applicative term.
We are now ready to define higher-order grammars.

\begin{defn}[higher-order grammar]
A \emphd{higher-order grammar} is a quadruple \((\TERMS,\NONTERMS,\RULES,S)\), where
\begin{inparaenum}[(i)]
\item \(\TERMS\) is a ranked alphabet;
\item \(\NONTERMS\) is a map from a finite set of non-terminals 
to their types;
\item \(\RULES\) is a finite set of \emphd{rewriting rules} of the form
\( A\,\to \lambda x_1.\cdots \lambda x_\ell.t\),
where \(\NONTERMS(A) = \kappa_1\ra\cdots\ra\kappa_\ell\ra\T\),
\(t\) is an applicative term, and 
\(\NONTERMS,x_1\COL\kappa_1,\ldots,x_\ell\COL\kappa_\ell\pK t:\T\) holds for some 
\iffull
\(\kappa_1,\ldots,\kappa_\ell\)\footnote{%
We assume some total order on non-terminals
and %
regard \(\NONTERMS\) as a sequence;
the condition \(\NONTERMS,x_1\COL\kappa_1,\ldots,x_\ell\COL\kappa_\ell\pK t:\T\)
does not depend on the order.}.
\fi
\item \(S\) is a non-terminal called \emphd{the start symbol},
and \(\NONTERMS(S)=\T\). 
\end{inparaenum}
The \emphd{order} of a grammar \(\GRAM\), written
\(\order(\GRAM)\), is the largest order 
of the types of non-terminals.
We sometimes write \(\TERMS_\GRAM,\NONTERMS_\GRAM,\RULES_\GRAM,S_\GRAM\) for the four components
of \(\GRAM\).

For a grammar \(\GRAM=(\TERMS,\NONTERMS,\RULES,S)\), 
the rewriting relation \(\redwith{\GRAM}\) is defined by:\\[1ex]
\InfruleS{0.46}{(A\to \lambda x_1.\cdots \lambda x_k.t)\in \RULES}
{A\,t_1\,\cdots\,t_k \redwith{\GRAM} [t_1/x_1,\ldots,t_k/x_k]t}
\InfruleS{0.48}{t_i\redwith{\GRAM}t_i'\andalso i\in\set{1,\ldots,k}\andalso \TERMS(a)=k}
 {a\,t_1\,\cdots\,t_k \redwith{\GRAM}
  a\,t_1\,\cdots\,t_{i-1}\,{t'_{i}}\,t_{i+1}\,\cdots\,t_k
}\\[1ex]
Here, \([t_1/x_1,\ldots,t_k/x_k]t\) is the term obtained by substituting 
\(t_i\) for the free occurrences of \(x_i\) in \(t\).
We write \(\redswith{\GRAM}\) for the reflexive transitive closure of \(\redwith{\GRAM}\).

The \emphd{tree language generated by} \(\GRAM\), written \(\Lang(\GRAM)\), is
the set \(\set{\pi\in\TREE{\TERMS_\GRAM} \mid S\redswith{\GRAM} \pi}\).
We call a grammar \(\GRAM\) a \emphd{word grammar} if 
all the terminal symbols have arity \(1\) except the special terminal \(\TT{e}\), whose
arity is \(0\).
The \emphd{word language} generated by a word grammar \(\GRAM\), written \(\Wlang(\GRAM)\), is 
\(\set{a_1\cdots a_n \mid a_1(\cdots (a_n\,\TT{e})\cdots)\in\Lang(\GRAM)}\).
The frontier word of a tree \(\pi\), written \(\leaves(\pi)\),
is the sequence of symbols in the leaves of \(\pi\). It is
defined inductively by:
\(\leaves(a) = a\)  when \(\TERMS(a)=0\), and \(\leaves(a\,\pi_1\,\cdots\,\pi_k) =
\leaves(\pi_1)\cdots \leaves(\pi_k)\) when \(\TERMS(a)=k>0\).  
The \emphd{frontier language} generated by \(\GRAM\), written \(\LLang(\GRAM)\),
is the set:
\(\set{\leaves(\pi) \mid S\redswith{\GRAM} \pi\in\TREE{\TERMS_\GRAM}}\).
In our main theorem,
we assume that there is a special nullary symbol \(\Te\)
and consider \(\TT{e}\in \LLang(\GRAM)\) as the empty word \(\empword\); i.e.,
we consider \(\LLangE(\GRAM)\) defined by:
\[
\LLangE(\GRAM) \defe
(\LLang(\GRAM) \setminus \set{\Te}) \cup \set{\empword \mid \Te \in \LLang(\GRAM)}.
\]
\end{defn}

We note that the classes of order-0 and order-1 word languages coincide with those
of regular and context-free languages respectively.
We often write 
\(A\,x_1\,\cdots\,x_k\to t\) for the rule \(A\to \lambda x_1.\cdots \lambda x_k.t\).
When considering the frontier language of a tree grammar, we assume,
without loss of generality, that
the ranked alphabet \(\TERMS\) has a unique binary symbol \(\TT{br}\), and 
that all the other terminals have arity \(0\).

\begin{example}
\label{ex:grammar}
Consider the order-2 (word) grammar \(\GRAM_1 = (\set{\TT{a}\COL 1, \TT{b}\COL 1, \TT{e}\COL 0},
\set{S\COL \T, F\COL (\T\to\T)\to \T, A\COL (\T\to\T)\to (\T\to\T),
B\COL (\T\to\T)\to (\T\to\T)},\REL_1,S)\), where \(\REL_1\) consists of:
\[
\begin{array}{l}
S\to F\,\TT{a}\qquad S\to F\,\TT{b}\qquad A\,f\,x \to \TT{a}(f\,x)\qquad B\,f\,x \to \TT{b}(f\,x),\\
F\,f\to f(f\,\TT{e})\qquad
F\,f\to F\,(A\,f)\qquad
F\,f\to F\,(B\,f).
\end{array}
\]
\(S\) is reduced, for example, as follows.
\[S \red{} F\,\TT{b}\red{} F\,(A\,\TT{b})
\red{} (A\,\TT{b})(A\,\TT{b}\,\TT{e}) 
\red{} \TT{a}\,(\TT{b}\,(A\,\TT{b}\,\TT{e}))
\red{}
\TT{a}(\TT{b}\,(\TT{a}(\TT{b}\,\TT{e}))).\]
The word language \(\Wlang(\GRAM_1)\) is
\(\set{ ww\mid w\in\set{\TT{a},\TT{b}}^+}\).

Consider the order-1 (tree) grammar \(\GRAM_2 = (\set{\TT{br}\COL 2,
\TT{a}\COL 0,\TT{b}\COL 0,\TT{e}\COL 0},
\set{S\COL\T,F\COL\T\to\T}, \REL_2,S)\), where \(\REL_2\) consists of:
\[
S\to F\,\TT{a}\qquad S\to F\,\TT{b}\qquad F\,f\to \TT{br}\,f\,f\qquad
F\,f\to F(\TT{br}\;\TT{a}\,f)\quad
F\,f\to F(\TT{br}\;\TT{b}\,f).
\]
The frontier language \(\LLangE(\GRAM_2)\) coincides with \(\Wlang(\GRAM_1)\) above.

\end{example}

The following is the main theorem we shall prove in this paper.
\begin{theorem}
\label{th:main}
For any order-\((n+1)\) word grammar \(\GRAM\) (\(n\geq 0\)),
there exists an order-\(n\)
tree grammar \(\GRAM'\) %
such that \(\Wlang(\GRAM) = \LLangE(\GRAM')\).
\end{theorem}

The converse of the above theorem also holds:
\begin{theorem}
\label{th:main-converse}
For any order-\(n\) tree grammar \(\GRAM'\) 
such that no word in \(\LLangE(\GRAM')\) contains \(\Te\),
there exists a word grammar \(\GRAM\) {of order at most \(n+1\)} such that
\(\Wlang(\GRAM) = \LLangE(\GRAM')\).
\end{theorem}
Since the construction of \(\GRAM\) is easy, we sketch it here;
\iffull
see Appendix~\ref{sec:ProofConverseMainTheorem} for a proof.
\else
\fi
For \(n\geq 1\), the grammar \(\GRAM\) is obtained by 
{(i) changing the arity of each nullary terminal 
\(a\; (\neq\Te)\) to one, i.e., \(\TERMS_{\GRAM}(a)\defe 1\),
(ii) replacing the terminal \(\Te\) with a new non-terminal \(E\) of type \(\T\to\T\), defined by
\(E\,x\to x\), and also} the unique binary terminal
\(\TT{br}\) with a new non-terminal \(\mathit{Br}\) of type \((\T\to\T)\to(\T\to\T)\to(\T\to\T)\), 
defined by \(\mathit{Br}\,f\,g\,x\to f(g\, x)\),
(iii) applying \(\eta\)-expansion {to the right hand side of each (original) rule} to add an order-0 argument,
and (iv) adding new start symbol \(S'\) with rule \(S' \to S \Te\).
For example, given the grammar \(\GRAM_2\) above,
the following grammar is obtained:
\[
\begin{array}{l}
{S' \to S\,\Te} \qquad
S\,x \to F\,{\Ta}\,x \qquad
S\,x\to F\,{\Tb}\,x\\
F\,f\,x\to \mathit{Br}\,f\,f\,x\qquad
F\,f\,x\to F(\mathit{Br}\,{\Ta}\,f)\,x\qquad
F\,f\,x\to F(\mathit{Br}\,{\Tb}\,f)\,x\\
{E\,x\to x} \qquad
\mathit{Br}\,f\,g\,x\to f(g\,x).
\end{array}
\]

Theorem~\ref{th:main} is proved by two-step grammar transformations, both of which
are based on intersection types.
In the first step, we transform an order-\((n+1)\) word grammar \(\GRAM\)
to an order-\(n\) tree grammar \(\GRAM''\) such that 
\(\Wlang(\GRAM) = \remeps{\LLang(\GRAM'')}\),
where \(\remeps{\Lang}\) is the word language
obtained from \(\Lang\) by removing {all the occurrences of} the special terminal \(\Te\); that is,
the frontier language of \(\GRAM''\) is almost the same as \(\Wlang(\GRAM)\), except
that the former may contain multiple occurrences of the special, dummy symbol \(\TT{e}\).
In the second step, we clean up the grammar to eliminate \(\TT{e}\) (except that a singleton tree
\(\TT{e}\) may be generated when \(\epsilon\in\Wlang(\GRAM)\)).
The first and second steps shall be formalized in Sections~\ref{sec:leaf} and
\ref{sec:leaf2} respectively.

For the target of the transformations,
we use the following extended terms, in which a \emph{set} of terms may occur
in an argument position:
\[
\begin{aligned}
u \mbox{ (extended terms) } &::= x \mid A \mid a \mid u_0 U \mid \lambda x. u
\\
U &::= \tset{u_1,\ldots,u_k} \ {(k \ge 1)}.
\end{aligned}
\]
Here, \(u_0\,u_1\) is interpreted as just a shorthand for \(u_0\tset{u_1}\).
Intuitively, \(\tset{u_1,\ldots,u_k}\) is considered a non-deterministic choice
\(u_1+\cdots+u_k\), which (lazily) reduces to \(u_i\) non-deterministically.
The typing rules are extended accordingly by:\\[1ex]
\InfruleS{0.48}{\SE\pK u_0:\sty_1\ra\sty\andalso \SE\pK U:\sty_1}
  {\SE\pK \app{u_0}{U}:\sty}
\InfruleS{0.48}{\SE\pK u_i:\sty \mbox{ for each $i\in\set{1,\ldots,k}$}}
  {\SE\pK \set{u_1,\ldots,u_k}:\sty}

An \emph{extended higher-order grammar} is the same as
a higher-order grammar, except that 
each rewriting rule in \(\RULES\) may be of the form \(\lambda x_1\cdots\lambda x_\ell.u\),
where \(u\) may be an applicative extended term.
The reduction rule for non-terminals is replaced by:
\infrule{(A\to \lambda x_1\,\cdots\,\lambda x_k.u)\in \RULES\andalso
  u'\in \DSubst{[U_1/x_1,\ldots,U_k/x_k]}{u}}
{A\,U_1\,\cdots\,U_k \redwith{\GRAM} u'}
where the substitution \(\DSubst{\mvsub}{u}\) is defined by:
\[
\begin{array}{l}
\DSubst{\mvsub}{a} = \set{a}\qquad
\DSubst{\mvsub}{x} = 
\begin{cases}
\mvsub(x) &\text{ (if \(x \in \dom(\mvsub)\))}
\\
\set{x} &\text{ (otherwise)}
\end{cases}\\
\DSubst{\mvsub}{(u_0 U)}
 = \set{v (\DSubst{\mvsub}{U}) \mid v\in\DSubst{\mvsub}{u_0}}\qquad
\DSubst{\mvsub}{\tset{u_1,\ldots,u_k}} = \DSubst{\mvsub}{u_1}\cup\cdots\cup\DSubst{\mvsub}{u_k}\,.
\end{array}
\]
Also, the other reduction rule is replaced by the following two rules:
\infrule{u\redwith{\GRAM}u'\andalso i\in\set{1,\ldots,k}\andalso \TERMS(a)=k}
 {a\,U_1\,\cdots\,U_{i-1}\,\set{u}\,U_{i+1}\,\cdots\,U_k
 \redwith{\GRAM}
  a\,U_1\,\cdots\,U_{i-1}\,\set{u'}\,U_{i+1}\,\cdots\,U_k
}
\infrule{u \in U_i \andalso U_i\text{ is not a singleton} \andalso i\in\set{1,\ldots,k}\andalso \TERMS(a)=k}
 {a\,U_1\,\cdots\,U_k \redwith{\GRAM}
  a\,U_1\,\cdots\,U_{i-1}\,\set{u}\,U_{i+1}\,\cdots\,U_k
}

Note that unlike in the extended grammar introduced in \cite{KIT14FOSSACS},
there is no requirement that each \(u_i\) in \(\tset{u_1,\ldots,u_k}\) is used at least once.
Thus, the extended syntax does not change the expressive power of grammars.
A term set \(\tset{u_1,\ldots,u_k}\) can be replaced by
\(A\,x_1\,\cdots\,x_\ell\) with the rewriting rules
\(A\,x_1\,\cdots\, x_\ell \Hra u_i\), 
where \(\set{x_1,\ldots,x_\ell}\) is the set of variables occurring in some of \(u_1,\ldots,u_k\).
In other words, for any order-\(n\) extended grammar \(\GRAM\), there is an (ordinary)
order-\(n\) grammar \(\GRAM'\) such that \(\Lang(\GRAM)=\Lang(\GRAM')\).

\section{Step 1: from order-\((n+1)\) grammars to order-\(n\) tree grammars}
\label{sec:leaf}
In this section, we show that
for any order-\((n+1)\) grammar \(\GRAM =(\TERMS,\NONTERMS,\RULES,S)\) such that 
\(\TERMS(\Te)=0\) and 
\(\TERMS(a)= 1\) for every \(a\in\dom(\TERMS)\setminus\set{\Te}\), 
there exists an order-\(n\) grammar {\(\GRAM'\)
 such that
\(\TERMS_{\GRAM'} = \set{\TT{br}\mapsto 2,\Te\mapsto 0}\cup\set{a\mapsto 0\mid \TERMS(a)=1}\)
and}
\(\Wlang(\GRAM)=\remeps{\LLang(\GRAM')}\).

For technical convenience, we assume below that, for every type \(\sty\) occurring in
\(\NONTERMS_\GRAM(A)\) for some \(A\), 
if \(\sty\) is of the form \(\T\ra\sty'\), then \(\order(\sty')\leq 1\).
This does not lose generality, since any function \(\lambda x\COL\T.t\) of type
\(\T\ra\sty'\) with \(\order(\sty')>1\) can be replaced by the term \(\lambda x'\COL\T\ra\T.[x'\Te/x]t\)
of type \((\T\to\T)\to\sty'\) (without changing the order of the term),
and any term \(t\) of type \(\T\) can be replaced by the term \(\K\,t\) of type 
\(\T\ra\T\), where 
\(\K\) is a non-terminal of type \(\T\ra\T\ra\T\), with rule \(\K\,x\,y\Hra x\).
\iffull
See Appendix~\ref{sec:TransAssumpOrder} for the details of this transformation.
\fi

The basic idea of the transformation is to remove all the order-0 arguments (i.e.,
arguments of tree type \(\T\)).
This reduces the order of each term by \(1\); for example, terms of types \(\T\to \T\) 
and \((\T\to \T)\to \T\) will respectively be transformed to those of types 
\(\T\) and \(\T\to\T\).
Order-0 arguments can indeed be removed as follows. Suppose we have a term \(t_1\, t_2\)
where \(t_1\COL\T\to \T\). If \(t_1\) does not use the order-0 argument \(t_2\), then
we can simply replace \(t_1\, t_2\) with \(t_1^\#\) (where
\(t_1^\#\) is the result of recursively applying the transformation to \(t_1\)).
If \(t_1\) uses the argument \(t_2\), the word generated by \(t_1\,t_2\) must be of the form
\(w_1w_2\), where \(w_2\) is generated by \(t_2\); in other words, \(t_1\) can only append
a word to the word generated by \(t_2\). Thus, \(t_1\,t_2\) can be transformed to
\(\TT{br}\;t_1^\#\;t_2^\#\), which can generate a tree whose frontier coincides with \(w_1w_2\)
(if \(\TT{e}\) is ignored).
As a special case, a constant word \(\TT{a}\,\TT{e}\) can be transformed to
\(\TT{br}\;\TT{a}\;\TT{e}\).
As a little more complex example, consider the term \(A\,(\TT{b}\,\TT{e})\), where
\(A\) is defined by \(A\,x\to \TT{a}\,x\).
Since \(A\) uses the argument, 
the term \(A\,(\TT{b}\,\TT{e})\) is transformed to
\(\TT{br}\;A\;(\TT{br}\;\Tb\;\Te)\).
Since \(A\) no longer takes an argument, we substitute \(\Te\) for \(x\)
in the body of the rule for \(A\) (and apply the transformation recursively to \(\TT{a}\,\TT{e}\)). 
The resulting rule for \(A\) is: \(A\to \TT{br}\;\TT{a}\;\TT{e}\).
Thus, the term after the transformation generates 
the tree \(\TT{br}\,(\TT{br}\;\TT{a}\,\TT{e})\,(\TT{br}\,\Tb\,\Te)\). Its frontier
word is \(\Ta\Te\Tb\Te\), which is equivalent to the word \(\Ta\Tb\) generated by the original
term, up to removals of \(\Te\); recall that redundant occurrences of \(\Te\) will be removed
by the second transformation.
Note that the transformation sketched above depends on whether each order-0 argument is
actually used or not. Thus, we introduce intersection types to express such information,
and define the transformation as a type-directed one.

Simple types are refined to the following intersection types.
\[
\begin{array}{l}
 \uty ::= \T \mid \ity \ra \uty
\qquad
 \ity ::= \uty_1\land \cdots \land \uty_k \quad (k \ge 0)
\end{array}
\]
We write \(\top\) for \(\uty_1\land \cdots \land \uty_k\) when \(k=0\).
We assume some total order \(<\) on intersection types,
and require that \(\uty_1<\cdots < \uty_k\) whenever 
\(\uty_1\land \cdots \land \uty_k\) occurs in an intersection type.
Intuitively, \((\uty_1\land\cdots \land \uty_k)\to \uty\) describes a function
that uses an argument according to types \(\uty_1,\ldots,\uty_k\), and the returns a value of
type \(\uty\). As a special case, the type \(\top\to \T\) describes a function that ignores
an argument, and returns a tree. 
Thus, according to the idea of the transformation sketched above,
if \(x\) has type \(\top\to \T\), \(x\,t\) would be transformed to \(x\);
 if \(x\) has type \(\T\to \T\), \(x\,t\) would be transformed to \(\TT{br}\;x\;t^\#\).
In the last example above, the type \(\T\to\T\) should be interpreted as a function that uses
the argument \emph{just once}; otherwise the transformation to \(\TT{br}\;x\;t^\#\) would be incorrect. 
Thus, the type \(\T\) should be treated as a linear type, 
for which weakening and dereliction are disallowed.
In contrast, we need not enforce, for example, that a value of the intersection type \(\T\to \T\) 
should be used just once.
Therefore, we classify intersection types into two kinds; one called \emph{balanced},
which may be treated as non-linear types, and the other called \emph{unbalanced},
which must be treated as linear types.
For that purpose,
we introduce two refinement relations \(\brr{\uty}{\sty}\) and
\(\urr{\uty}{\sty}\); the former means that \(\uty\) is a balanced intersection type of sort \(\sty\),
and the latter means that \(\uty\) is an unbalanced intersection type of sort \(\sty\).
The relations are defined as follows, by mutual induction; \(k\) may be \(0\).\\[1ex]
\InfruleS{0.55}{\urr{\uty_j}{\sty} \andalso j\in\set{1,\ldots,k} \\
\brr{\uty_i}{\sty} \mbox{ (for each $i\in\set{1,\ldots,k}\setminus\set{j}$)}
}{\urr{\uty_1\land \cdots \land \uty_k}{\sty}}
\InfruleS{0.45}{
\brr{\uty_i}{\sty} \mbox{ (for each $i\in\set{1,\ldots,k}$)}
}{\brr{\uty_1\land \cdots \land \uty_k}{\sty}}\\[1ex]
\InfruleS{0.2}{}{\urr{\T}{\T}}
\InfruleS{0.25}{
\brr{\ity}{\sty} \andalso \urr{\uty}{\sty'}
}{
\urr{\ity \ra \uty}{\sty \ra \sty'}
}
\InfruleS{0.25}{
\urr{\ity}{\sty} \andalso \urr{\uty}{\sty'}
}{
\brr{\ity \ra \uty}{\sty \ra \sty'}
}
\InfruleS{0.25}{
\brr{\ity}{\sty} \andalso \brr{\uty}{\sty'}
}{
\brr{\ity \ra \uty}{\sty \ra \sty'}
}\\
\noindent
A type \(\uty\) is called \myemph{balanced} if \(\brr{\uty}{\sty}\) for some \(\sty\),
and called \myemph{unbalanced} if \(\urr{\uty}{\sty}\) for some \(\sty\). 
Intuitively, unbalanced types describe trees or 
closures that contain the end of a word (i.e., symbol \(\Te\)).
Intersection types that are neither balanced nor unbalanced 
are considered ill-formed, and excluded out. 
For example, the type \(\T\to\T\to\T\) (as an intersection type) is ill-formed;
since \(\T\) is unbalanced, \(\T\to\T\) must also be unbalanced according to 
the rules for arrow types, but it is actually balanced.
Note that, in fact, no term can have the intersection type \(\T\to\T\to\T\) in a word grammar.
We write \(\uty\DCOL\sty\) if \(\brr{\uty}{\sty}\) or \(\urr{\uty}{\sty}\).

We introduce a type-directed transformation relation \(\UE\p t:\uty\tr u\) for terms,
where \(\UE\) is a set of type bindings of the form \(x\COL\uty\), called a \emph{type
environment}, \(t\) is a source term, and \(u\) is the image of the transformation,
which may be an extended term.
We write \(\UE_1\cup\UE_2\) for the union of \(\UE_1\) and \(\UE_2\); it is defined only
if, whenever \(x\COL\uty\in\UE_1\cap \UE_2\), \(\uty\) is balanced. In other words,
unbalanced types are treated as linear types, whereas balanced ones as non-linear (or
idempotent) types. We write \(\balanced(\TE)\) if \(\uty\) is balanced for every
\(x\COL\uty\in\TE\).

The relation \(\UE \p t:\uty\tr u\) is defined inductively by the following rules.\\[1ex]

\typicallabel{Tr1-Const}
\InfruleSR{0.4}{Tr1-Var}{\balanced(\TE)}
{\TE,x\COL\uty \p x\COL\uty\tr x_{\uty}}
\InfruleSR{0.55}{Tr1-NT}
{\uty\DCOL \NONTERMS(A)\andalso \balanced(\TE)}
{\TE\p A\COL\uty \tr A_\uty}\\[1ex]

\InfruleSR{0.4}{Tr1-Const0}
{\balanced(\TE)}
{\TE\p \Te\COL\T\tr \Te}
\InfruleSR{0.55}{Tr1-Const1}
{\TERMS(a)=1\andalso \balanced(\TE)}
{\TE \p a\COL\T\ra\T\tr a}

\infrule[Tr1-App1]
  {\TE_0 \p s\COL \uty_1\land\cdots\land \uty_k\ra\uty\tr v \\
\TE_i\p t\COL\uty_i\tr U_i \text{ and }\uty_i\neq \T
\text{ (for each $i \in \set{1,\dots,k}$)} %
}
{\TE_0\cup\TE_1\cup\cdots\cup\TE_k\p st:\uty\tr vU_1\cdots U_k}

\infrule[Tr1-App2]
{\TE_0\p s\COL \T\ra\uty\tr V\andalso
\TE_1\p t\COL\T\tr U}
{\TE_0\cup\TE_1 \p st\COL\uty\tr \TT{br}\,V\,U}

\infrule[Tr1-Set]
{\TE\p t\COL \uty \tr u_i\mbox{ (for each $i\in\set{1,\ldots,k}$)}
{\andalso k \ge 1}\\
}
{\TE \p t\COL\uty \tr \set{u_1,\ldots,u_k}}

\infrule[Tr1-Abs1]
{\TE,x\COL\uty_1,\ldots,x\COL\uty_k \p t:\uty\tr u\andalso x \notin \envdom{\TE}\\
\uty_i\neq \T\mbox{ for each $i\in\set{1,\ldots,k}$}}
{\TE\p \lambda x.t:\uty_1\land\cdots\land\uty_k\ra\uty\tr 
\lambda x_{\uty_1}\cdots\lambda x_{\uty_k}.u}

\infrule[Tr1-Abs2]
{\TE,x\COL\T\p t:\uty\tr u}
{\TE\p \lambda x.t:\T\ra\uty \tr [\Teps/x_{\T}]u}

In rule \rname{Tr1-Var}, a variable is replicated for each type.
This is because the image of the transformation of a term substituted for \(x\) 
is different depending on the type of the term; accordingly, in rule \rname{Tr1-Abs1},
bound variables are also replicated, and in rule \rname{Tr1-App1}, arguments are
replicated.
In rule \rname{Tr1-NT}, a non-terminal is also replicated for each type.
In rules \rname{Tr1-Const0} and \rname{Tr1-Const1}, constants are mapped to themselves;
however, the arities of all the constants become \(0\).
In these rules, \(\TE\) may contain only bindings on balanced types.

In rule \rname{Tr1-App1}, the first premise indicates that 
the function \(s\) uses the argument \(t\) according to types \(\uty_1,\ldots,\uty_k\).
Since the image of the transformation of \(t\) depends on its type,
we replicate the argument to \(U_1,\ldots,U_k\). For each type \(\uty_i\), the result
of the transformation is not unique (but finite); thus, we represent the image of
the transformation as a \emph{set} \(U_i\) of terms. (Recall the remark at the end of
Section~\ref{sec:pre} that a set of terms can be replaced by an ordinary term
by introducing auxiliary non-terminals.) For example, consider a term
\(A (x\,y)\). It can be transformed to \(A_{\uty_1\to\uty}\set{x_{\uty_0\to\uty_1}y_{\uty_0}, 
x_{\uty'_0\to\uty_1}y_{\uty_0'}}\) under the type environment
\(\set{x\COL \uty_0\to\uty_1, x\COL\uty_0'\to\uty_1, y\COL\uty_0,y\COL\uty_0'}\).
Note that \(k\) in rule \rname{Tr1-App1} (and also \rname{Tr1-Abs1}) may be \(0\),
in which case the argument disappears in the image of the transformation.

In rule \rname{Tr1-App2},
as explained at the beginning of this section,
the argument \(t\) of type \(\T\) is removed from \(s\)
and instead attached as a sibling node of the tree generated by (the transformation image of) \(s\).
Accordingly, in rule \rname{Tr1-Abs2}, the binder for \(x\) is 
removed and \(x\) in 
the body of the abstraction
is replaced with the empty tree \(\Te\).
In rule \rname{Tr1-Set}, type environments are shared. This is because \(\set{u_1,\ldots,u_k}\)
represents the choice \(u_1+\cdots+u_k\); unbalanced (i.e. linear) values should be used in the same manner
in \(u_1,\ldots,u_k\).

The transformation rules for rewriting rules and grammars 
are given by:

\infrule[Tr1-Rule]
{\emptyset \p \lambda x_1.\cdots\lambda x_k.t:\uty
 \tr \lambda x'_1.\cdots\lambda x'_\ell.u \andalso \uty\DCOL\NONTERMS(A)}
{(A\,x_1\,\cdots\,x_k\Hra t)\tr
 (A_{\uty}\,x'_1\,\cdots\,x'_{\ell} \Hra u)}

\infrule[Tr1-Gram]
{
{\TERMS' = \set{\TT{br}\mapsto 2,\Te\mapsto 0}\cup\set{a\mapsto 0\mid \TERMS(a)=1}}
\\
\NONTERMS' = \set{A_\uty\COL \utytosty{\uty\DCOL\sty} \mid \NONTERMS(A)=\sty \land \uty\DCOL\sty}\andalso
\RULES' = \set{ r' \mid \exists r\in \RULES. r\tr r'}
}
{(\TERMS,\NONTERMS,\RULES,S)\tr 
({\TERMS'},\NONTERMS',\RULES',S_\T)}

Here, \(\utytosty{\uty\DCOL\sty}\) is defined by:
\[
\begin{array}{l}
\utytosty{\uty\DCOL\sty} =
\T \qquad \mbox{ (if $\order(\sty) \le 1$)}
\\
\utytosty{(\uty_1\land\cdots\land\uty_k \ra \uty)\DCOL(\sty_0\to\sty)} =
\utytosty{\uty_1\DCOL\sty_0} \ra \dots \ra \utytosty{\uty_k\DCOL\sty_0} \ra \utytosty{\uty\DCOL\sty}
\\
\mspace{520mu}\mbox{ (if $\order(\sty_0\to\sty)>1$)}
\end{array}
\]

\begin{example}
\label{ex:tr1-gram1}
Recall the grammar \(\GRAM_1\) in Example~\ref{ex:grammar}.
For the term \(\lambda f.\lambda x.\Ta(f\,x)\) of the rule for \(A\), we have
the following derivation:
\[
\infers[Abs1]{\emptyset\p \lambda f.\lambda x.\Ta(f\,x):
(\T\to\T)\to\T\to\T\tr \lambda f_{\T\to\T}.\TT{br}\,\Ta\,(\TT{br}\,f_{\T\to\T}\,\Te)
}
{\infers[Abs2]{f\COL\T\to\T\p \lambda x.\Ta(f\,x):\T\to\T\tr \TT{br}\,\Ta\,(\TT{br}\,f_{\T\to\T}\,\Te)}
{\infers[App2]{f\COL\T\to\T, x\COL\T\p \Ta(f\,x):\T\tr \TT{br}\,\Ta\,(\TT{br}\,f_{\T\to\T}\,x_{\T})}
{\infers[Const1]{\emptyset\p \Ta:\T\to\T\tr \Ta}{}
 & \infers[App2]{f\COL\T\to\T, x\COL\T\p f\,x:\T\tr \TT{br}\,f_{\T\to\T}\,x_{\T}}
  {\infers[Var]{f\COL\T\to\T\p f:\T\to\T\tr f_{\T\to\T}}{} & 
   \infers[Var]{x\COL\T \p x:\T\tr x_\T}{}}}
}}
\]
Notice that the argument \(x\) has been removed, and the result of the transformation
has type \(\T\to\T\).
The whole grammar is transformed to the grammar consisting of the following rules.
\[
\begin{array}{l}
S_{\T} \to F_{(\T\to\T)\to\T}\,\Ta
\qquad
S_{\T} \to F_{(\T\to\T)\to\T}\,\Tb\\
A_{(\T\to\T)\to\T\to\T}\,f_{\T\to\T} \to \TT{br}\,\Ta\,(\TT{br}\,f_{\T\to\T}\,\Te)\qquad
B_{(\T\to\T)\to\T\to\T}\,f_{\T\to\T} \to \TT{br}\,\Tb\,(\TT{br}\,f_{\T\to\T}\,\Te)\\
F_{(\T\to\T)\to\T}\,f_{\T\to\T} \to \TT{br}\,f_{\T\to\T}\,(\TT{br}\,f_{\T\to\T}\,\Te)\qquad
F_{(\T\to\T)\to\T}\,f_{\T\to\T} \to F_{(\T\to\T)\to\T}(A_{(\T\to\T)\to\T\to\T}\,f_{\T\to\T})\\
F_{(\T\to\T)\to\T}\,f_{\T\to\T} \to F_{(\T\to\T)\to\T}(B_{(\T\to\T)\to\T\to\T}\,f_{\T\to\T}).
\end{array}
\]
Here, we have omitted rules that are unreachable from \(S_\T\).
For example, the rule
\[
F_{(\top\to\T)\land (\T\to\T)\to\T}\,f_{\top\to\T}\,f_{\T\to\T}
\to \TT{br}\,f_{\T\to\T}\,f_{\top\to\T}
\]
may be obtained from the following derivation, but it is unreachable from \(S_\T\),
since \(F\) is never called with an argument of type \((\top\to\T)\land (\T\to\T)\).
\[
\infers[Abs1]{\emptyset\p \lambda f.f(f\,\Te):
(\top\to\T)\land (\T\to\T)\to\T\tr \lambda f_{\top\to\T}.\lambda f_{\T\to\T}.
\TT{br}\,f_{\T\to\T}\,f_{\top\to\T}
}
{\infers[App2]{f\COL\top\to\T, f\COL\T\to\T\p f(f\,\Te):\T\tr \TT{br}\,f_{\T\to\T}\,f_{\top\to\T}}
{\infers[Var]{f\COL\T\to\T\p f\tr f_{\T\to\T}}{}
 & \infers[App1]{f\COL\top\to\T \p f\,\Te:\T\tr f_{\top\to\T}}{
 \infers[Var]{f\COL\top\to\T \p f:\top\to\T\tr f_{\top\to\T}}{}
}}}
\]
\end{example}

The following theorem states the correctness of the first transformation.
\iffull
A proof is given in Appendix~\ref{sec:proof-step1}.
\else
\fi
\begin{theorem}
\label{th:tr1-correctness}
Let \(\GRAM\) be an order-\((n+1)\) word grammar.
If \(\GRAM\tr \GRAM''\), then \(\GRAM''\) is an (extended) grammar {of order at most \(n\)}.
Furthermore, \(\Wlang(\GRAM) = \remeps{\LLang(\GRAM'')}\).
\end{theorem}
\section{Step 2: removing dummy symbols}
\label{sec:leaf2}
\newcommand\toleaf[1]{\tyTOsty{#1}}

We now describe the second step for eliminating redundant symbols \(\Te\),
which have been introduced by \rname{Tr1-Abs2}.
By the remark at the end of Section~\ref{sec:pre}, we assume that the result
of the first transformation is an ordinary grammar, not containing extended terms.
We also assume that \(\TT{br}\) occurs only in the fully applied form. This does not lose generality,
because otherwise we can replace \(\TT{br}\) by a new non-terminal \(\mathit{Br}\)
and add the rule \(\mathit{Br}\,x\,y\Hra \TT{br}\,x\,y\).

The idea of the transformation is to use intersection types to
distinguish between terms that generate trees consisting of only {\(\br\) and} \(\Te\), 
and those that generate trees containing other arity-0 terminals. We assign the type \(\Tempty\)
to the former terms, and \(\Tplus\) to the latter.
A term \(\TT{br}\,t_0\,t_1\) is transformed to 
(i)
\(\TT{br}\,t_0^\#\,t_1^\#\) if both \(t_0\) and \(t_1\) have type \(\Tplus\)
(where \(t_i^\#\) is the image of the transformation of \(t_i\)),
(ii) \(t_i^\#\) 
if \(t_i\) has type \(\Tplus\) and \(t_{1-i}\) has type \(\Tempty\), and 
(iii) \(\Te\) if both \(t_0\) and \(t_1\) have type \(\Tempty\).
As in the transformation of the previous section,
 we replicate each non-terminal and variable for each intersection type. For example,
the nonterminal \(A\COL\T\to\T\) defined by \(A\,x\to x\) would be replicated to
\(A_{\Tplus\to \Tplus}\) and \(A_{\Tempty\to\Tempty}\).

We first define the set of intersection types by:
\[ \tty ::= \Tempty \mid \Tplus \mid \tty_1\land\cdots\land \tty_k\ra \tty\]
We assume some total order \(<\) on intersection types, and require that
whenever we write \(\tty_1\land\cdots\land \tty_k\), \(\tty_1<\cdots< \tty_k\) holds.
We define the refinement relation \(\tty\DCOL\sty\) inductively by:
(i) \(\Tempty\DCOL \T\), (ii) \(\Tplus\DCOL\T\), and (iii) \(
(\tty_1\land\cdots\land \tty_k\ra \tty)\DCOL (\sty_1\to\sty_2)\) if
\(\tty\DCOL\sty_2\) and \(\tty_i\DCOL\sty_1\) for every \(i\in\set{1,\ldots,k}\).
We consider only types \(\tty\) such that \(\tty\DCOL\sty\) for some \(\sty\).
For example, we forbid an ill-formed type like \(\Tplus\land (\Tplus\to\Tplus) \to \Tplus\).

We introduce a type-based transformation relation \(\TTE\p t:\tty \tr u\),
where \(\TTE\) is a type environment (i.e., a set of bindings of the form \(x\COL\tty\)),
 \(t\) is a source term, \(\tty\) is the type of \(t\), and \(u\) is the result of transformation.
The relation is defined inductively by the rules below. \\[1ex]
\typicallabel{Tr2-NT}
\InfruleSR{0.3}{Tr2-Var}{} %
{\TTE,x\COL\tty \p x\COL\tty\tr x_{\tty}}
\InfruleSR{0.3}{Tr2-Const0}
{}
{\TTE\p \Teps\COL\Tempty\tr \Teps}
\InfruleSR{0.35}{Tr2-Const1}
{\TERMS(a)=0\andalso a\neq \Teps}
{\TTE\p a\COL\Tplus\tr a}\\

\infrule[Tr2-Const2]
{\TTE \p t_0\COL\tty_0\tr u_0\andalso \TTE \p t_1\COL\tty_1\tr u_1\\
  (u,\tty) = \left\{
  \begin{array}{ll}
    (\TT{br}\,u_0\,u_1, \Tplus) & \mbox{ if \(\tty_0=\tty_1=\Tplus\)}\\
    (u_i, \Tplus) & \mbox{ if \(\tty_i=\Tplus\) and \(\tty_{1-i}=\Tempty\)}\\
    (\Teps, \Tempty) & \mbox{ if \(\tty_0=\tty_1=\Tempty\)}\\
  \end{array}\right.
}
{\TTE\p \TT{br}\,t_0\,t_1\COL\tty\tr u}

\infrule[Tr2-NT]{\tty\DCOL \NONTERMS(F)\andalso A\,x_1\,\cdots\,x_k\Hra t\in \RULES
\andalso \emptyset\p \lambda x_1.\cdots\lambda x_k.t\COL\tty
\tr \lambda y_1.\cdots\lambda y_\ell.u}
{\TTE\p A\COL\tty\tr A_\tty}

\infrule[Tr2-App]
  {\TTE \p s\COL \tty_1\land\cdots\land \tty_k\ra\tty\tr v\andalso
  \TTE\p t\COL\tty_i\tr U_i \mbox{ (for each $i\in\set{1,\ldots,k}$)}}
{\TTE\p st:\tty\tr vU_1\cdots U_k}

\infrule[Tr2-Set]
  {\TTE\p t\COL\tty\tr u_i \mbox{ (for each $i\in\set{1,\ldots,k}$)} \andalso k\ge 1}
{\TTE\p t:\tty\tr \set{u_1,\ldots,u_k}}
\infrule[Tr2-Abs]
{\TTE,x\COL\tty_1,\ldots,x\COL\tty_k \p t:\tty\tr u}
{\TTE\p \lambda x.t:\tty_1\land\cdots\land\tty_k\ra\tty\tr 
\lambda x_{\tty_1}\cdots\lambda x_{\tty_k}.u}

The transformation of rewriting rules and grammars is defined by:
\infrule[Tr2-Rule]
{\emptyset \p \lambda x_1.\cdots\lambda x_k.t:\tty
 \tr \lambda x'_1.\cdots\lambda x'_\ell.t' \andalso \tty\DCOL\NONTERMS(A)}
{ (A\to \lambda x_1.\cdots\lambda x_k.t)\tr
 (A_{\tty}\to \lambda x'_1.\cdots\lambda x'_{\ell}.t')}

\infrule[Tr2-Gram]
{
\NONTERMS' = \set{A_\tty\COL \ttytosty{\tty} \mid \NONTERMS(A)=\sty \land \tty\DCOL\sty}\\
\RULES' = \set{ r' \mid \exists r\in \RULES. r\tr r'}\cup \set{S'\to S_{\Tempty}, S'\to S_{\Tplus}}
}
{(\TERMS,\NONTERMS,\RULES,S)\tr 
(\TERMS,\NONTERMS',\RULES',S')}
Here, \(\ttytosty{\tty}\) is defined by:
\[
\begin{array}{l}
\ttytosty{\Tempty}=\ttytosty{\Tplus}=\T\qquad
\ttytosty{\tty_1\land\cdots\land\tty_k\ra\tty} =
 \ttytosty{\tty_1}\ra\cdots\ra\ttytosty{\tty_k}\ra\ttytosty{\tty}
\end{array}
\]

We explain some key rules. 
In \rname{Tr2-Var}
we replicate a variable for each type,
as in the first transformation.
The rules \rname{Tr2-Const0} and \rname{Tr2-Const1} are for nullary constants,
which are mapped to themselves. We assign type \(\Tempty\) to \(\Te\) and
\(\Tplus\) to the other constants.
The rule \rname{Tr2-Const2} is for the binary tree constructor \(\TT{br}\). 
As explained above, we eliminate terms that generate empty trees (those consisting of
only \(\br\) and \(\Te\)). For example, if \(\tty_0=\Tempty\) and \(\tty_1=\Tplus\), then 
\(t_0\) may generate an empty tree; thus, the whole term is transformed to \(u_1\).

The rule \rname{Tr2-NT} replicates a terminal for each type, as in the case of variables.
{The middle and rightmost premises require that there is some body \(t\) of \(A\)
that can indeed be transformed according to type \(\tty\).}
Without this condition, for example,
\(A\) defined by the rule \(A\to A\) would be transformed to \(A_{\Tempty}\) by
 \(\emptyset\p A:\Tempty \tr A_{\Tempty}\), but \(A_{\Tempty}\) diverges and does not produce
an empty tree. That would make the rule \rname{Tr2-Const2} unsound: 
when a source term is \(\TT{br}\,A\,\Ta\), it would be transformed to \(\Ta\), but
while the original term does not generate a tree, the result of the transformation does.
In short, 
{the two premises are required}
to ensure that whenever \(\emptyset\p t:\Tempty\tr u\)
holds, \(t\) can indeed generate an empty tree.
In \rname{Tr2-App}, the argument is replicated for each type. 
Unlike in the transformation in the previous section, type environments can be shared
among the premises, since linearity does not matter here.
The other rules for terms are analogous to those in the first transformation.

In rule \rname{Tr2-Gram} for grammars, we prepare a start symbol \(S'\) 
and add the rules \(S'\to S_{\Tempty}, S'\to S_{\Tplus}\). 
We remark that the rewriting rule for
\(S_{\Tempty}\) (resp. \(S_{\Tplus}\)) is generated only if the original grammar
generates an empty (resp. non-empty) tree.
For example, in the extreme case where \(\REL = \set{S\to S}\), 
we have \(\REL'=\set{S'\to S_{\Tempty}, S'\to S_{\Tplus}}\), without any rules to rewrite
\(S_{\Tempty}\) or \(S_{\Tplus}\).

\begin{example}
\label{ex:tr2-gram1}
Let us consider the grammar \(\GRAM_3 = (\TERMS,\NONTERMS,\RULES,S)\)
where \(\NONTERMS = \set{S\COL\T, A\COL\T\to\T, B\COL\T\to\T, F\COL\T\to\T}\),
and \(\RULES\) consists of:
\[
\begin{array}{l}
S \to F\,\Ta
\qquad
S \to F\,\Tb\qquad
A\,f \to \TT{br}\,\Ta\,(\TT{br}\,f\,\Te)\qquad
B\,f \to \TT{br}\,\Tb\,(\TT{br}\,f\,\Te)\\
F\,f \to \TT{br}\,f\,(\TT{br}\,f\,\Te)\qquad
F\,f \to F(A\,f)\qquad
F\,f \to F(B\,f)
\end{array}
\]
It is the same as the grammar obtained in Example~\ref{ex:tr1-gram1}, 
except that redundant subscripts on non-terminals and variables have been removed.
The body of the rule for \(A\) is transformed as follows.
\[
\infers[Abs]{\emptyset \p \lambda f.\TT{br}\,\Ta\,(\TT{br}\,f\,\Te):{\Tplus}\to{\Tplus}\tr 
  \lambda f_{\Tplus}.\TT{br}\,\Ta\,f_{\Tplus}}
{\infers[Const2]{f\COL{\Tplus} \p \TT{br}\,\Ta\,(\TT{br}\,f\,\Te):{\Tplus}\tr \TT{br}\,\Ta\,f_{\Tplus}}
 {\infers[Const1]{f\COL{\Tplus} \p \Ta:\Tplus\tr \Ta}{} &
 \infers[Const2]{f\COL{\Tplus} \p \TT{br}\,f\,\Te:\Tplus\tr f_{\Tplus}}
 {\infers[Var]{f\COL{\Tplus} \p f:\Tplus\tr f_{\Tplus}}{} & 
  \infers[Const0]{f\COL{\Tplus} \p \Te:\Tempty\tr \Te}{}}
}
}
\]
The whole rules are transformed to:
\[
\begin{array}{l}
S' \to S_{\Tplus}\qquad S' \to S_{\Tempty}\qquad
S_{\Tplus} \to F_{\Tplus\to\Tplus}\,\Ta\qquad
S_{\Tplus} \to F_{\Tplus\to\Tplus}\,\Tb\qquad\\
A_{\Tplus\to\Tplus}\,f_{\Tplus}\to \TT{br}\;\Ta\,f_{\Tplus}\qquad
B_{\Tplus\to\Tplus}\,f_{\Tplus}\to \TT{br}\;\Tb\,f_{\Tplus}\qquad
F_{\Tplus\to\Tplus}\,f_{\Tplus}\to  \TT{br}\;f_{\Tplus}\,f_{\Tplus}\\
F_{\Tplus\to\Tplus}\,f_{\Tplus}\to  F_{\Tplus\to\Tplus}(A_{\Tplus\to\Tplus}\,f_{\Tplus})\qquad
F_{\Tplus\to\Tplus}\,f_{\Tplus}\to  F_{\Tplus\to\Tplus}(B_{\Tplus\to\Tplus}\,f_{\Tplus})
\end{array}
\]
Here, we have omitted rules on non-terminals unreachable from \(S'\).

\iffull
If the rules for \(S\) in the source grammar were replaced by:
\[
S \to F\,E\qquad E\to \Ta\qquad E\to \Tb\qquad E\to \Te,
\]
then
{\(F_{\Tempty\to\Tempty}\) and \(F_{\Tempty\land \Tplus\to\Tplus}\) would become reachable.
Hence, the following rules 
generated from \(F\,f \to \TT{br}\,f\,(\TT{br}\,f\,\Te)\) %
would also become reachable:}
\[
\begin{array}{l}
F_{\Tempty\to\Tempty}\,f_{\Tempty}\to \Te
\qquad
F_{\Tempty\land \Tplus\to\Tplus}\,f_{\Tempty}\,f_{\Tplus}\to f_{\Tplus}.
\end{array}
\]
From \(F\,f\to F\,(A\,f)\), many {reachable} rules would be generated. One of the rules is:
\[
F_{\Tempty\land\Tplus\to\Tplus}\,f_{\Tempty}\,f_{\Tplus}
\to F_{\Tplus\to\Tplus} \set{A_{\Tempty\to\Tplus}\,f_{\Tempty}, A_{\Tplus\to\Tplus}\,f_{\Tplus}},
\]
which can be replaced by the following rules without extended terms:
\[
F_{\Tempty\land\Tplus\to\Tplus}\,f_{\Tempty}\,f_{\Tplus}
\to F_{\Tplus\to\Tplus} (C\,f_{\Tempty}\,f_{\Tplus})\qquad
C\,f_1\,f_2 \to A_{\Tempty\to\Tplus}\,f_1\qquad
C\,f_1\,f_2 \to A_{\Tplus\to\Tplus}\,f_2.
\]
\fi
\end{example}
The following theorem claims the correctness of the transformation. 
The proof is given in 
\iffull
Appendix~\ref{sec:proof-step2}.
\fi
The main theorem (Theorem~\ref{th:main}) follows from
Theorems~\ref{th:tr1-correctness}, \ref{th:tr2-correctness}, and the fact that any
order-\(m\) grammar with \(m<n\) can be converted to an order-\(n\)
grammar by adding a dummy non-terminal of order \(n\).
\begin{theorem}
\label{th:tr2-correctness}
Let \(\GRAM=(\TERMS,\NONTERMS,\RULES,S)\) be an order-\(n\) tree grammar.
If \(\GRAM\tr \GRAM'\), then 
\(\GRAM'\) is a tree grammar {of order at most \(n\)}, and 
\(\remeps{\LLang(\GRAM)}=\LLangE(\GRAM')\).
\end{theorem}

\section{Applications}
\label{sec:app}

\subsection{Unsafe order-2 word languages = safe order-2 word languages}

As mentioned in Section~\ref{sec:intro}, 
many of the earlier results on higher-order grammars~\cite{DammTCS,Knapik01TLCA}
were for the subclass called \emph{safe} higher-order grammars.
In safe grammars, the (simple) types of terms are restricted to 
\emph{homogeneous types}~\cite{DammTCS} of the form
\(\kappa_1\to \cdots \to \kappa_k\to\T\), where \(\order(\kappa_1)\geq \cdots \geq \order(\kappa_k)\),
and arguments of the same order must be supplied simultaneously. 
For example, if \(A\) has type \((\T\to\T)\to(\T\to\T)\to \T\),
then the term \(f\,(A\,f\,f)\) where \(f\COL\T\to\T\) is valid,
but \(g\,(A\,f)\) where \(g\COL((\T\to\T)\to\T)\to\T, f\COL\T\to\T\) is not:
the partial application \(A\,f\) is disallowed, since \(A\) expects another order-1 argument.
\emph{Unsafe} grammars (which are just called higher-order grammars
in the present paper) are higher-order grammars without the safety restriction.

For order-2 word languages, Aehlig et al.~\cite{AehligSafety} have shown that 
safety is not a genuine restriction. Our result in the present paper provides
an alternative, short proof. 
Given an unsafe order-2 word grammar \(\GRAM\), we can obtain an 
equivalent order-1 grammar \(\GRAM'\) such that \(\Wlang(\GRAM) = \LLangE(\GRAM')\).
Note that \(\GRAM'\) is necessarily safe, since {it is order-1 and hence}
there are no partial applications.
Now, apply the backward transformation
sketched in Section~\ref{sec:pre} to obtain an order-2 word grammar \(\GRAM''\)
such that \(\Wlang(\GRAM'') = \LLangE(\GRAM')\). By the construction of the backward
transformation, \(\GRAM''\) is clearly a safe grammar: Since the type of each term
occurring in \(\GRAM'\) is \(\T\to\cdots\to\T\to\T\), the type of the corresponding term of \(\GRAM''\)
is \((\T\to\T)\to\cdots\to(\T\to\T)\to(\T\to\T)\). Since all the arguments of type \(\T\) are
applied simultaneously in \(\GRAM'\), all the arguments of type \(\T\to\T\) are also
applied simultaneously in \(\GRAM''\). Thus, for any unsafe order-2 word grammar,
there exists an equivalent safe order-2 word grammar.

\subsection{Diagonal problem}

The diagonal problem~\cite{Czerwinski14} asks, given a 
 (word or tree) language \(L\) and a set \(S\) of symbols, whether for all \(n\),
there exists \(w_n\in L\) such that \(\forall a\in S.\, |w_n|_{a}\geq n\).
Here, \(|w|_a\) denotes the number of occurrences of \(a\) in \(w\).
A decision algorithm for the diagonal problem can be used for computing 
 downward closures~\cite{Zetzsche15},
which in turn have applications to program verification.
Hague et al.~\cite{Hague16POPL} recently showed that the diagonal problem is
decidable for safe higher-order word languages, and Clemente et al.~\cite{Parys16diagonality}
extended the result for unsafe tree languages.
For the single letter case of  the  diagonal problem (where \(|S|=1\)),  we can obtain
an alternative proof as follows. First, following the approach of Hague et al.~\cite{Hague16POPL},
we can use logical reflection to reduce the single letter diagonal problem for an unsafe
order-\(n\) tree language to that for the path language of an unsafe order-\(n\) tree language.
We can then use our transformation to reduce the latter to
the single letter diagonal problem for an unsafe order-\((n-1)\) tree language.
Unfortunately, this approach does not apply to the general diagonal problem;
since the logical reflection in the first step yields 
an order-\(n\) language of ``narrow'' trees~\cite{Parys16diagonality} instead of words, 
we need to extend our translation from order-\(n\) word languages to order-(\(n-1\)) tree languages
to one from order-\(n\) narrow tree languages to order-(\(n-1\)) tree languages. Actually, that translation
is the key of Clemente et al.'s proof of the decidability of the (general) diagonal problem~\cite{Parys16diagonality}. 

\subsection{Context-sensitivity of order-3 word languages}

By using the result of this paper 
and the context-sensitivity of order-2 tree languages~\cite{KIT14FOSSACS}, 
we can prove that any order-3 word language is context-sensitive, i.e.,
the membership problem for an order-3 word language can be decided 
in non-deterministic linear space. Given an order-3 word grammar \(\GRAM\), 
we first construct a corresponding order-2 tree grammar \(\GRAM'\) in advance. 
Given a word \(w\), we can construct a tree \(\pi\) whose 
frontier word is \(w\) one by one, and check whether \(\pi\in\Lang(\GRAM')\). Since the size of \(\pi\) is
linearly bounded by the length \(|w|\) of \(w\), 
\(\pi\stackrel{?}{\in}\Lang(\GRAM')\) can be checked in space linear with respect to \(|w|\).
Thus, \(w\in\Wlang(\GRAM)\) can be decided in non-deterministic linear space
(with respect to the size of \(w\)).

\section{Related Work}
\label{sec:related}

As already mentioned in Section~\ref{sec:intro}, higher-order grammars have been
extensively studied in 1980's~\cite{DammTCS,Engelfriet91,EngelfrietHTT}, but
most of those results have been for safe grammars. In particular, 
Damm~\cite{DammTCS} has shown an analogous result for safe grammars, 
but his proof does not extend to the unsafe case. 

As also mentioned in Section~\ref{sec:intro}, intersection types have
been used in recent studies of (unsafe) higher-order grammars. In particular,
type-based transformations of grammars and \(\lambda\)-terms 
have been studied in \cite{KMS13HOSC,KIT14FOSSACS,Parys16diagonality}.
Clement et al.~\cite{Parys16diagonality}, independently from ours, %
gave a
transformation from an order-\((n+1)\) ``narrow'' tree language (which subsumes a word language
as a special case) to
an order-\(n\) tree language; this transformation preserves the number of occurrences of each symbol in each tree.
When restricted to word languages, our result is stronger in that our transformation is
guaranteed to preserve the order of symbols as well, and does not add any additional leaf 
symbols (though
they are introduced in the intermediate step); consequently, our proofs are more involved.
They use different intersection types, but the overall effect of their transformation 
seems similar
to that of our first transformation.
Thus, it may actually be the case that their transformation also preserves the order of symbols,
although they have not proved so.

\section{Conclusion}
\label{sec:conc}

We have shown that for any unsafe order-\((n+1)\) word grammar \(\GRAM\), 
there exists an unsafe order-\(n\) tree grammar \(\GRAM'\) whose frontier language
coincides with the word language \(\Wlang(\GRAM)\). 
The proof is constructive in that we provided (two-step) transformations that indeed construct
\(\GRAM'\) from \(\GRAM\). The transformations are based on a combination of
linear/non-linear intersection types, which may be interesting in its own right.
As Damm~\cite{DammTCS} suggested, we expect the result to be useful for further studies of higher-order
languages; in fact, we have discussed a few applications of the result.
 \subsection*{Acknowledgments}
We would like to thank Takeshi Tsukada for helpful discussions and thank Pawel Parys for information about
the related work~\cite{Parys16diagonality}. We would also like to thank Igor Walukiewicz for spotting an error in
our argument on the diagonal problem in an earlier version of this paper. This work was supported by 
JSPS Kakenhi 23220001 and 15H05706.
\bibliographystyle{plain}

\iffull
\newpage
\appendix
\section*{Appendix}
\section{Proof of Theorem~\ref{th:tr1-correctness}}
\label{sec:proof-step1}

We give a proof of Theorem~\ref{th:tr1-correctness}
in Section~\ref{sec:app1-basic} after preparing some basic definitions.
Lemmas for the proof are given after that.
In Section~\ref{sec:app1-basic} we give basic lemmas.
In Sections~\ref{sec:lemmaForwardDirection} and~\ref{sec:lemmaBackwardDirection},
we give main lemmas for forward and backward directions of the theorem, i.e.,
left-to-right and right-to-left simulations, respectively.
The both lemmas need one key lemma, which is given in Section~\ref{sec:keyLemma}.

Throughout this section, we often write \(\TT{br}\,u_1\,u_2\) as \(u_1 * u_2\).
For \(s\), \(t_1,\ldots,t_n\),
we write an iterated application \((\cdots(s\,t_1)\,t_2\cdots)\,t_n\) as \(\appi{s}{t_i}{i}{n}\).
We also write \([t_1/x_1,\ldots,t_k/x_k]\) as \([t_i/x_i]_{\ind{i}{k}}\).

\subsection{Proof of Theorem~\ref{th:tr1-correctness} and basic definitions and lemmas}\label{sec:app1-basic}

The extended terms can be embedded into the simply typed \lambday-calculus with non-determinism and the same
constants as the terminal symbols (but without any non-terminals);
we represent also the non-determinism in this \lambday-calculus by the set-representation 
\(\set{u_1,\dots,u_n}\) (\(n\ge 1\)).
The embedding transformation is given in the standard way: the mutual recursion allowed in a grammar
is handled by using Beki\v{c} property of \Y-combinator.
Also for this \lambday-calculus, we consider call-by-name reduction. %
We call terms in this calculus simply \emphd{\lambday-terms}, which are also ranged over by \(u\) and \(v\);
but if we use \(u\) and \(v\) without mentioning where they range, they are meant to be extended
applicative terms for a given grammar.
Through this transformation, we identify extended terms in a grammar
with the embedded \lambday-terms.

We define \emphd{\(\Teps\)-observational preorder} \(\lse\) and \emphd{\(\Teps\)-observational equivalence}
\(\se\) as follows.
First we define \(\vsim\) for  trees as the least congruence (w.r.t. the definition of 
trees) satisfying
\(\pi \vsim \Teps \ibr \pi\)
and
\(\pi_1 \ibr (\pi_2 \ibr \pi_3) \vsim (\pi_1 \ibr \pi_2) \ibr \pi_3\).
Now, for two \lambday-terms 
\[
x_1\COL\kappa_1,\dots,x_n\COL\kappa_n \p u, u' : \kappa
\]
we define \(u \lse u'\) if, for any \lambday-term \(C: (\kappa_1\ra\cdots\ra\kappa_n \ra\kappa) \ra \T\)
and for any tree \(\pi\) such that
\(C(\lambda x_1.\cdots\lambda x_n.u) \reds \pi\), there exists \(\pi'\) such that
\(C(\lambda x_1.\cdots\lambda x_n.u') \reds \pi'\) and \(\pi \vsim \pi'\).
And we define \(u \se u'\) if
\(u \lse u'\) and \(u \gse u'\).

We define the set \(\FV(u)\) of \emph{free variables} of an extended term \(u\) as follows:
\begin{align*}
\FV(x) &\defe \set{x}
\\
\FV(a) &\defe \emptyset
\\
\FV(A) &\defe \emptyset
\\
\FV(\app{u}{U}) &\defe \FV(u) \cup \FV(U)
\\
\FV(\set{u_1,\dots,u_k}) &\defe \cup_{i\le k} \FV(u_i)
\end{align*}
For a word \(a_1\cdots a_n\),
we define term \(\wtt{(a_1\cdots a_n)}\) inductively by:
\(\wtt{\epsilon} = \Te\) and \(\wtt{(as)} = \TT{br}\,a\,\wtt{s}\).%

We write \(\TE \pp t : \uty \tr u\) if the judgement is derived
by using the following restricted rule instead of \rname{Tr1-Set}.
\infrule[Tr1-SetS]
{\TE\p t\COL \uty \tr u_i\mbox{ (for each $i\in\set{1,\ldots,k}$)}
{\andalso k \ge 1}\\
k=1\mbox{ if $\uty$ is unbalanced}
}
{\TE \p t\COL\uty \tr \set{u_1,\ldots,u_k}}
Clearly, if \(\TE \pp t : \uty \tr u\) then \(\TE \p t : \uty \tr u\).
We use this restriction in the proof of the forward direction of the theorem.

Now we prove Theorem~\ref{th:tr1-correctness}, whose statement is:
Let \(\GRAM\) be an order-\((n+1)\) word grammar.
If \(\GRAM\tr \GRAM''\), then \(\GRAM''\) is an (extended) grammar of order at most \(n\).
Furthermore, \(\Wlang(\GRAM) = \remeps{\LLang(\GRAM'')}\).
\begin{proof}[Proof of Theorem~\ref{th:tr1-correctness}]
The well-typedness of the right hand side term of every rewriting rule of \(\GRAM''\) can be proved straightforwardly
(in a way similar to Lemma~\ref{lem:UTr-typing} in Section~\ref{sec:keyLemma}). %
By induction on \(\sty\), we can show that
\(\order(\utytosty{\uty\DCOL\sty})\le\order(\sty)-1\) if \(\order(\sty)\ge1\) and
\(\order(\utytosty{\uty\DCOL\sty})=\order(\sty)=0\) otherwise.

Now we show \(\Wlang(\GRAM) = \remeps{\LLang(\GRAM'')}\).
Suppose \(a_1\cdots a_n\in \Wlang(\GRAM)\), i.e.,
\(S \redswith{\GRAM} a_1(\cdots (a_n\Te)\cdots)\).
By Lemma~\ref{lem:leaf-fwd}, we have 
\(\pp a_1(\cdots (a_n\Te)\cdots)\COL\T\tr
 \wtt{(a_1\cdots a_n)}\).
By Lemma~\ref{lem:red-fwd}, we have \(u\) such 
that \( \pp S:\T\tr u\) with \(u(\redwith{\GRAM''}\gse)^*  \wtt{(a_1\cdots a_n)}\).
By the transformation rule, \(u\) must be \(S_\T\). Thus,
we have \(S_\T (\redwith{\GRAM''}\gse)^* \wtt{(a_1\cdots a_n)}\),
which implies \(a_1\cdots a_n\in \remeps{\LLang(\GRAM'')}\) as required.

Conversely, suppose \(a_1\cdots a_n\in \remeps{\LLang(\GRAM'')}\),
i.e.,
\(S_\T\redswith{\GRAM''} \pi\) with \(\remeps{\leaves(\pi)} = a_1\cdots a_n\) for some \(\pi\).
By repeating Lemma~\ref{lem:red-bwd}, we have
\(S\redswith{\GRAM}s\) and \(\p s:\T \tr \pi'\) with \(\pi'\vsim \pi\).
By Lemma~\ref{lem:leaf-bwd}, \(s=a_1(\cdots (a_n\Te)\cdots)\).
Thus, we have \(a_1\cdots a_n\in \Wlang(\GRAM)\) as required.
\end{proof}

\begin{lemma}
\label{lem:leaf-fwd}
\( \pp a_1(\cdots(a_n\,\Te)\cdots):\T \tr \wtt{(a_1\cdots a_n)}\).
\end{lemma}
\begin{proof}
This follows by straightforward induction on \(n\).
\end{proof}

\begin{lemma}
\label{lem:leaf-bwd}
Let \(t\) be an applicative term.
If \( \p t:\T\tr \pi\) %
then \(t =a_1(\cdots(a_n\,\Te)\cdots)\) with \(\wtt{(a_1\cdots a_n)} = \pi\).
\end{lemma}
\begin{proof}
This follows by induction on the structure of \(\pi\). 
\begin{itemize}
\item Case \(\pi=\Te\): 
\( \p t:\T\tr \pi\) must have been derived by using \rname{Tr1-Const0}. Therefore \(t=\Te\)
as required.
\item Case \(\pi=\TT{br}\,\pi_1\,\pi_2\):
\( \p t:\T\tr \pi\) must have been derived by using \rname{Tr1-App2}.
Thus, we have:
\[\begin{array}{l}
t = t_1t_2\qquad
  \p t_1\COL\T\ra\T \tr \pi_1
\qquad 
  \p t_2\COL\T \tr \pi_2\qquad \pi = \TT{br}\,\pi_1\,\pi_2
\end{array}
\]
By the condition
\( \p t_1\COL\T\ra\T \tr \pi_1\),
 the head symbol of \(t_1\) must be a terminal.
(Because the type environment is empty, the head cannot be a variable, and because the
output of transformation does not contain a non-terminal, the head cannot be a non-terminal.)
Thus, \(t_1\) is actually a terminal \(a_1\). By the induction hypothesis and
\(  \p t_2\COL\T \tr \pi_2\), we have 
\(t_2=a_2(\cdots(a_n\,\Te)\cdots)\) with \(\wtt{(a_2\cdots a_n)}=\pi_2\). 
Thus, we have 
\(t =a_1(a_2(\cdots(a_n\,\Te)))\), with \(\wtt{(a_1a_2\cdots a_n)} = \pi\) as required. 
\end{itemize}
\end{proof}

\begin{lemma}[Context Lemma]\label{lem:context}
Given two \lambday-terms
\((x_1\COL\kappa_1,\dots,x_n\COL\kappa_n \p u, u' : \kappa)\)
where
\(\kappa= \kappa_{n+1} \ra \cdots \ra \kappa_\ell \ra \T\),
we have \(u \lse u'\) iff
for any closed terms \(U_1,\dots,U_\ell\)
of type \(\kappa_1 , \dots , \kappa_\ell\), respectively,
and for any \(\pi\) such that
\(\appi{(
\lambda x_1.\ldots\lambda x_n.
u
)}{U_i}{i}{\ell} \reds \pi\),
there exists \(\pi'\) such that
\(\appi{(
\lambda x_1.\ldots\lambda x_n.
u'
)}{U_i}{i}{\ell} \reds \pi'\)
and
\(\pi \vsim \pi'\).
(We write \(u \las u'\) if the latter condition of this equivalence holds.)
\end{lemma}
\begin{proof}
The proof is obtained by a trivial modification of the proof of the context lemma for PCF by a logical relation
 given in~\cite{DBLP:books/daglib/0018087}.

The logical relation is between a cpo model and the syntax.
The cpo model is the standard (call-by-name) cpo model extended with Hoare powerdomain, which corresponds to
 may convergence.
Specifically, the interpretation \(\ip{\T}\) of the base type \(\T\) is defined as
\((P(\VT),\subseteq)\) where \(\VT\) is the quotient set
of the set of  trees modulo \(\vsim\), and \(P(\VT)\) is the powerset of \(\VT\).
(This is the Hoare powerdomain of the flat cpo \(\VT_{\bot}\).) The interpretation of function types
is given by the usual continuous function spaces.
The interpretation of the constants is given as follows:
\begin{align*}
\ip{\TT{br}}(L_1,L_2) &\defe \set{\coset{\TT{br}\,\pi_1\,\pi_2}\,|\,\coset{\pi_i} \in L_i}
\qquad(L_1,L_2 \in P(\VT))
\\
\ip{a} &\defe \set{\coset{a}}
\qquad(\TERMS(a)=0).
\end{align*}

Now the logical relation \(R=(R_\kappa)_\kappa\) is defined as below.
Let \(\cterm{\kappa}\) be the set of closed \lambday-terms of sort \(\kappa\).
Then \(R_\kappa \subseteq \ip{\kappa} \times \cterm{\kappa}\) is defined inductively as follows:
\begin{align*}
L\,R_{\T}\,u &\quad\text{if}\quad
\text{for any \(d\in L\) there exists \(\pi\) such that \(u \reds \pi\) and \(d=\coset{\pi}\)}
\\
f\,R_{\kappa \ra \kappa'}\,u &\quad\text{if}\quad
\text{for any \(g \in \ip{\kappa}\) and \(v \in \cterm{\kappa}\),\,
\(g\,R_{\kappa}\,v\) implies \(f(g)\,R_{\kappa'}\,(\app{u}{v})\).}
\end{align*}

For \(u,u'\in\cterm{\kappa}\),
we can show that 
\[
u \lse u'
\quad\Longrightarrow\quad
u \las u'
\quad\Longrightarrow\quad
\ip{u}\,R_{\kappa}\,u'
\quad\Longrightarrow\quad
u \lse u'
\]
whose proof is obtained in the same way as that of~\cite[Theorem~5.1]{DBLP:books/daglib/0018087}.
\end{proof}
\begin{lemma}\label{lem:envWellFormed}
Given \(\TE, x\COL\uty_1,\dots,x\COL\uty_k \p t : \uty \tr u\) where \(x \notin \dom(\TE)\),\,
\(\uty_1\land\dots\land\uty_k \ra \uty\) is well-formed.
\end{lemma}
\begin{proof}
By straightforward induction on \(t\).
\end{proof}

\begin{lemma}\label{lem:envu}
Given \(\TE \p t : \uty \tr u\) %
and \(y \in \FV(u)\)
there exists \(x\COL\uty' \in \TE\) such that \(y=x_{\uty'}\).
\end{lemma}
\begin{proof}
By straightforward induction on \(t\).
\end{proof}

\begin{lemma}\label{lem:precong}
\begin{enumerate}
\item
For any \(u\), \(u_1\), \(u_2\), and \(u_3\),
\[
u \lse \Teps \ibr u
\qquad \text{and} \qquad
u_1 \ibr (u_2 \ibr u_3) \se (u_1 \ibr u_2) \ibr u_3\,.
\]
\item
For any \(\pi_1\) and \(\pi_2\),
\[
\pi_1 \se \pi_2
\qquad\text{iff}\qquad
\pi_1 \vsim \pi_2\,.
\]
\end{enumerate}
\end{lemma}
\begin{proof}
The both items can be easily shown by using the context lemma.
\end{proof}

\begin{lemma}\label{lem:subst-monotone}
If \(u \lse u'\),
then \(\mvsub u \lse \mvsub u'\).
\end{lemma}
\begin{proof}
The proof is trivial from the definition of the contextual preorder \(\lse\).
\end{proof}

\begin{lemma}\label{lem:subst-iter}
\end{lemma}
If \(\dom(\mvsub) \cap \dom(\mvsub') = \emptyset\), then \(\mvsub (\mvsub' u) = (\mvsub \cup \mvsub') u\).
\begin{proof}
The proof is given by straightforward induction on \(u\).
\end{proof}

\begin{lemma}\label{lem:strengthening}
Given \(\TE,x\COL\uty' \pp t : \uty \tr u\),
if \(x \notin \FV(t)\),
then we also have \(\TE \pp t : \uty \tr u\)
and \(\uty'\) is balanced.
\end{lemma}
\begin{proof}
This follows by straightforward induction on \(\TE,x\COL\uty' \pp t : \uty \tr u\).
\end{proof}
\subsection{Key lemma}
\label{sec:keyLemma}

\begin{lemma}\label{lem:groundsbst}
Given \(x\COL\T \p s  :  \uty \tr v\)
where \(\odr{\uty}\le1\)
and \(\p t : \T \tr U\),
\[
([\Teps/x_{\T}]v) \ibr U \se
[U/x_{\T}]v\,.
\]
Moreover,
for any \(p \ge 0\), \(\pi\), and a reduction sequence
\[
([\Teps/x_{\T}]v) \ibr U \red^p \pi
\]
there exists \(\pi'\) such that
\[
[U/x_{\T}]v \red^p \pi'
 \vsim \pi\,.
\]
\end{lemma}

The above lemma is the key of the proof of Theorem~\ref{th:tr1-correctness},
and says that the variable \(x_{\T}\) occurs at the rightmost position in (the trees of) \(v\).
For the proof of this lemma, we introduce 
a type system for the transformed grammar \(\GRAM''\).
The set of types is given by the following grammar.
\[
\oty ::= \ot \mid \rt \mid \oty \ra \oty
\]
Intuitively,
\(\rt\) is the type of trees that can occur only at the rightmost position of a tree
while \(\ot\) is the type of trees without any such restriction;
for example, if \(t\) has type \(\rt\) and \(t'\) has type \(\ot\), 
then \(t' * t\) is valid but \(t * t'\) is not.

We define a notion of balance/unbalance, which is similar to that for the types \(\uty\):
\begin{gather*}
\frac{}{\ot\text{ is balanced}}
\qquad
\frac{}{\rt\text{ is unbalanced}}
\qquad
\frac{\oty\text{ is balanced} \quad \oty'\text{ is balanced}}{\oty\ra\oty'\text{ is balanced}}
\\[10pt]
\frac{\oty\text{ is unbalanced} \quad \oty'\text{ is unbalanced}}{\oty\ra\oty'\text{ is balanced}}
\qquad
\frac{\oty\text{ is balanced} \quad \oty'\text{ is unbalanced}}{\oty\ra\oty'\text{ is unbalanced}}
\end{gather*}
A type \(\oty\) is \emph{well-formed} if it is either balanced or unbalanced.
We assume that all the types occurring below are well-formed.

A type environment \(\RE\) is a set of type bindings of the form \(x\COL\oty\).
We write \(\balanced(\RE)\) and say \(\RE\) is \emph{balanced} if \(\oty\) is balanced for every
\(x\COL\oty\in\RE\).
As before, we treat unbalanced types as linear types, i.e.,
the union \(\RE_1\cup\RE_2\) of \(\RE_1\) and \(\RE_2\) is defined only
if \(\balanced(\RE_1\cup\RE_2)\).

We define three type transformations \(\rf{-}\), \(\fg{-}\), and \(\st{-}\) as follows:
\begin{gather*}
\begin{aligned}
\rf{\uty} &\defe \rt
&&(\odr{\uty} \le 1, \text{ \(\uty\) is unbalanced})
\\
\rf{\uty} &\defe \ot
&&(\odr{\uty} \le 1, \text{ \(\uty\) is balanced})
\\
\rf{\land_{i\le k} \uty_i \ra \uty} &\defe
\rf{\uty_1} \ra \dots \ra \rf{\uty_k} \ra \rf{\uty}
&&(\odr{\land_{i\le k} \uty_i \ra \uty} \ge 2)
\\
\rf{x_1\COL\uty_1,\dots,x_n\COL\uty_n}
&\defe
\big((x_1)_{\uty_1}\COL\rf{\uty_1}
,\dots,
(x_n)_{\uty_n}\COL\rf{\uty_n}\big)
\\
\fg{\ot} &\defe \ot
\\
\fg{\rt} &\defe \ot
\\
\fg{\oty \ra \oty'} &\defe \fg{\oty} \ra \fg{\oty'}
\\
\fg{x_1\COL\oty_1,\dots,x_n\COL\oty_n}
&\defe
\big(x_1\COL\fg{\oty_1}
,\dots,
x_n\COL\fg{\oty_n}\big)
\\
\st{\uty} &\defe \fg{\rf{\uty}}
\\
\st{\TE} &\defe \fg{\rf{\TE}}
\end{aligned}
\end{gather*}
It is obvious that, if \(\uty\) is balanced (resp. unbalanced),
then \(\rf{\uty}\) is balanced (resp. unbalanced).

Then the typing rules are given as follows:
\infrule[\rtVar]{\balanced(\RE)
}{
\RE,x\COL \oty \p x:\oty
}

\infrule[\rtAlph]{\balanced(\RE)\andalso
\TERMS(a)=1 \text{ in } \GRAM
}{
\RE\p a : \ot
}

\noindent
\InfruleSR{0.49}{\rtBrO}{\balanced(\RE)}{
\RE\p \TT{br} : \ot \ra \ot \ra \ot
}\hspace{-.04555\textwidth}
\InfruleSR{0.55}{\rtBrR}{\balanced(\RE)}{
\RE\p \TT{br} : \ot \ra \rt \ra \rt
}
\\[2ex]

\noindent
\InfruleSR{0.49}{\rtEpsO}{\balanced(\RE)}{
\RE\p \eps : \ot
}\hspace{-.04555\textwidth}
\InfruleSR{0.55}{\rtEpsR}{\balanced(\RE)}{
\RE\p \eps : \rt
}
\\[2ex]

\noindent
\InfruleSR{0.49}{\rtNtO}{\balanced(\RE)
}{
\RE\p A_{\uty} : \st{\uty}
}\hspace{-.04555\textwidth}
\InfruleSR{0.55}{\rtNtR}{\balanced(\RE)
}{
\RE\p A_{\uty} : \rf{\uty}
}

\infrule[\rtApp]{
\RE_0 
\p v : \oty_1 \ra \oty
\andalso
\RE_1
\p U : \oty_1
}{
\RE_0 \cup \RE_1
\p \app{v}{U} : \oty
}
\infrule[\rtSet]{
\RE
\p u_i : \oty
\quad\text{(for each \(\ind{i}{k}\))}
}{
\RE
\p \set{u_1,\dots,u_k} : \oty
}
\infrule[\rtAbs]{
\RE, x\COL \oty' \p u : \oty
}{
\RE \p \lambda x.u : \oty' \ra \oty
}

We prepare some lemmas for proving Lemma~\ref{lem:groundsbst}.

\begin{lemma}\label{lem:envWellFormed2}
If \(\RE,x\COL\oty \p u :\oty'\), then \(\oty\ra\oty'\) is well-formed.
\end{lemma}
\begin{proof}
This follows by straightforward induction on the derivation \(\RE,x\COL\oty \p u :\oty'\).
\end{proof}

\begin{lemma}[substitution]\label{lem:rty-subst}
Given
\ \(\RE,x'\COL\oty'\p v : \oty\)\ 
and
\ \(\RE' \p U : \oty'\),
we have \(\RE \cup \RE' \p [U/x']v:\oty\).
\end{lemma}
\begin{proof}
The proof is given by induction on \(v\).
The base case is clear.
The remaining case is application: we have rule \rname{\rtApp}
\[
\InfruleS{}{
\RE_0 
\p v' : \oty_1 \ra \oty
\andalso
\RE_1
\p U' : \oty_1
}{
\RE_0 \cup \RE_1
\p \app{v'}{U'} : \oty
}
\]
where
\[
\RE,x'\COL\oty' = \RE_0 \cup \RE_1
\qquad
v = \app{v'}{U'}\,.
\]
Further we have \rname{\rtSet}
\[
\InfruleS{}{
\RE_1
\p u'_i : \oty_1
\quad\text{(for each \(i \in \set{1,\ldots,k}\))}
}{
\RE_1
\p \set{u'_1,\dots,u'_k} : \oty_1
}
\]
where \(U' = \set{u'_1,\dots,u'_k}\).

Now we perform a case analysis on whether \(\oty'\) is balanced or unbalanced.
\begin{itemize}
\item
Case where \(\oty'\) is balanced:
In this case, %
{\(\RE'\)} %
 is balanced.
By the induction hypotheses, we have
\[
(\mn{\RE_0}{\set{x'\COL\oty'}}) \cup \RE'
\p [U/x']v' : \oty_1 \ra \oty
\qquad
(\mn{\RE_1}{\set{x'\COL\oty'}}) \cup \RE'
\p [U/x']u'_i : \oty_1
\quad\text{(for each \(\ind{i}{k}\))}\,.
\]
and by \rname{\rtSet},
\begin{align*}&
\infer{
(\mn{\RE_0}{\set{x'\COL\oty'}}) \cup \RE'
\p \set{v'_1,\dots,v'_{k_0}} : \oty_1 \ra \oty
}{
(\mn{\RE_0}{\set{x'\COL\oty'}}) \cup \RE'
\p v'_j : \oty_1 \ra \oty
\quad\text{(for each \(j \in \set{1,\ldots,k_0}\))}
}
&
\set{v'_1,\dots,v'_{k_0}}
=&\ 
[U/x']v' 
\\[10pt]&
\infer{
(\mn{\RE_1}{\set{x'\COL\oty'}}) \cup \RE'
\p \set{u'^{i}_{1},\ldots,u'^{i}_{k_i}} : \oty_1
}{
(\mn{\RE_1}{\set{x'\COL\oty'}}) \cup \RE'
\p u'^{i}_{j} : \oty_1
\quad\text{(for each \(j \in \set{1,\ldots,k_{i}}\))}
}
&
\set{u'^{i}_{1},\ldots,u'^{i}_{k_i}}
=&\ 
[U/x']u'_i 
\end{align*}
Then, by \rname{\rtSet}
\[
(\mn{\RE_1}{\set{x'\COL\oty'}}) \cup \RE'
\p [U/x']U' : \oty_1
\]
and by \rname{\rtSet} and \rname{\rtApp}, we have
\begin{gather*}
(\mn{\RE_0}{\set{x'\COL\oty'}}) \cup \RE'
\cup (\mn{\RE_1}{\set{x'\COL\oty'}}) \cup \RE'
\p ([U/x']v')([U/x']U') : \oty
\end{gather*}
where the linearity condition is obvious, since \(\RE'\) is balanced and
\[
(\mn{\RE_0}{\set{x'\COL\oty'}}) \cap (\mn{\RE_1}{\set{x'\COL\oty'}})
\subseteq
\RE_0 \cap \RE_1 \subseteq (\text{the set of balanced bindings})\,.
\]
\item
Case where \(\oty'\) is unbalanced and \(x'\COL\oty' \in \RE_1\):
By the induction hypotheses, we have
\[
(\mn{\RE_1}{\set{x'\COL\oty'}}) \cup \RE'
\p [U/x']u'_i : \oty_1
\quad\text{(for each \(\ind{i}{k}\))}
\]
and by \rname{\rtSet}, similarly to the previous case, we have
\[
(\mn{\RE_1}{\set{x'\COL\oty'}}) \cup \RE'
\p [U/x']U' : \oty_1\,.
\]
Then by \rname{\rtApp}, we have
\begin{gather*}
\RE_0
\cup (\mn{\RE_1}{\set{x'\COL\oty'}}) \cup \RE'
\p v'([U/x']U') : \oty
\end{gather*}
as required; here the linearity condition holds as follows: Since
\(\oty'\) is unbalanced, \(\RE\) is balanced.
Now \(x'\COL\oty' \in \RE_1\), and therefore \(\RE_0\) and \(\mn{\RE_1}{\set{x'\COL\oty'}}\) are balanced.
\item
Case where \(\oty'\) is unbalanced and \(x'\COL\oty' \in \RE_0\):
By the induction hypothesis, we have
\[
(\mn{\RE_0}{\set{x'\COL\oty'}}) \cup \RE'
\p [U/x']v' : \oty_1 \to \oty\,.
\]
Then by \rname{\rtApp}, we have
\begin{gather*}
(\mn{\RE_0}{\set{x'\COL\oty'}}) \cup \RE'
\cup \RE_1
\p ([U/x']v')U' : \oty
\end{gather*}
as required; here the linearity condition holds since
\(\mn{\RE_0}{\set{x'\COL\oty'}}\) and \(\RE_1\) are balanced (similarly to the previous case).
\end{itemize}
\end{proof}

\begin{lemma}\label{lem:UTr-typing}
For any \(\TE \p s : \uty \tr v\),
we have \(\rf{\TE} \p v : \rf{\uty}\).
\end{lemma}
\begin{proof}
The proof proceeds by straightforward induction on the derivation \(\TE \p s : \uty \tr v\).
Note that, since if \(\uty\) is balanced so is \(\rf{\uty}\),
\(\balanced(\TE)\) implies \(\balanced(\rf{\TE})\).

\begin{myitemize}
\myitem
Case of \rname{Tr1-Var}:
\[
\InfruleS{}{\balanced(\TE)}
{\TE,x\COL\uty \p x\COL\uty\tr x_{\uty}}
\]
The goal:
\[
\rf{\TE},x_{\uty }\COL\rf{\uty }\p  x_{\uty} : \rf{\uty}
\]
is obtained by \rname{\rtVar}.

\myitem
Case of \rname{Tr1-Const0}:
\[
\InfruleS{}
{\balanced(\TE)}
{\TE\p \Te\COL\T\tr \Te}
\]
The goal:
\[
\rf{\TE}\p  \Te : \rt
\]
is obtained by \rname{\rtEpsR}.

\myitem
Case of \rname{Tr1-Const1}:
\[
\InfruleS{}
{\balanced(\TE)\andalso\TERMS(a)=1}
{\TE\p a\COL\T\ra\T\tr a}
\]
The goal:
\[
\rf{\TE}\p  a : \ot
\]
is obtained by \rname{\rtAlph}.

\myitem
Case of \rname{Tr1-NT}:
\[
\InfruleS{}
{\balanced(\TE)}
{\TE\p A\COL\uty \tr A_\uty}
\]
The goal:
\[
\rf{\TE}\p  A_\uty : \rf{\uty }
\]
is obtained by \rname{\rtNtR}.

\myitem
Case of \rname{Tr1-App1}:
\[
\InfruleS{}
  {\TE_0 \p s\COL \uty_1\land\cdots\land \uty_k\ra\uty\tr v \\
\TE_i\p t\COL\uty_i\tr U_i \text{ and }\uty_i\neq \T
\text{ (for each $i \in \set{1,\dots,k}$)} 
}
{\TE_0\cup\TE_1\cup\cdots\cup\TE_k\p st:\uty\tr vU_1\cdots U_k}
\]
The induction hypotheses are
\begin{align*}&
\rf{\TE_0} \p  v : \rf{\uty_1}\ra\cdots\ra \rf{\uty_k}\ra\rf{\uty}
\\&
\rf{\TE_i} \p U_i : \rf{\uty_i} \quad (\text{for each $i \in \set{1,\dots,k}$})
\end{align*}
where the latter are obtained through \rname{Tr1-Set} and \rname{\rtSet}.
The goal:
\[
\rf{\TE_0}\cup\rf{\TE_1}\cup\cdots\cup\rf{\TE_k}\p vU_1\cdots U_k:\rf{\uty}
\]
is obtained by \rname{\rtApp}.

\myitem
Case of \rname{Tr1-App2}:
\[
\InfruleS{}{
\TE_0\p s\COL \T\ra\uty\tr V
\andalso
\TE_1\p t\COL\T\tr U}
{\TE_0\cup\TE_1 \p st\COL\uty\tr \TT{br}\,V\,U}
\]
By the well-formedness, \(\uty\) is unbalanced.
Hence, the induction hypotheses are
\begin{align*}&
\rf{\TE_0}\p  V : \ot
\\&
\rf{\TE_1}\p  U : \rt\,.
\end{align*}
The goal:
\[
\rf{\TE_0}\cup\rf{\TE_1} \p  \TT{br}\,V\,U : \rt
\]
is obtained by \rname{\rtApp} and \rname{\rtBrR}.

\myitem
Case of \rname{Tr1-Set}:
\[
\InfruleS{}
{\TE\p t\COL \uty \tr u_i\mbox{ (for each $i\in\set{1,\ldots,k}$)}
\andalso k \ge 1
}
{\TE \p t\COL\uty \tr \set{u_1,\ldots,u_k}}
\]
The induction hypotheses are
\[
\rf{\TE}\p  u_i : \rf{ \uty } \qquad(i\in\set{1,\ldots,k})\,.
\]
The goal:
\[
\rf{\TE} \p  \set{u_1,\ldots,u_k} : \rf{\uty }
\]
is obtained by \rname{\rtSet}.
\end{myitemize}
\end{proof}

\begin{lemma}\label{lem:refine-forget}
\begin{enumerate}
\item\label{item:forget}
Given \(\RE\p u:\oty\),
we have \(\fg{\RE}\p u:\fg{\oty}\).
\item\label{item:forget-typing}
Given \(\TE\p t:\uty \tr u\),
we have \(\st{\TE}\p u:\st{\uty}\).
\item\label{item:forget-subst}
Given \(\TE, x\COL\T\p t:\uty \tr u\),
we have \(\st{\TE}\p [\Teps/x_\T]u:\st{\uty}\).
\end{enumerate}
\end{lemma}
\begin{proof}
Case of item~\ref{item:forget}:
The proof is given by straightforward induction on \(u\);
the base case is trivial, since for every terminal and non-terminal there is a typing rule for having a type of the form
 \(\st{\uty}\).
In the case of application, we have rule \rname{\rtApp}:
\[
\InfruleS{}{
\RE_0 
\p v : \oty_1 \ra \oty
\andalso
\RE_1
\p U : \oty_1
}{
\RE_0 \cup \RE_1
\p \app{v}{U} : \oty
}
\]
This case is also clear by the induction hypotheses.

Case of item~\ref{item:forget-typing}:
By Lemma~\ref{lem:UTr-typing},
\(\rf{\TE}\p u: \rf{\uty}\).
By item~\ref{item:forget},
we have
\(\st{\TE}\p u: \st{\uty}\).

Case of item~\ref{item:forget-subst}:
By item~\ref{item:forget-typing},
we have
\(\st{\TE}, x_{\T}\COL\ot\p u: \st{\uty}\).
Since \(\p\Teps:\ot\), by Lemma~\ref{lem:rty-subst},
we have \(\st{\TE}\p [\Teps/x_{\T}]u:\st{\uty}\).
\end{proof}

\begin{lemma}[subject reduction]\label{lem:rty-sub-red}
Given a reduction \(u \red u'\),
\begin{enumerate}
\item\label{item:subred-unb}
if \(x_\T\COL\rt \p u : \rt\) then \(x_\T\COL\rt \p u' : \rt\), and 
\item\label{item:subred-bal}
if \(\p u : \ot\) then \(\p u' : \ot\).
\end{enumerate}
\end{lemma}
\begin{proof}
The proof is given by induction on \(u\) simultaneously for the both items.
Since \(u \red u'\), the head of \(u\) is either \(\TT{br}\) or a non-terminal.

Case where the head of \(u\) is \(\TT{br}\):
Let \(u=\TT{br}\,U_1\,U_2\).
When \(U_1\) is reduced, the case that \(U_1\) is not a singleton is clear,
since in the rule \rname{\rtSet}, the type parts and the environment parts of judgments are common.
Suppose \(U_1 = \{u_1\}\) and \(u_1 \red u'_1\) and \(u'=\TT{br}\,u'_1\,U_2\).
First we consider item~\ref{item:subred-unb}.
For \((x_\T\COL\rt \p \TT{br}\,u_1\,U_2 : \rt)\),
\rname{\rtApp} and \rname{\rtBrR} are used,
i.e.,
\(\p \TT{br} : \ot \ra \rt \ra \rt\).
In the derivation tree,
\((x_\T\COL\rt)\) becomes an environment of either \(u_1\) or \(U_2\).
If \((x_\T\COL\rt\p u_1 : \ot)\), %
{by Lemma~\ref{lem:envWellFormed2},} \(\rt\ra\ot\) is well-formed,
 which is a contradiction;
hence, we have
\[
\p u_1 : \ot
\qquad
x_\T\COL\rt \p U_2 : \rt\,.
\]
By item~\ref{item:subred-bal} of the induction hypothesis for \(u_1\), we have
\(\p u'_1 : \ot\) and hence
\((x_\T\COL\rt \p \TT{br}\,u'_1\,U_2 : \rt)\) as required.
Item~\ref{item:subred-bal} is similar (and easier);
and the case where \(U_2\) is reduced is also similar.

Case where the head of \(u\) is a non-terminal:
Let \(u=A_{\uty}\,U_1\cdots\,U_\ell\), \(u'\in[U_i/x'_i]_{\ind{i}{\ell}}v\),
and the rule used for \(u \ra u'\) be 
\(A_{\uty}\,x'_1\,\cdots\,x'_\ell \Hra v\).
Suppose
\begin{align*}
\uty &= \land_{\ind{i}{k_1}} \uty^{1}_{i} \ra \cdots \ra \land_{\ind{i}{k_m}} \uty^{m}_{i} \ra \uty^{0}
\\
\uty^{0} &= \land_{\ind{i}{k_{m+1}}} \uty^{m+1}_{i} \ra \cdots \ra \land_{\ind{i}{k_n}} \uty^{n}_{i} \ra \T
\end{align*}
where \(\odr{\uty^{j}_{i}} \ge 1\) for \(\ind{j}{m}\) and \(\ind{i}{k_j}\) and \(\odr{\uty^{0}} \le 1\).
In the case where \(\uty^0\) is unbalanced,
\(k_j=0\) for all \(j \in \set{m+1,\ldots,n}\), and
in the case where \(\uty^0\) is balanced,
\(k_{j_0} = 1\) for some (unique) \(j_0 \in \set{m+1,\ldots,n}\).

Let
\[
\begin{aligned}&
\RE\defe(x_\T\COL\rt)
&&\oty\defe\rt
&&(\text{in the case of item~\ref{item:subred-unb}})
\\&
\RE\defe \emptyset 
&&\oty\defe\ot
&&(\text{in the case of item~\ref{item:subred-bal}}).
\end{aligned}
\]
For the hypothesis
\((\RE \p A_{\uty}\,U_1\cdots\,U_\ell : \oty)\),
\rname{\rtApp} are used \(\ell\)-times, and we have
\begin{gather}
\label{eq:sr-nt-typing}
 \p A_{\uty} : \oty_1 \ra \cdots \ra \oty_\ell \ra \oty
\\
\label{eq:sr-arg-typing}
\RE_i\p U_i : \oty_i \quad(\ind{i}{\ell})
\\
\label{eq:sr-envunion}
\RE = \RE_1\cup\cdots\cup\RE_\ell\,.
\end{gather}

The rule used for~\eqref{eq:sr-nt-typing}
is \rname{\rtNtR} or \rname{\rtNtO}:
in the former case, we have
\begin{align*}
&\oty_1 \ra \cdots \ra \oty_\ell \ra \oty
=
\rf{\uty}
\\
=\ &
\rf{\uty^{1}_{1}} \ra \cdots \ra \rf{\uty^{1}_{k_{1}}} \ra \cdots \ra 
\rf{\uty^{m}_{1}} \ra \cdots \ra \rf{\uty^{m}_{k_{m}}} \ra
\rf{\uty^{0}}
\end{align*}
i.e.,
\begin{align}
\label{eq:otyellR}
(\oty_1 , \ldots , \oty_\ell)
&=
\big(
\rf{\uty^{1}_{1}} , \ldots , \rf{\uty^{1}_{k_{1}}} , \ldots , 
\rf{\uty^{m}_{1}} , \ldots , \rf{\uty^{m}_{k_{m}}} \big)
\\
\label{eq:otyzeroR}
\oty &= \rf{\uty^{0}}
\,.
\end{align}
In the latter case, similarly we have
\begin{align}
\label{eq:otyellO}
(\oty_1 , \ldots , \oty_\ell)
&=
\big(
\st{\uty^{1}_{1}} , \ldots , \st{\uty^{1}_{k_{1}}} , \ldots , 
\st{\uty^{m}_{1}} , \ldots , \st{\uty^{m}_{k_{m}}} \big)
\\
\label{eq:otyzeroO}
\oty &= \st{\uty^{0}}
\,.
\end{align}

Meanwhile, since the rule \(A_{\uty}\,x'_1\,\cdots\,x'_\ell \Hra v\) in \(\GRAM''\) is produced by \rname{Tr1-Rule},
there is a rule \(A\,x_1\,\cdots\,x_n \Hra s\) in \(\GRAM\) such that
\[
\p \lambda x_1.\cdots\lambda x_n.s:\uty
 \tr \lambda x'_1.\cdots\lambda x'_\ell.v 
\qquad
\uty\DCOL\NONTERMS(A).
\]
Therefore,
by \rname{Tr1-Abs1} and/or \rname{Tr1-Abs2}, we have the following.
\begin{align}
\label{eq:xsvp}&
x_1\COL\uty^1_1,
\dots,
x_1\COL\uty^1_{k_1},
\dots,
x_{n}\COL\uty^{n}_{1},
\dots,
x_{n}\COL\uty^{n}_{k_{n}}
\p s : \T \tr v'
\\\label{eq:sr-vvp-unb}&
v=v'
&&\mspace{-450mu}\text{(when \(\uty^0\) is unbalanced)}
\\\label{eq:sr-vvp-bal}&
v=[\Teps/(x_{j_0})_\T]v'
&&\mspace{-450mu}\text{(when \(\uty^0\) is balanced)}
\\&
\label{eq:xell}
\big(x'_1,\ldots,x'_\ell\big)
=
\big((x_1)_{\uty^1_1},
\dots,
(x_1)_{\uty^1_{k_1}},
\dots,
(x_{m})_{\uty^{m}_{1}},
\dots,
(x_{m})_{\uty^{m}_{k_{m}}}\big)\,.
\end{align}

From now on, the proof goes separately for each item.

Case of item~\ref{item:subred-unb}:
Since \(\oty=\rt\),
we have~\eqref{eq:otyellR} and~\eqref{eq:otyzeroR}.
By~\eqref{eq:otyzeroR}, \(\uty^0\) is unbalanced, and so we have~\eqref{eq:sr-vvp-unb}.
By~\eqref{eq:xsvp} and Lemma~\ref{lem:UTr-typing} with~\eqref{eq:otyellR},~\eqref{eq:sr-vvp-unb}, and~\eqref{eq:xell}, we have
\[
x'_1\COL\oty_1,\ldots,x'_\ell\COL\oty_\ell \p v : \rt\,.
\]
Then, by~\eqref{eq:sr-arg-typing},~\eqref{eq:sr-envunion}, and Lemma~\ref{lem:rty-subst},
we have
\[
x_{\T}\COL\rt \p [U_i/x'_i]_{\ind{i}{\ell}}v : \rt
\]
and by \rname{\rtSet},
we have
\[
x_{\T}\COL\rt \p u' : \rt
\]
as required.

Case of item~\ref{item:subred-bal}:
We have
\[
(x_1)_{\uty^1_1}\COL\st{\uty^1_1},
\dots,
(x_1)_{\uty^1_{k_1}}\COL\st{\uty^1_{k_1}},
\dots,
(x_{n})_{\uty^{n}_{1}}\COL\st{\uty^{n}_{1}},
\dots,
(x_{n})_{\uty^{n}_{k_{n}}}\COL\st{\uty^{n}_{k_{n}}}
\p v : \ot
\]
either
by using~\eqref{eq:xsvp}, Lemma~\ref{lem:refine-forget}-\ref{item:forget-typing}, and~\eqref{eq:sr-vvp-unb}
when \(\uty^0\) is unbalanced,
or
by using~\eqref{eq:xsvp}, Lemma~\ref{lem:refine-forget}-\ref{item:forget-subst}, and~\eqref{eq:sr-vvp-bal}
when \(\uty^0\) is balanced.
By~\eqref{eq:sr-arg-typing} and Lemma~\ref{lem:refine-forget}-\ref{item:forget},
we have
\[
\p U_i : \fg{\oty_i} \quad (\ind{i}{\ell})\,.
\]
By either~\eqref{eq:otyellR} or~\eqref{eq:otyellO}, we have
\[
\big(\fg{\oty_1} , \ldots , \fg{\oty_\ell}\big)
=
\big(
\st{\uty^{1}_{1}} , \ldots , \st{\uty^{1}_{k_{1}}} , \ldots , 
\st{\uty^{m}_{1}} , \ldots , \st{\uty^{m}_{k_{m}}} \big).
\]
Hence, by Lemma~\ref{lem:rty-subst},
\[
\p [U_i/x'_i]_{\ind{i}{\ell}}v : \ot
\]
and by \rname{\rtSet},
we have
\[
\p u' : \rt
\]
as required.
\end{proof}

Below, we write \(u \notred\)
if \(u \red v\) does not hold for any \(v\).

\begin{lemma}\label{lem:keyCtx}
For any \(v\) such that \(v\notred\), \(x_{\T}\COL\rt \p v : \rt\), and \([U/x_\T]v \reds \pi\) for some \(U\) and \(\pi\),
there exist \(\pi_1,\ldots,\pi_n\ (n\ge 0)\) such that
\(v = \pi_1 \ibr (\pi_2 \ibr \cdots (\pi_n \ibr x_\T)\cdots)\).
\end{lemma}
\begin{proof}
The proof proceeds by induction on \(v\).
\begin{itemize}
\item Case where the head of \(v\) is a non-terminal \(A\):
\(A\) has a rewriting rule since \([U/x_\T]v \reds \pi\), but it contradicts \(v\notred\).
\item Case where the head of \(v\) is a variable: Since \(x_{\T}\COL\rt \p v : \rt\), \(v=x_\T\); hence the result
holds for \(n = 0\).
\item Case where the head of \(v\) is a terminal \(a\): \(a\) must have non-zero arity 
since \(x_{\T}\COL\rt \p a : \rt\) cannot be derived;
thus \(v = \br\,V_0\,V_1\) for some \(V_0\) and \(V_1\).
Since \(v=\br\,V_0\,V_1\notred\), \(V_0\) and \(V_1\) must be singletons \(\set{v_0}\) and \(\set{v_1}\), respectively.
Now \(\br\) must has type \(\ot \ra \rt \ra \rt\) and hence we have \(\p v_0 : \ot\) and \(x_{\T}\COL\rt \p v_1 : \rt\)
since if we had
\(x_{\T}\COL\rt \p v_0 : \ot\) then \(\rt \ra \ot\) would be
well-formed by Lemma~\ref{lem:envWellFormed2},
which is a contradiction.
Also, since \([U/x_\T]v = \br\,([U/x_\T]v_0)\,([U/x_\T]v_1) \reds \pi\), there exist \(\pi_0\) and \(\pi_1\) such that
\(([U/x_\T]v_i) \reds \pi_i\).
Thus we can use the induction hypothesis for \(v_1\).
Now \(v_0\) is closed and hence \(v_0 \reds \pi_0\), but since \(v_0\notred\), we have \(v_0 = \pi_0\).
\end{itemize}
\end{proof}

\begin{proof}[Proof of Lemma~\ref{lem:groundsbst}]
First note that,
for any \(v'\) such that \(v \reds v'\), we have \(x_{\T}\COL\rt \p v' : \rt\).
This is because, from the assumption \(x\COL\T \p s  :  \uty \tr v\),
\(\uty\) is unbalanced by Lemma~\ref{lem:envWellFormed}, and hence
we have \(x_{\T}\COL\rt \p v' : \rt\) by Lemmas~\ref{lem:UTr-typing}
 and~\ref{lem:rty-sub-red}-\ref{item:subred-unb}.

Now we prove the goal of the current lemma by using the context lemma (Lemma~\ref{lem:context}).

Given 
\([U/x_\T]v \reds \pi\),
there exists \(v'\) such that %
\[
v \reds v' \notred \qquad [U/x_{\T}]v' \reds \pi\,.
\]
{By Lemma~\ref{lem:keyCtx},} there exist \(\pi_1,\ldots,\pi_n\) such that 
\[
v' = \pi_1 \ibr (\pi_2 \ibr \cdots (\pi_n \ibr x_{\T})\cdots).
\]
Since
\[
[U/x_{\T}]v' = \pi_1 \ibr (\pi_2 \ibr \cdots (\pi_n \ibr U)\cdots) \reds \pi
\]
there exist \(u \in U\) and \(\pi'\) such that
\[
u \reds \pi'
\qquad
\pi_1 \ibr (\pi_2 \ibr \cdots (\pi_n \ibr \pi')\cdots) = \pi.
\]
Therefore
\begin{align*}
[\Teps/x_\T]v \ibr U 
=&\ 
(\pi_1 \ibr (\pi_2 \ibr \cdots (\pi_n \ibr \Teps)\cdots)) \ibr U
\\
\reds&\ 
(\pi_1 \ibr (\pi_2 \ibr \cdots (\pi_n \ibr \Teps)\cdots)) \ibr \pi'
\\
\vsim&\ 
\pi_1 \ibr (\pi_2 \ibr \cdots (\pi_n \ibr \pi')\cdots) = \pi.
\end{align*}

On the other hand,
given \(([\Teps/x_{\T}]v)\ibr U \red^p \pi\),
there exist \(p_0\), \(p_1\), \(\pi'_0\), and \(\pi'_1\) such that
\[
[\Teps/x_{\T}]v \red^{p_0} \pi'_0
\qquad
\pi'_0 \ibr U \red^{p_1} \pi'_0 \ibr \pi'_1
 = \pi
\qquad
p_0+p_1 =p\,.
\]
For \([\Teps/x_{\T}]v \red^{p_0} \pi'_0\),
we have \(v'\) such that %
\[
v \red^{p_2} v' \notred \qquad [\Teps/x_{\T}]v' \red^{p_3} \pi'_0
\qquad
p_2+p_3=p_0\,.
\]
{By Lemma~\ref{lem:keyCtx},}  there exist \(\pi_1,\ldots,\pi_n\) such that 
\[
v' = \pi_1 \ibr (\pi_2 \ibr \cdots (\pi_n \ibr x_{\T})\cdots).
\]
Since 
\[
[\Teps/x_{\T}]v' 
=
\pi_1 \ibr (\pi_2 \ibr \cdots (\pi_n \ibr \Teps)\cdots)
\red^{p_3} \pi'_0,
\]
we have
\[
p_3=0
\qquad
p_2=p_0
\qquad
\pi_1 \ibr (\pi_2 \ibr \cdots (\pi_n \ibr \Teps)\cdots)
= \pi'_0\,.
\]
Hence,
\begin{align*}
[U/x_\T]v
\red^{p_2}
[U/x_\T]v'
=
\pi_1 \ibr (\pi_2 \ibr \cdots (\pi_n \ibr U)\cdots)
\red^{p_1}
\pi_1 \ibr (\pi_2 \ibr \cdots (\pi_n \ibr \pi'_1)\cdots)
\end{align*}
i.e.,
\[
[U/x_\T]v \red^{p} 
\pi_1 \ibr (\pi_2 \ibr \cdots (\pi_n \ibr \pi'_1)\cdots) \vsim 
\pi'_0 \ibr \pi'_1 = \pi.
\]
\end{proof}

\subsection{Lemmas for forward direction}
\label{sec:lemmaForwardDirection}

\begin{lemma}[de-substitution]\label{lem:desubst}
Given \(\TE \pp [t/x]s \COL \uty \tr v\) where \(t\) is closed and \(s\) and \(t\) are applicative terms,
there exist \(k \ge 0,\, (\uty_i)_{\ind{i}{k}},\, (U_i)_{\ind{i}{k}}\), and \(\des{v}\) such that
\begin{enumerate}
\item\label{item:desubs}
\(\TE, x \COL \uty_1,\dots,x\COL \uty_k \pp s \COL \uty \tr \des{v}\)
\item\label{item:desubt}
\(\pp t \COL \uty_i \tr U_i  \quad  (\ind{i}{k})\)
\item\label{item:unb}
for each \(\ind{i}{k}\), if \(\uty_i\) is unbalanced, \(U_i\) is a singleton
\item\label{item:subs}
\(v \lse 
[U_i/x_{\uty_i}]_{\ind{i}{k}} \des{v}
\).
\end{enumerate}
\end{lemma}
\begin{proof}
  The proof is by induction on \(s\) and analysis on the
  last rule used for deriving \(\TE \pp [t/x]s \COL \uty \tr v\).

Case \(s=x\):
Since \(\TE \pp ([t/x]x =)\, t \COL \uty \tr v\) and \(t\) is closed, %
by Lemma~\ref{lem:strengthening}, we also have \(\pp  t \COL \uty \tr v\) and \(\TE\) is balanced.
For item~\ref{item:desubs}, we define
 \(k\defe 1\), \(\uty_1 \defe \uty\), and \(\des{v} \defe x_\uty\).
We define \(U_1\defe\set{v}\)
for items~\ref{item:desubt} and~\ref{item:unb}.
Then, item~\ref{item:subs} is clear.

Case \(s=a\), \(A\), or \(y\neq x\): Since \(x \notin \FV(s)\), \([t/x]s=s\).
So we define \(k \defe 0\), \(\des{v} \defe v\).

Case \(s\) is an application:
the last rule used for \(\TE \pp [t/x]s \COL \uty \tr v\) is
\rname{Tr1-App1} or \rname{Tr1-App2}:
\[
\InfruleS{}
  {\TEpz \pp [t/x]s'\COL \uty'_1\land\cdots\land \uty'_{k'}\ra\uty \tr v' \\
\TE'_j\pp [t/x]t'\COL\uty'_j\tr U'_j 
\quad\text{(for each $j \in \set{1,\dots,{k'}}$)} 
}
{\TEpz\cup\TE'_1\cup\cdots\cup\TE'_{k'}\pp \app{([t/x]s')}{([t/x]t')} :\uty\tr 
\begin{cases}
v'U'_1\cdots U'_{k'} &(\odr{t'} \ge 1 \lor k'=0)
\\
v' \ibr U'_1 &(\odr{t'} =0 \land k'=1)
\end{cases}
}
\]
where
\begin{align}\label{eq:desubAppEnv}
\TE
&=
\TEpz\cup\TE'_1\cup\cdots\cup\TE'_{k'}
\\\notag
s
&=
s't'
\\\notag
v
&=
\begin{cases}
v'U'_1\cdots U'_{k'} &(\odr{t'} \ge 1 \lor k'=0)
\\
v' \ibr U'_1 &(\odr{t'} =0 \land k'=1)\,.
\end{cases}
\end{align}
For each \(\ind{j}{k'}\), 
the rule used last for 
\(\TE'_j\pp [t/x]t'\COL\uty'_j\tr U'_j\) is \rname{Tr1-SetS}:
\begin{equation}\label{eq:desSetS}
\InfruleS{}{
\TE'_j\pp [t/x]t'\COL\uty'_j\tr u'_{jh} 
\quad\text{(for each $h \in \set{1,\dots,{k_j}}$)} 
\\
k_j=1\mbox{ if $\uty'_j$ is unbalanced}
}{
\TE'_j\pp [t/x]t'\COL\uty'_j\tr \set{u'_{j1},\dots,u'_{jk_{j}}} (=U'_j)
}
\end{equation}
Hence, by induction hypotheses for \(s'\) and for \(t'\),
there exist \(k^0 \ge 0,\, (\uty^0_i)_{\ind{i}{k^0}},\, (U^0_i)_{\ind{i}{k^0}}\), and \(\des{v'}\) such that
\begin{gather}\label{eq:desubIH0}
\TEpz, x \COL \uty^0_1,\dots,x\COL \uty^0_{k^0} \pp s' \COL \uty'_1\land\cdots\land \uty'_{k'}\ra\uty \tr \des{v'}
\\\label{eq:desubIH0t}
\pp t \COL \uty^0_i \tr U^0_i  \quad  (\ind{i}{k^0})
\\\label{eq:desubIH0single}
\text{for each \(\ind{i}{k^0}\), if \(\uty^0_i\) is unbalanced, \(U^0_i\) is a singleton}
\\\label{eq:desubIH0-subs}
v' \lse 
[U^0_i/x_{\uty^0_i}]_{\ind{i}{k^0}} \des{v'}
\end{gather}
and for each \(\ind{j}{k'}\) and \(\ind{h}{k_j}\)
there exist \(k^{jh} \ge 0,\, (\uty^{jh}_i)_{\ind{i}{k^{jh}}},\, (U^{jh}_i)_{\ind{i}{k^{jh}}}\), and \(\des{u_{jh}'}\) such that
\begin{gather}\label{eq:desubIHjh}
\TE'_j, x \COL \uty^{jh}_1,\dots,x\COL \uty^{jh}_{k^{jh}} \pp t' \COL \uty'_j \tr \des{u_{jh}'}
\\\label{eq:desubIHjht}
\pp t \COL \uty^{jh}_i \tr U^{jh}_i  \quad  (\ind{i}{k^{jh}})
\\\label{eq:desubIHjhsingle}
\text{for each \(\ind{i}{k^{jh}}\), if \(\uty^{jh}_i\) is unbalanced, \(U^{jh}_i\) is a singleton}
\\\label{eq:desubIHjh-subs}
u_{jh}' \lse 
[U^{jh}_i/x_{\uty^{jh}_i}]_{\ind{i}{k^{jh}}} \des{u_{jh}'}\,.
\end{gather}

For each \(\ind{j}{k'}\),
by \rname{Tr1-SetS}
and a (derived) weakening rule, we have
\[
\infers{
\TE'_j\cup
\set{
 x \COL \uty^{jh}_i
\,|\, h \le {k_j},
i \le k^{jh}
}
\pp t' \COL \uty'_j \tr 
\set{
\des{u_{jh}'}
 \,|\, \ind{h}{k_j}}
}{
\infers{
\TE'_j\cup
\set{
 x \COL \uty^{jh}_i
\,|\, h \le {k_j},
i \le k^{jh}
}
 \pp t' \COL \uty'_j \tr \des{u_{jh}'}
\quad \text{($\ind{h}{k_j}$)} 
}{
\TE'_j, x \COL \uty^{jh}_1,\dots,x\COL \uty^{jh}_{k^{jh}} \pp t' \COL \uty'_j \tr \des{u_{jh}'}
\quad \text{($\ind{h}{k_j}$)}
}
\andalso
k_j=1\mbox{ if $\uty'_j$ is unbalanced}
}
\]
where, when \(\uty'_j\) is unbalanced, since \(k_j=1\) we do not need the weakening rule;
when \(\uty'_j\) is balanced, by {Lemma~\ref{lem:envWellFormed} applied to~\eqref{eq:desubIHjh}},
\(\uty^{jh}_i\) must be balanced for each \(h\) and \(i\), and hence we can use the weakening rule.
Now we define
\[
\des{U_j'} \defe \set{
\des{u_{jh}'}
 \,|\, \ind{h}{k_j}}\,.
\]
Then, by \rname{Tr1-App1} or \rname{Tr1-App2} with~\eqref{eq:desubIH0} and~\eqref{eq:desubAppEnv},
we have
\[
\TE \cup \set{x\COL\uty^0_i\,|\,\ind{i}{k^0}}
\cup (\cup_{\ind{j}{k'}}
\set{
 x \COL \uty^{jh}_i
\,|\, h \le {k_j},
i \le k^{jh}})
\pp
\app{s'}{t'} : \uty \tr
\des{v}
\]
where
\[
\des{v} \defe
\begin{cases}
\appi{
\des{v'}
}{
\des{U_j'}
}{j}{k'}
 &(\odr{t'} \ge 1 \lor k'=0)
\\
\des{v'} \ibr \des{U_1'} &(\odr{t'} =0 \land k'=1).
\end{cases}
\]
We define \(k\) and \((\uty_i)_{\ind{i}{k}}\) as the following enumeration:
\[
\set{x\COL\uty_i\,|\, \ind{i}{k}}
\defe
\set{x\COL\uty^0_i\,|\,\ind{i}{k^0}}
\cup
\set{x\COL\uty^{jh}_i\,|\,\ind{j}{k'},\ind{h}{k_j},\ind{i}{k^{jh}}}\,.
\]
Thus we have obtained item~\ref{item:desubs}.

For each \(\ind{i}{k}\) we define
\[
U_i \defe \cup(\set{U^0_{i'}\,|\, \uty^0_{i'} = \uty_i} \cup \set{U^{jh}_{i'} \,|\, \uty^{jh}_{i'} = \uty_i})\,.
\]
By~\eqref{eq:desubIH0t}, ~\eqref{eq:desubIHjht}, and \rname{Tr1-SetS}, for each \(\ind{i}{k}\) we have
\[
\pp t : \uty_i \tr U_i
\]
where, when \(\uty_i\) is unbalanced, we use~\eqref{eq:desubIH0single} and~\eqref{eq:desubIHjhsingle} and
we can show that if \(\uty_i\) is unbalanced, then 
\(\set{i' \le k^0 \,|\, \uty^0_{i'} = \uty_i} \cup \set{(j,h,i') \,|\, \uty^{jh}_{i'} = \uty_i}\) 
is a singleton as follows.
The set is non-empty by the definition of \(\uty_i\). 
If \(\set{i' \le k^0 \,|\, \uty^0_{i'} = \uty_i}\)
is non-empty, it is a singleton,
and we show that \(\set{(j,h,i') \,|\, \uty^{jh}_{i'} = \uty_i}\) is empty.
For every \(j\), \(h\), and \(i'\), 
by {Lemma~\ref{lem:envWellFormed} applied to~\eqref{eq:desubIH0} and~\eqref{eq:desubIHjh}}
and the fact that \(\uty^0_{i'} (=\uty_i)\) is unbalanced for some \(i'\),
\(\uty'_{j}\) and
\(\uty^{jh}_{i'}\) must be balanced. Hence, \(\uty^{jh}_{i'}\neq \uty_i\) as \(\uty_i\) is unbalanced.
If \(\set{(j,h,i') \,|\, \uty^{jh}_{i'} = \uty_i}\) is non-empty,
similarly, \(\set{i' \le k^0 \,|\, \uty^0_{i'} = \uty_i}\) is empty.
To show that \(\set{(j,h,i') \,|\, \uty^{jh}_{i'} = \uty_i}\) is a singleton,
suppose \(\uty^{{j_0}{h_0}}_{i_0} = \uty^{{j_1}{h_1}}_{i_1} = \uty_i\). 
By {Lemma~\ref{lem:envWellFormed} applied to~\eqref{eq:desubIHjh}},
\(\uty'_{j_0}\) and \(\uty'_{j_1}\) are unbalanced. 
Hence \(j_0=j_1\) from the well-formedness of \(\uty'_1\land\cdots\land \uty'_{k'}\ra\uty\), and we have also \(k_{j_0}=1\) from~\eqref{eq:desSetS}.
Therefore \(h_0=h_1 \ (\le k_{j_0}=1)\), and then \(i_0=i_1\).
Thus, we have obtained items~\ref{item:desubt} and~\ref{item:unb}.

Finally we show item~\ref{item:subs}, i.e.,
 \(v \lse [U_i/x_{\uty_i}]_{i\le k} \des{v}\).
In the case where \(\odr{t'}\ge1 \lor k'=0\),
\begin{align*}
v
=&\
\appi{v'}{U'_j}{j}{k'}
\\ %
\lse&\
\appi{
([U^0_i/x_{\uty^0_i}]_{\ind{i}{k^0}}\des{v'})
}{
\left\{
[U^{jh}_i/x_{\uty^{jh}_i}]_{\ind{i}{k^{jh}}}
\des{u_{jh}'}
\,\middle|\,
\ind{h}{k_j}
\right\}
}{j}{k'}
&&\reason{by~\eqref{eq:desubIH0-subs} and~\eqref{eq:desubIHjh-subs}}
\\ \lse&\ %
\appi{
([U_i/x_{\uty_i}]_{\ind{i}{k}}\des{v'})
}{
\left\{
[U_i/x_{\uty_i}]_{\ind{i}{k}}\des{u_{jh}'}
\,\middle|\,
\ind{h}{k_j}
\right\}
}{j}{k'}
\\ %
=&\
[U_i/x_{\uty_i}]_{i\le k} \des{v}\,.
\end{align*}
In the case where \(\odr{t'}=0 \land k'=1\),
\begin{align*}
v=&\
v'\ibr U'_1
\\ %
\lse&\
([U^0_i/x_{\uty^0_i}]_{\ind{i}{k^0}} \des{v'})
\ibr
([U^{11}_i/x_{\uty^{11}_i}]_{\ind{i}{k^{11}}} \des{u_{11}'})
&&\reason{by~\eqref{eq:desubIH0-subs} and~\eqref{eq:desubIHjh-subs}, and now \(k_1=1\)}
\\ %
\lse&\
([U_i/x_{\uty_i}]_{\ind{i}{k}} \des{v'})
\ibr
([U_i/x_{\uty_i}]_{\ind{i}{k}} \des{u_{11}'})
\\ %
=&\
[U_i/x_{\uty_i}]_{i\le k} \des{v}\,.
\end{align*}
\end{proof}

The following lemma states that the transformation relation (up to \(\lse\))
is a left-to-right backward simulation relation.

\begin{lemma}[subject expansion]
\label{lem:red-fwd}
If \(t\redwith{\GRAM}t'\) and \(\pp t'\COL\T \tr u'\),
then there exists \(u\) such that 
\( \pp t\COL\T \tr u\) with \(u\redwith{\GRAM''}\gse u'\).
\end{lemma}
\begin{proof}
The proof is given by the induction on \(t\) and by the case analysis of the reduction \(t\redwith{\GRAM}t'\).

Case where \(t=\Te\): Trivial.

Case where \(t=\app{a}{t_1}\) and \(\TERMS(a)=1\): 
Let the last rule used for \(t\redwith{\GRAM}t'\) be
\infrule{t_1\redwith{\GRAM}t_1'}
 {a\,t_1 \redwith{\GRAM}
  a\,t_1'
}
Since \(\pp (t' =)\ a\,t_1'\COL\T \tr u'\) is derived by~\rname{Tr1-App2} and~\rname{Tr1-SetS},
we have \(\pp t'_1 : \T \tr u'_1 \) such that \(u' = \TT{br}\,a\,u'_1\).
Hence by the induction hypothesis for \(t'_1\), there exists \(u_1\) such that
\(\pp t_1 : \T \tr u_1\) and \(u_1 \redwith{\GRAM''} \gse u'_1\).
Therefore we have \(u\defe \TT{br}\,a\,u_1 \) with
\(\pp a\,t_1 : \T \tr \TT{br}\,a\,u_1\) and \(\TT{br}\,a\,u_1 \redwith{\GRAM''}\gse \TT{br}\,a\,u_1'\).

Case where \(t=A\,t_1\,\dots\,t_n\):
Let the last rule used for \(t\redwith{\GRAM}t'\) be
\infrule{A\,x_1\,\cdots\,x_n \Hra s\in \RULES}
{A\,t_1\,\cdots\,t_n \redwith{\GRAM} [t_1/x_1,\ldots,t_n/x_n]s}
and let \(\NONTERMS(A) = \kappa_1\ra\cdots\ra\kappa_n\ra\T\).
By the assumption on sorts, there exists unique \(m\) such that
\(0\le m\le n\), \(\odr{\kappa_j} \ge 1\) for all \(\ind{j}{m}\),
and \(\odr{\kappa_j}=0\) for all \(j >m\).
Let \(v^{n+1} \defe u'\) and \(\TE^{n+1} \defe \emptyset\);
then 
\[
\TE^{n+1} \pp [t_1/x_1,\ldots,t_n/x_n]s :\T \tr v^{n+1}\,.
\]
Hence by Lemma~\ref{lem:desubst},
for each \(j=n,\dots,1\),
there exist
\(
\TE^j,
k^j,
(\uty^{j}_{i})_{\ind{i}{k^j}},
(U^{j}_{i})_{\ind{i}{k^j}},
v^j
\) such that
\begin{align*}&
\TE^j=(\TE^{j+1}, x_j\COL\uty^{j}_{1},\dots,x_j\COL\uty^{j}_{k^{j}})
\\&
\phantom{\TE^j}\pp t_j : \uty^{j}_{i} \tr U^{j}_{i} \qquad(\ind{i}{k^j})
\\&
\TE^j \pp [t_{j'}/x_{j'}]_{\ind{j'}{j-1}}s : \T \tr v^{j}
\\&
\!\!\text{
for each \(\ind{i}{k^j}\), if \(\uty^{j}_{i}\) is unbalanced, \(U^{j}_{i}\) is a singleton 
}
\\&
v^{j+1} \lse 
[U^j_i/{(x_{j})}_{\uty^j_i}]_{\ind{i}{k^j}} v^{j}\,.
\end{align*}
Note that for each \(j>m\), \({k^j}\le1\), and there is at most one \(j>m\) such that \({k^j}=1\) 
by %
Lemma~\ref{lem:envWellFormed}. %

By~\rname{Tr1-NT}, we have \(\pp A : \uty \tr A_{\uty}\).
Since we have also \(\pp t_j : \uty^{j}_{i} \tr U^{j}_{i}\) (\(\ind{i}{k^j}\))
for \(j=1,\dots,m\), by using~\rname{Tr1-App1} iteratively,
we have
\[
\pp A\,t_1\cdots\,t_m :
\land_{\ind{i}{k^{m+1}}}\uty^{m+1}_i\ra\cdots\land_{\ind{i}{k^{n}}}\uty^{n}_i\ra\T
\tr
\appi{\appi{
A_{\uty}
}{U^1_i}{i}{k^1}
\!\!\!\cdots
}{U^m_i}{i}{k^m}.
\]
Then, since we have \(\pp t_j : \uty^{j}_{i} \tr U^{j}_{i}\) (\(\ind{i}{k^j}\))
for \(j=m+1,\dots,n\),
by using~\rname{Tr1-App1} where \(k=0\) and/or~\rname{Tr1-App2} iteratively,
we have
\begin{align*}&
\pp A\,t_1\cdots\,t_n :
\T
\tr u
\\&
u\defe
\begin{cases}
\Big(\appi{\appi{
A_{\uty}
}{U^1_i}{i}{k^1}
\!\!\!\cdots
}{U^m_i}{i}{k^m}\Big)
\ibr
U^{j}_1
&\text{(\(k_{j}>0\) for some (unique) \(j\in\set{m+1,\dots,n}\))}
\\
\phantom{(}
\appi{\appi{
A_{\uty}
}{U^1_i}{i}{k^1}
\!\!\!\cdots
}{U^m_i}{i}{k^m}
&\text{(otherwise)}.
\end{cases}
\end{align*}

Meanwhile, since
we have 
\(\TE^1 \pp [t_{j'}/x_{j'}]_{\ind{j'}{0}}s : \T \tr v^{1}\), so do
\(\TE^1 \p [t_{j'}/x_{j'}]_{\ind{j'}{0}}s : \T \tr v^{1}\),
i.e.,
\[
x_1\COL\uty^{1}_{1},\dots,x_1\COL\uty^{1}_{k^{1}},
\dots,
x_n\COL\uty^{n}_{1},\dots,x_n\COL\uty^{n}_{k^{n}}
\p s : \T \tr v^{1}\,.
\]
Now we define
\begin{equation}\label{eq:def-vzero}
v^0\defe
\begin{cases}
[\Teps/{(x_{j})}_{\T}]v^1 & \text{(\(k_{j}>0\) for some (unique) \(j\in\set{m+1,\dots,n}\))}
\\
v^1 & \text{(otherwise)}.
\end{cases}
\end{equation}
By~iterating \rname{Tr1-Abs1} where \(k=0\)
and/or~\rname{Tr1-Abs2},
we have
\begin{align*}&
x_1\COL\uty^{1}_{1},\dots,x_1\COL\uty^{1}_{k^{1}},
\dots,
x_m\COL\uty^{m}_{1},\dots,x_m\COL\uty^{m}_{k^{m}}
\p
\\&
\lambda x_{m+1}. \cdots \lambda x_n. s
: 
\land_{\ind{i}{k^{m+1}}}\uty^{m+1}_{i}\ra\cdots\ra\land_{\ind{i}{k^{n}}}\uty^{n}_{i}\ra
\T \tr v^{0}
\end{align*}
and by~iterating \rname{Tr1-Abs1}, we have
\begin{align*}
\p
&\lambda x_{1}. \cdots \lambda x_n. s
: 
\land_{\ind{i}{k^{1}}}\uty^{1}_{i}\ra\cdots\ra\land_{\ind{i}{k^{n}}}\uty^{n}_{i}\ra
\T \tr 
\\
&\lambda {(x_1)}_{\uty^{1}_{1}}.\cdots \lambda {(x_1)}_{\uty^{1}_{k^{1}}}.
\cdots
\lambda {(x_m)}_{\uty^{m}_{1}}.\cdots \lambda {(x_m)}_{\uty^{m}_{k^{m}}}.
v^{0}\,.
\end{align*}
Hence, by \rname{Tr1-Rule}, we have
\begin{align}&
\p
(A\,x_1\,\cdots\,x_n \Hra s)
\tr
\big(A_{\uty} 
\, {(x_1)}_{\uty^{1}_{1}}\,\cdots \, {(x_1)}_{\uty^{1}_{k^{1}}}\,
\cdots
\, {(x_m)}_{\uty^{m}_{1}}\,\cdots \, {(x_m)}_{\uty^{m}_{k^{m}}}
\Hra
v^{0}\big)
\label{eq:subexp-Xformed-rule}
\end{align}
where \(\uty \defe \land_{\ind{i}{k^{1}}}\uty^{1}_{i}\ra\cdots\ra\land_{\ind{i}{k^{n}}}\uty^{n}_{i}\ra\T\).

By Lemmas~\ref{lem:subst-monotone} and~\ref{lem:subst-iter} and
since 
\(
v^{2} \lse 
[U^1_i/{(x_{1})}_{\uty^1_i}]_{\ind{i}{k^1}} v^{1}
\),
we have
\(
v^{3} \lse
[U^{2}_i/{(x_{2})}_{\uty^{2}_i}]_{\ind{i}{k^{2}}}v^{2} \lse 
([U^1_i/{(x_{1})}_{\uty^1_i}]_{\ind{i}{k^1}} \cup [U^2_i/{(x_{2})}_{\uty^2_i}]_{\ind{i}{k^2}}) v^{1}
\).
Iterating this reasoning, we have
\begin{equation}\label{eq:subexp-subst}
v^{m+1} \lse ([U^1_i/{(x_{1})}_{\uty^1_i}]_{\ind{i}{k^1}} \cup\cdots\cup [U^m_i/{(x_{m})}_{\uty^m_i}]_{\ind{i}{k^m}})v^1\,.
\end{equation}
Further,
\begin{align}\notag
v^{n+1} 
\lse\ & ([U^{m+1}_i/{(x_{m+1})}_{\uty^{m+1}_i}]_{\ind{i}{k^{m+1}}} \cup\cdots\cup
 [U^n_i/{(x_{n})}_{\uty^n_i}]_{\ind{i}{k^n}})v^{m+1}
\\ = \ &\label{eq:subexp-subst2}
\begin{cases}
[U^j_1/{(x_{j})}_{\T}]v^{m+1} & \text{(\(k_{j}>0\) for some (unique) \(j\in\set{m+1,\dots,n}\))}
\\
v^{m+1} & \text{(otherwise)}.
\end{cases}
\end{align}

In the case where \(k_{j}>0\) for some \(j\in\set{m+1,\dots,n}\),
we have 
\begin{align*}
u=\
&\Big(\appi{\appi{
A_{\uty}
}{U^1_i}{i}{k^1}
\!\!\!\cdots
}{U^m_i}{i}{k^m}\Big)
\ibr U^j_1
\\
\redwith{\GRAM''}\ &
\big(([U^1_i/{(x_{1})}_{\uty^1_i}]_{\ind{i}{k^1}} \cup\cdots\cup [U^m_i/{(x_{m})}_{\uty^m_i}]_{\ind{i}{k^m}})v^0\big)
\ibr U^j_1
&& \reason{by~\eqref{eq:subexp-Xformed-rule}}
\\ = \ &
\big([\Teps/{(x_j)}_{\T}]([U^1_i/{(x_{1})}_{\uty^1_i}]_{\ind{i}{k^1}} \cup\cdots\cup [U^m_i/{(x_{m})}_{\uty^m_i}]_{\ind{i}{k^m}})v^1\big)
\ibr U^j_1
&& \reason{by~\eqref{eq:def-vzero}}
\\ \gse \ &
\big([\Teps/{(x_j)}_{\T}]
v^{m+1}\big)
\ibr U^j_1
&&\reason{by~\eqref{eq:subexp-subst}}
\\ \se \ &
[U^j_1/{(x_j)}_{\T}]v^{m+1}
&&\reason{Lemma~\ref{lem:groundsbst}}
\\ \gse \ &
v^{n+1} = u'
&&\reason{by~\eqref{eq:subexp-subst2}}.
\end{align*}
In the other case,
we have 
\begin{align*}
u=\
&\appi{\appi{
A_{\uty}
}{U^1_i}{i}{k^1}
\!\!\!\cdots
}{U^m_i}{i}{k^m}
\\
\redwith{\GRAM''}\ &
([U^1_i/{(x_{1})}_{\uty^1_i}]_{\ind{i}{k^1}} \cup\cdots\cup [U^m_i/{(x_{m})}_{\uty^m_i}]_{\ind{i}{k^m}})v^0
&& \reason{by~\eqref{eq:subexp-Xformed-rule}}
\\ = \ &
([U^1_i/{(x_{1})}_{\uty^1_i}]_{\ind{i}{k^1}} \cup\cdots\cup [U^m_i/{(x_{m})}_{\uty^m_i}]_{\ind{i}{k^m}})v^1
&& \reason{by~\eqref{eq:def-vzero}}
\\ \gse \ &
v^{m+1} 
&&\reason{by~\eqref{eq:subexp-subst}}
\\ \gse \ &
v^{n+1} = u'
&&\reason{by~\eqref{eq:subexp-subst2}}.
\end{align*}

\end{proof}

\subsection{Lemmas for backward direction}
\label{sec:lemmaBackwardDirection}

\newcommand{\dlt}[1]{} %

For a given \(\TE\),\,
we write \(\mn{\TE}{x}\) for \(\TE'\) such that \(\TE = (\TE' , x\COL\uty_1,\dots, x\COL\uty_n)\)
for some \(\uty_1,\dots,\uty_n\) and \(x \notin \dom(\TE')\).

\begin{lemma}[substitution]\label{lem:substOld}
Given
\(\TE, x \COL \uty_1,\dots,x\COL \uty_k \p s \COL \uty \tr v\) 
where \(x \notin \dom(\TE)\) and \(k \ge 0\), and given
\(\p t \COL \uty_i \tr U_i\) for each \(\ind{i}{k}\), we have
\[
\TE \p [t/x]s \COL \uty \tr %
\dlt{\sbs{V} \gp} [U_i/x_{\uty_i}]_{\ind{i}{k}} v\,.
\]
\end{lemma}

\begin{proof}
The proof is given by induction on \(\TE, x \COL \uty_1,\dots,x\COL \uty_k \p s \COL \uty \tr v\).
For any \(\TE\), we define
\(\res{\TE} \defe \set{i \in \set{1,\ldots,k} \,|\, x:\uty_i \in \TE}\).
The base cases are clear; in the case of variables, we use a derived rule of weakening for
balanced environments.

Case of \rname{Tr1-App1}:
\[
\InfruleS{}{
\TE'_0 \p s' : \uty'_1 \land \cdots \land \uty'_{k'} \ra \uty \tr v'
\\
\TE'_j \p t' : \uty'_j \tr U'_j
 \text{ and }\uty'_i\neq \T
\text{ (for each $i \in \set{1,\dots,k'}$)}
}{
\TE'_0 \cup \TE'_1 \cup \cdots \cup \TE'_{k'}
\p \app{s'}{t'} : \uty \tr 
\appi{v'}{U'_{j}}{j}{k'} 
}
\]
We have
\begin{align*}&
\TE, x \COL \uty_1,\dots,x\COL \uty_{k} = 
\TE'_0 \cup \TE'_1 \cup \cdots \cup \TE'_{k'}
\\&
s= \app{s'}{t'}
\\&
v= \appi{v'}{U'_{j}}{j}{k'} \,.
\end{align*}
The rule used for \((\TE'_j \p t' : \uty'_j \tr U'_j)\) is \rname{Tr1-Set}:
\[
\InfruleS{}{
\TE'_{j} \p t'\COL \uty'_j \tr u'_{jh} \ \ (\ind{h}{k_j})
}{
\TE'_{j} \p t'\COL \uty'_j \tr \set{u'_{j1},\ldots,u'_{jk_j}} \ (= U'_j)
}
\]

By the induction hypothesis for \(s'\) and \(t'\), we have
\begin{gather}
\label{eq:sbsIH0-app1}
\mn{\TE'_0}{x} \p [t/x]s' : \uty'_1 \land \cdots \land \uty'_{k'} \ra \uty \tr
\dlt{\sbs{V'} \gp} [U_i/x_{\uty_i}]_{i \in \res{\TE'_0}}v'
\\\label{eq:sbsIHjh-app1}
\mn{\TE'_{j}}{x} \p [t/x]t' : \uty'_j \tr
\dlt{\sbs{U_{jh}'} \gp}
 [U_i/x_{\uty_i}]_{i \in \res{\TE'_{j}}}u'_{jh}
\qquad (j \in \set{1,\ldots,k'}, \ind{h}{k_j})
\end{gather}
and by using \rname{Tr1-Set}, from~\eqref{eq:sbsIHjh-app1}, we have
\begin{equation}\label{eq:sbsIHjh-gather-app1}
\mn{\TE'_{j}}{x} \p [t/x]t' : \uty'_j \tr
\dlt{\sbs{U_{jh}'} \gp}
\cup_{\ind{h}{k_j}}
 [U_i/x_{\uty_i}]_{i \in \res{\TE'_{j}}}u'_{jh}
\qquad (j \in \set{1,\ldots,k'}).
\end{equation}

Now
\begin{align*}
[U_i/x_{\uty_i}]_{\ind{i}{k}}
v
&=
\appi{
([U_i/x_{\uty_i}]_{\ind{i}{k}}
v')
}{
\cup_{\ind{h}{k_j}}
[U_i/x_{\uty_i}]_{\ind{i}{k}} u'_{jh}
}{j}{k'}  
\\
&=
\appi{
([U_i/x_{\uty_i}]_{i \in \res{\TE'_0}}
v')
}{
\cup_{\ind{h}{k_j}}
[U_i/x_{\uty_i}]_{i \in \res{\TE'_{j}}} u'_{jh}
}{j}{k'}  
\\
&=
\left\{
\appi{
\sbs{v'}
}{
\cup_{\ind{h}{k_j}}
[U_i/x_{\uty_i}]_{i \in \res{\TE'_{j}}} u'_{jh}
}{j}{k'}  
\,\middle|\,
\sbs{v'} \in
[U_i/x_{\uty_i}]_{i \in \res{\TE'_0}} v'
\right\}
\end{align*}
where the second equation is shown by Lemma~\ref{lem:envu}.
For any
\(\sbs{v'} \in
[U_i/x_{\uty_i}]_{i \in \res{\TE'_0}} v'\), 
by \rname{Tr1-Set} and~\eqref{eq:sbsIH0-app1}, we have
\[
\mn{\TE'_0}{x} \p [t/x]s' : \uty'_1 \land \cdots \land \uty'_{k'} \ra \uty \tr
\sbs{v'}
\]
and hence,
by \rname{Tr1-App1} with~\eqref{eq:sbsIHjh-gather-app1}, we have
\[
\cup_{j\in\set{0,\ldots,k'}} (\mn{\TE'_{j}}{x})
\p
([t/x]s')([t/x]t') : \uty \tr
\appi{
\sbs{v'}
}{
\cup_{\ind{h}{k_j}}
[U_i/x_{\uty_i}]_{i \in \res{\TE'_{j}}} u'_{jh}
}{j}{k'}
\]
where the linearity condition is satisfied, as
\[
(\mn{\TE'_{j}}{x}) \cap
(\mn{\TE'_{j'}}{x})
\subseteq
\TE'_{j} \cap \TE'_{j'}
\subseteq
\text{ (the set of balanced terms)}
\]
for \(j\neq j'\).
Therefore, again by \rname{Tr1-Set},
\[
\cup_{j\in\set{0,\ldots,k'}} (\mn{\TE'_{j}}{x})
\p
([t/x]s')([t/x]t') : \uty \tr
[U_i/x_{\uty_i}]_{\ind{i}{k}}
v
\]
Since \(\cup_{j\in\set{0,\ldots,k'}} (\mn{\TE'_{j}}{x}) = \TE\)
and \(([t/x]s')([t/x]t')=[t/x](\app{s'}{t'})\), we have shown the required condition.

Case of \rname{Tr1-App2}:
\[
\InfruleS{}{
\TE'_0 \p s' : \T \ra \uty \tr V'
\andalso
\TE'_1 \p t' : \T \tr U'
}{
\TE'_0 \cup \TE'_1
\p \app{s'}{t'} : \uty \tr 
\TT{br}\,V'\,U'
}
\]
We have
\begin{align*}&
\TE, x \COL \uty_1,\dots,x\COL \uty_{k} = 
\TE'_0 \cup \TE'_1
\\&
s= \app{s'}{t'}
\\&
v= \TT{br}\,V'\,U'\,.
\end{align*}
The rule used for \((\TE'_0 \p s' : \T \ra \uty \tr V')\) and
\((\TE'_1 \p t' : \T \tr U')\) is \rname{Tr1-Set}:
\[
\InfruleS{}{
\TE'_{0} \p s'\COL \T \ra \uty \tr v'_{h} \ \ (\ind{h}{k_0})
}{
\TE'_0 \p s' : \T \ra \uty \tr
\set{v'_{1},\ldots,v'_{k_0}} \ (=  V')
}
\]
\[
\InfruleS{}{
\TE'_{1} \p t'\COL \T \tr u'_{h} \ \ (\ind{h}{k_1})
}{
\TE'_1 \p t' : \T \tr
\set{u'_{1},\ldots,u'_{k_1}} \ (=  U')
}
\]

By the induction hypothesis for \(s'\) and \(t'\), we have
\begin{align}
\label{eq:sbsIH0-app2}
&\mn{\TE'_0}{x} \p [t/x]s' : \T \ra \uty \tr
\dlt{\sbs{V'} \gp} [U_i/x_{\uty_i}]_{i \in \res{\TE'_0}}v'_h
&&\mspace{-100mu} (\ind{h}{k_0}) \mspace{100mu}
\\\label{eq:sbsIHjh-app2}&
\mn{\TE'_{1}}{x} \p [t/x]t' : \T \tr
\dlt{\sbs{U_{1h}'} \gp}
 [U_i/x_{\uty_i}]_{i \in \res{\TE'_{1}}}u'_h
&&\mspace{-100mu} (\ind{h}{k_1}) \mspace{100mu}
\end{align}
and by using \rname{Tr1-Set} (going and back), from~\eqref{eq:sbsIH0-app2} and~\eqref{eq:sbsIHjh-app2},
we have
\begin{align*}
&\mn{\TE'_0}{x} \p [t/x]s' : \T \ra \uty \tr
\dlt{\sbs{V'} \gp} [U_i/x_{\uty_i}]_{i \in \res{\TE'_0}}V'
\\&
\mn{\TE'_{1}}{x} \p [t/x]t' : \T \tr
\dlt{\sbs{U_{1h}'} \gp}
 [U_i/x_{\uty_i}]_{i \in \res{\TE'_{1}}}U'
\end{align*}
Hence,
by \rname{Tr1-App2}, we have
\[
(\mn{\TE'_{0}}{x}) \cup (\mn{\TE'_{1}}{x})
\p
([t/x]s')([t/x]t') : \uty \tr
\TT{br}\,
([U_i/x_{\uty_i}]_{i \in \res{\TE'_0}}V')\,
([U_i/x_{\uty_i}]_{i \in \res{\TE'_1}}U')
\]
where the linearity condition is clear as shown in the previous case.
Since
\begin{align*}
[U_i/x_{\uty_i}]_{\ind{i}{k}} v
&=
\TT{br}\,
([U_i/x_{\uty_i}]_{\ind{i}{k}}V')\,
([U_i/x_{\uty_i}]_{\ind{i}{k}}U')
\\&=
\TT{br}\,
([U_i/x_{\uty_i}]_{i \in \res{\TE'_0}}V')\,
([U_i/x_{\uty_i}]_{i \in \res{\TE'_1}}U')
\end{align*}
we have shown the required condition.
\end{proof}

The following lemma states that, roughly speaking, the transformation relation is a right-to-left forward
simulation relation; also, this can be seen as a form of subject reduction.
\begin{lemma}[]
\label{lem:red-bwd}
Given \(u \red^p_{\GRAM''} \pi\) where \(p>0\) and \( \p t\COL\T \tr u\),
there exist \(t'\), \(u'\), \(\pi'\) and \(q<p\) such that 
\(t\reds_{\GRAM}t'\), \( \p t'\COL\T \tr u'\), and \(u' \red^{q} \pi' \vsim \pi\).
\end{lemma}
\begin{proof}
The proof is given by the induction on \(t\) and by the case analysis of the head of \(u\).
Since \(u \red^{+} \pi\), the head of \(u\) must be \(\TT{br}\) or a non-terminal.

Case where \(u=\TT{br}\,V'\,U'\):
In the reduction \(\br\,V'\,U'\red^p\pi\) 
suppose that \(v' \in V'\) and \(u' \in U'\) are chosen.
The last rule used for \(\p t:\T \tr \TT{br}\,V'\,U'\) is either
\rname{Tr1-App1}:
\[
\InfruleS{}{
\p s'\COL \top\ra\T\tr \TT{br}\,V'\,U'
}{
\p s't':\T\tr \TT{br}\,V'\,U'
}
\]
or \rname{Tr1-App2}:
\[
\InfruleS{}{
\p s'\COL \T\ra\T\tr V'
\andalso
\p t'\COL\T\tr U'
}{
\p \app{s'}{t'} \COL\T\tr \TT{br}\,V'\,U'
}
\]
In the former case above, we can iterate this reasoning, and then
there exist \({n'}\ge1\), \(s'\), \(t'_1,\ldots,t'_{n'}\) such that
\(t = s'\,t'_1\,\cdots\,t'_{n'}\) and the following:
\begin{center}
\infers[(Tr1-App1)]{
\p s'\,t'_1\,\cdots\,t'_{n'} :\T \tr \TT{br}\,V'\,U'
}{\infers[(Tr1-App1)]{
\p s'\,t'_1\,\cdots\,t'_{{n'}-1} :\top\ra\T\tr \TT{br}\,V'\,U'
}{\infers[(Tr1-App1)]{
\vdots}{\infers[(Tr1-App2)]{
\p s'\,t'_1 :\top\ra\cdots\ra\top\ra\T\tr \TT{br}\,V'\,U'
}{\infers[(Tr1-Set)]{
\p s'\COL \T\ra\top\ra\cdots\ra\top\ra\T\tr V'
}{
\p s'\COL \T\ra\top\ra\cdots\ra\top\ra\T\tr v'
\andalso \cdots
}
\andalso
\infers[(Tr1-Set)]{
\p t'_1\COL\T\tr U'
}{
\p t'_1\COL\T\tr u'
\andalso \cdots
}
}}}}
\end{center}

The head of \(v'\) is not \(\br\) since
if it is \(\br\), by the same reasoning as above we have
some \(s''\) and \(v''\) with 
\[
\p s''\COL \T\ra\top\ra\cdots\ra\top\ra\T\ra\top\ra\cdots\ra\top\ra\T\tr v''
\]
which contradicts the well-formedness condition on types.
Hence, the head of \(v'\) must be a nullary terminal or a non-terminal.

In the case where the head of \(v'\) is a nullary terminal \(a\),
by \(\p s'\COL \T\ra\top\ra\cdots\ra\top\ra\T\tr v'\),
\(s'\) is a unary non-terminal and hence \(s'=v'=a\) and \({n'}=1\).
Since \(v'\) is a  tree, reduction of \(u\) goes on \(u'\)-side.
Thus there exist \(p'\le p\) and \(\pi'\) such that \(u' \red^{p'} \pi'\) and \(\br\,a\,\pi'=\pi\).

If \(p'>0\), by the induction hypothesis for \(t'_1\), 
there exist \(t''_1\), \(u''_1\), \(\pi''_1\) and \(q''_1<p'\) such that 
\[
t'_1\reds_{\GRAM}t''_1
\qquad
\p t''_1\COL\T \tr u''_1
\qquad
u''_1 \red^{q''_1} \pi''_1 \vsim \pi'\,.
\]
Then, we have
\[
t=\app{a}{t'_1} \reds \app{a}{t''_1}
\qquad
\p \app{a}{t''_1} : \T \tr \br\,a\,u''_1
\qquad
\br\,a\,u''_1 \red^{q''_1} \br\,a\,\pi''_1 \vsim \br\,a\,\pi' = \pi\,.
\]
If \(p'=0\),
\[
t=\app{a}{t'_1} \red^0 \app{a}{t'_1}
\qquad
\p \app{a}{t'_1} : \T \tr \br\,a\,u'
\qquad
\br\,a\,u' \red^{0} \br\,a\,\pi' = \pi\,.
\]

In the case where the head of \(v'\) is a non-terminal,
let \(v' = A_\uty\,U_1\,\cdots\,U_\ell\).
The rule used for
\[
\p s'\COL \T\ra\top\ra\cdots\ra\top\ra\T\tr v'( = A_\uty\,U_1\,\cdots\,U_\ell)
\]
is \rname{Tr1-NT} or \rname{Tr1-App1}; in the latter case, we have:
\[
\InfruleS{}{
\p s''\COL 
\uty_{\ell'+1}\land\cdots\land\uty_{\ell} \ra
\T\ra\top\ra\cdots\ra\top\ra\T\tr A_\uty\,U_1\,\cdots\,U_{\ell'}
\quad(\ell'\le\ell)
\\
\p t''\COL\uty_i\tr U_i \text{ and }\uty_i\neq \T
\text{ (for each $i \in \set{\ell'+1,\dots,\ell}$)} 
}{
\p (s'=) \app{s''}{t''}\COL \T\ra\top\ra\cdots\ra\top\ra\T\tr A_\uty\,U_1\,\cdots\,U_\ell
}
\]
Here if \(\odr{\uty_i}=0\) then \(\ell'=\ell\).
Repeating this reasoning to the function side (i.e., \(s''\)) terminates at the case of \rname{Tr1-NT}.
Thus,
there exist \(m\), \(m'\), \(t''_1,\ldots,t''_{m'}\), \(\ell_0,\ldots,\ell_m\)
such that
\begin{align*}&
m \le m' \qquad \ell_0 = 0 \qquad \ell_m = \ell
\\&
s' = A\,t''_1\,\cdots\,t''_{m'}
\\&
\odr{t''_j} \ge 1 \quad(j\in{1,\ldots,m})
\qquad
\odr{t''_j} = 0 \quad(j\in{m+1,\ldots,m'})
\\&
\p t''_j\COL\uty_i\tr U_i \text{ and }\uty_i\neq \T
\qquad(j \in \set{1,\ldots,m}, i \in \set{\ell_{j-1}+1,\dots,\ell_j})
\\&
\uty=
\uty_{1}\land\cdots\land\uty_{\ell_{1}} \ra \cdots \ra
\uty_{\ell_{m-1}+1}\land\cdots\land\uty_{\ell_{m}} \ra 
\top\ra\cdots\ra\top\ra
\T\ra\top\ra\cdots\ra\top\ra\T
\end{align*}

For the reduction sequence
\(u= \br\,V'\,U' \red^p \pi\),
we can assume that
\(V'=\set{v'}\) for simplicity and that the first reduction of the reduction sequence is 
on \(v'\).
This does not lose generality since we can choose an argument to be reduced arbitrarily.
Suppose that \(v'\) is reduced by a rule
\(A_{\uty}\,x'_1\,\cdots\,x'_\ell 
\Hra v\).
Since this is produced by \rname{Tr1-Rule}, there is a rule 
\(A \,x^1\,\cdots\,x^n
\Hra s\) in \(\GRAM\) such that
\[
\p \lambda x^1.\cdots\lambda x^n. s : \uty \tr
\lambda x'_1.\cdots\lambda x'_\ell. v\,.
\]
Then we have the following derivation tree:
\begin{center}
\infers[(Tr1-Abs1)]{
\p \lambda x^1.\cdots\lambda x^n. s : \uty \tr
\lambda x'_1.\cdots\lambda x'_\ell. v
}{
\infers[(Tr1-Abs1)]{
x^1\COL\uty_1,\ldots,x^1\COL\uty_{\ell_1}
\p \lambda x^2.\cdots\lambda x^n. s : \uty^1 \tr
\lambda x'_{\ell_1+1}.\cdots\lambda x'_\ell. v
}{
\infers[(Tr1-Abs1)]{
\vdots}{
\infers[(Tr1-Abs1)]{
x^1\COL\uty_1,\ldots,
x^1\COL
\uty_{\ell_1},
\ldots,
x^m\COL\uty_{\ell_{m-1}},\ldots,
x^m\COL
\uty_{\ell_{m}}
\p \lambda x^{m+1}.\cdots\lambda x^n. s : \uty^m \tr v
}{
\infers[(Tr1-Abs1)]{
\vdots
}{
\infers[(Tr1-Abs2)]{
x^1\COL\uty_1,\ldots,
x^1\COL
\uty_{\ell_1},
\ldots,
x^m\COL\uty_{\ell_{m-1}},\ldots,
x^m\COL
\uty_{\ell_{m}}
\p \lambda x^{m'+1}.\cdots\lambda x^n. s : \uty^{m'} \tr v = [\Teps/x^{m'+1}_{\T}]v_0
}{
\infers[(Tr1-Abs1)]{
x^1\COL\uty_1,\ldots,
x^1\COL
\uty_{\ell_1},
\ldots,
x^m\COL\uty_{\ell_{m-1}},\ldots,
x^m\COL
\uty_{\ell_{m}},
x^{m'+1}\COL\T
\p \lambda x^{m'+2}.\cdots\lambda x^n. s : \uty^{m'+1} \tr v_0
}{
\infers[(Tr1-Abs1)]{
\vdots
}{
x^1\COL\uty_1,\ldots,
x^1\COL
\uty_{\ell_1},
\ldots,
x^m\COL\uty_{\ell_{m-1}},\ldots,
x^m\COL
\uty_{\ell_{m}},
x^{m'+1}\COL\T
\p s : \T \tr v_0
}}}}}}}}
\end{center}
where 
\(\uty^0 \defe \uty\) and
\(\uty^j\) is the codomain type of \(\uty^{j-1}\)
for each \(j \le n\);
especially,
for each \(j \le m\),
\[
\uty^j \defe
\uty_{\ell_{j}+1}\land\cdots\land\uty_{\ell_{j+1}} \ra \cdots \ra
\uty_{\ell_{m-1}+1}\land\cdots\land\uty_{\ell_{m}} \ra 
\top\ra\cdots\ra\top\ra
\T\ra\top\ra\cdots\ra\top\ra\T
\]
and
\[
\uty^{m'} = \T\ra\top\ra\cdots\ra\top\ra\T\,.
\]
Thus we find that
\[
x^j_{\uty_i} = x'_i
\qquad(j \le m, i \in \set{\ell_{j-1}+1,\ldots,\ell_{j}}).
\]

Now 
\[
v' = A_\uty\,U_1\,\cdots\,U_\ell
\red [U_i/x'_i]_{\ind{i}{\ell}} v
= [\Teps/x^{m'+1}_{\T}] [U_i/x'_i]_{\ind{i}{\ell}} v_0
\]
and hence
\begin{equation}\label{eq:remainreduction}
u= \br\,v'\,U' 
\red 
\br\,
(
[\Teps/x^{m'+1}_{\T}] [U_i/x'_i]_{\ind{i}{\ell}} v_0)
\,U'
\red^{p-1} \pi
\end{equation}
while
\[
t = s'\,t'_1\,\cdots\,t'_{n'}
= A\,t''_1\,\cdots\,t''_{m'}\,t'_1\,\cdots\,t'_{n'}
\red [t'_j/x^{m'+j}]_{\ind{j}{n'}} [t''_j/x^{j}]_{\ind{j}{m'}} s\,.
\]
Recall that \(\p t'_1\COL\T\tr U'\);
hence by Lemma~\ref{lem:substOld}, we have
\[
\p [t'_j/x^{m'+j}]_{\ind{j}{n'}} [t''_j/x^{j}]_{\ind{j}{m'}} s
: \T \tr
[U'/x^{m'+1}_{\T}][U_i/x'_i]_{\ind{i}{\ell}} v_0\,.
\]

Thus we define
\[
t' \defe [t'_j/x^{m'+j}]_{\ind{j}{n'}} [t''_j/x^{j}]_{\ind{j}{m'}} s
\qquad
u' \defe [U'/x^{m'+1}_{\T}][U_i/x'_i]_{\ind{i}{\ell}} v_0\,.
\]
Now by Lemma~\ref{lem:substOld}
\[
x^{m'+1}\COL\T \p [t''_j/x^{j}]_{\ind{j}{m}} s : \T \tr
[U_i/x'_i]_{\ind{i}{\ell}} v_0
\]
and hence by the key lemma (Lemma~\ref{lem:groundsbst}),
\begin{align*}
u' = [U'/x^{m'+1}_{\T}] ([U_i/x'_i]_{\ind{i}{\ell}} v_0)
\se&\ 
([\Teps/x^{m'+1}_{\T}] ([U_i/x'_i]_{\ind{i}{\ell}} v_0)) \ibr U'
\end{align*}
but from~\eqref{eq:remainreduction} furthermore we have
\[
u' = [U'/x^{m'+1}_{\T}] ([U_i/x'_i]_{\ind{i}{\ell}} v_0)
\red^{p-1} \pi' \vsim \pi
\]
for some \(\pi'\).

Case where the head of \(u\) is non-terminal:
This case is similar
to the above case where the head of \(v'\) is a non-terminal
(and easier in the sense that we do not need the key lemma):
replace \(v'\), \(s'\), \(\T\ra\top\ra\cdots\ra\top\ra\T\)
with
\(u\), \(t\), \(\T\), respectively.
\end{proof}

\section{Proof of Theorem~\ref{th:tr2-correctness}}
\label{sec:proof-step2}
In this section, we sometimes abbreviate a sequence 
\(t_{1,1}\,\cdots\,t_{1,m_1}\,\cdots\,t_{k,1}\,\cdots\,t_{k,m_k}\)
to \(\mylongseq{t_{i,j}}{i\in \set{1,\ldots,k},j\in \set{1,\ldots,m_i}}\).
We also write \([t_i/x_i]_{i\in \set{1,\ldots,k}}\) for \([t_1/x_1,\ldots,t_k/x_k]\).

We first prepare some lemmas.
First, we show that the transformation preserves typing.
We extend \(\ttytosty{\cdot}\) to the operation on type environments by:
\(\ttytosty{\TTE} = \set{x_{\tty}\COL \ttytosty{\tty}\mid x\COL\tty\in \TTE}\). 
\begin{lemma}
\label{lem:tr2-preserves-types}
If \(\TTE\p t:\tty\tr u\), then 
\(\ttytosty{\TTE}\p u:\ttytosty{\tty}\) holds.
\end{lemma}
\begin{proof}
This follows by straightforward induction on the derivation of \(\TTE\p t:\tty\tr u\).
\end{proof}

\begin{definition}
The relations \(u \subterm u'\) 
and \(U\subterm U'\) on (extended) terms and term sets are defined by:
\infrule{\mbox{$h$ is a variable, a non-terminal, or a terminal}}{h \subterm h}
\infrule{u\subterm u'\andalso U\subterm U'}
    {u U \subterm u'U'}
\infrule{\forall v\in U.\exists v'\in U'.v\subterm v'}
    {U \subterm U'}
\infrule{u\subterm u'}{\lambda x.u \subterm \lambda x.u'}
\end{definition}
It is clear that the two relations \(\subterm\) are pre-order relations.
\begin{lemma}[de-substitution]
\label{lem:desubstitution-tr2}
Let \(t\) be an applicative term.
If \(\TTE \p [s/x]t:\tty\tr u\),
then 
\[
\begin{array}{l}
 \TTE,x\COL\tty_1,\ldots,x\COL\tty_k \p t:\tty \tr u'\\
u \subterm \in [U_1/x_{\tty_1},\ldots,U_k/x_{\tty_k}]u'\\
\TTE\p s:\tty_i \tr U_i \text{ for each \(i\in\set{1,\ldots,k}\)}
\end{array}
\]
for some \(\tty_1,\ldots,\tty_k,U_1,\ldots,U_k,u'\).
Similarly, if \(\TTE \p [s/x]t:\tty\tr U\),
then 
\[
\begin{array}{l}
 \TTE,x\COL\tty_1,\ldots,x\COL\tty_k \p t:\tty \tr U'\\
U \subterm  [U_1/x_{\tty_1},\ldots,U_k/x_{\tty_k}]U'\\
\TTE\p s:\tty_i \tr U_i \text{ for each \(i\in\set{1,\ldots,k}\)}
\end{array}
\]
for some \(\tty_1,\ldots,\tty_k,U_1,\ldots,U_k,U'\).
\end{lemma}
\begin{proof}
The proof proceeds by simultaneous induction on the structure of \(t\).
We first show that the latter property follows from the former one for the same \(t\).
Suppose \(\TTE \p [s/x]t:\tty\tr \set{u_1,\ldots,u_\ell}\),
i.e., \(\TTE \p [s/x]t:\tty\tr u_j\) for each \(j\in\set{1,\ldots,\ell}\).
By the former property, we have:
\[
\begin{array}{l}
\TTE,x\COL\tty_{j,1},\ldots,x\COL\tty_{j,k_j} \p t:\tty \tr u_j'  \\
u_j \subterm\in [U_{j,1}/x_{\tty_{j,1}},\ldots,U_{j,k_j}/x_{\tty_{j,k_j}}]u'_j \\
\TTE \p s:\tty_{j,i}\tr U_{j,i} \mbox{ for each $i\in\set{1,\ldots,k_j}$}
\end{array}
\]
for each $j\in\set{1,\ldots,\ell}$.
Let
\(\set{\tty_1,\ldots,\tty_k} = \set{\tty_{j,i}\mid j\in\set{1,\ldots,\ell}, i\in \set{1,\ldots,k_j}}\),
and:
\[
U_i = \bigcup_{j,i'} \set{U_{j,i'} \mid \tty_{j,i'} = \tty_i}
\]
for each \(i\in\set{1,\ldots,k}\).
Then, the required result holds for \(U' = \set{u'_1,\ldots,u'_\ell}\).

Now, we show the former property (using the latter property for strict subterms of \(t\)).
\begin{itemize}
\item Case \(t=x\):
In this case, we have \(\TTE\p s:\tty \tr u\).
The result holds for \(k=1, \tty_1=\tty, U_1=\set{u}\), and \(u' = x_{\tty}\).
\item Case \(t=y\neq x\): We have \(u=y_\tty\). 
The result holds for \(k=0\) and \(u'=y_{\tty}\).
\item Case where \(t\) is a non-terminal or a constant: similar to the previous case.
\item Case \(t=\TT{br}\,t_0\,t_1\).
In this case, we have:
\[
\begin{array}{l}
\TTE \p [s/x]t_j: \tty'_j\tr u_j \mbox{ for each $j\in\set{0,1}$}\\
(u,\tty) = 
\left\{
  \begin{array}{ll}
    (\TT{br}\,u_0\,u_1, \Tplus) & \mbox{ if \(\tty'_0=\tty'_1=\Tplus\)}\\
    (u_j, \Tplus) & \mbox{ if \(\tty'_j=\Tplus\) and \(\tty'_{1-j}=\Tempty\)}\\
    (\Teps, \Tempty) & \mbox{ if \(\tty'_0=\tty'_1=\Tempty\)}\\
  \end{array}\right.
\end{array}
\]
By applying the induction hypothesis to 
\(\TTE \p [s/x]t_j: \tty'_j\tr u_j\), we obtain:
\[
\begin{array}{l}
\TTE, x:\tty_{j,1},\ldots,x:\tty_{j,k_j}\p t_j:\tty'_j\tr u_j'\\
u_j'\subterm \in [U_{j,1}/x_{\tty_{j,1}},\ldots,U_{j,k_j}/x_{\tty_{j,k_j}}]u_j'\\
\TTE \p s:\tty_{j,i}\tr U_{j,i} \mbox{ for $j\in \set{0,1}, i\in \set{1,\ldots,k_j}$}
\end{array}
\]
Let \(\set{\tty_1,\ldots,\tty_k} = \set{\tty_{j,i}\mid j\in\set{0,1}, i\in \set{1,\ldots,k_j}}\),
and:
\[
U_i = \bigcup_{j,i'} \set{U_{j,i'} \mid \tty_{j,i'} = \tty_i}
\]
for each \(i\in\set{1,\ldots,k}\).
Then, we have the required result for 
\[u' = 
\left\{
  \begin{array}{ll}
    (\TT{br}\,u'_0\,u'_1, \Tplus) & \mbox{ if \(\tty'_0=\tty'_1=\Tplus\)}\\
    (u'_j, \Tplus) & \mbox{ if \(\tty'_j=\Tplus\) and \(\tty'_{1-j}=\Tempty\)}\\
    (\Teps, \Tempty) & \mbox{ if \(\tty'_0=\tty'_1=\Tempty\)}\\
  \end{array}\right.
\]
\item Case \(t=t_0t_1\), where the head symbol of \(t_0\) is not \(\TT{br}\).
In this case, we have:
\[
\begin{array}{l}
u = u_0 V_1\cdots V_\ell\\
\TTE \p [s/x]t_0:\tty'_1\land \cdots \land\tty'_\ell\to \tty \tr u_0\\
\TTE \p [s/x]t_1:\tty'_j \tr V_j \mbox{ for each $j\in\set{1,\ldots,\ell}$}
\end{array}
\]
By applying the induction hypothesis, we obtain
\[
\begin{array}{l}
\TTE, x\COL{\tty_{0,1}},\ldots, x\COL{\tty_{0,k_0}}\p t_0:
\tty'_1\land \cdots \land\tty'_\ell\to \tty \tr u_0'\\
u_0\subterm\in [U_{0,1}/x_{\tty_{0,1}},\ldots,U_{0,k_0}/x_{\tty_{0,k_0}}]u_0'\\
\TTE \p s:\tty_{0,i}\tr U_{0,i} \mbox{ for each $i\in\set{1,\ldots,k_0}$}\\
\end{array}
\]
and
\[
\begin{array}{l}
\TTE, x\COL{\tty_{j,1}},\ldots, x\COL{\tty_{j,k_j}}\p t_1:
  \tty'_j\tr V_j'\\
V_j\subterm [U_{j,1}/x_{\tty_{j,1}},\ldots,U_{j,k_j}/x_{\tty_{j,k_j}}]V_j'\\
\TTE \p s:\tty_{j,i}\tr U_{j,i}\mbox{ for each $i\in\set{1,\ldots,k_j}$}
\end{array}
\]
for each $j\in\set{1,\ldots,\ell}$.
Let \(\set{\tty_1,\ldots,\tty_k} = 
\set{\tty_{j,i}\mid j\in\set{0,\ldots,\ell}, i\in \set{1,\ldots,k_j}}\),
and
\[
U_i = \bigcup_{j,i'} \set{U_{j,i'} \mid \tty_{j,i'} = \tty_i}
\]
for each \(i\in\set{1,\ldots,k}\).
Then, we have the required result for 
\(u' = u_0'V_1'\cdots V_\ell'\).
\end{itemize}
\end{proof}
Now we prove that the transformation is a left-to-right backward simulation.
\begin{lemma}[subject expansion]
\label{lem:tr2-g''-simulated-by-g}
Suppose \(\GRAM\tr \GRAM'\).
If \(s\redwith{\GRAM} s'\) with
\(\p s':\tty\tr u'\) and \(\tty\DCOL\T\),
then there exist \(u\) and \(u''\) such that 
\(u\redswith{\GRAM'} u''\) with
\(\p s:\tty\tr u\) and \(u'\subterm u''\).
\end{lemma}
\begin{proof}
This follows by induction on the structure of \(s\),
with case analysis on the head symbol of \(s\).
\begin{itemize}
\item Case where \(s\) is of the form \(A\,s_1\,\cdots\,s_k\):
In this case, we have:
\[
\begin{array}{l}
A\,x_1\,\cdots\,x_k \to t \in \GRAM\\
s' = [s_1/x_1,\ldots,s_k/x_k]t
\end{array}
\]
By Lemma~\ref{lem:desubstitution-tr2} and \(\p [s_1/x_1,\ldots,s_k/x_k]t:\tty\tr u'\),
we have:
\[
\begin{array}{l}
x_1\COL \IT_{j=1,\ldots,m_1}\tty_{1,j}, \ldots,x_k\COL \IT_{j=1,\ldots,m_k}\tty_{k,j} \p 
 t:\tty \tr v\\
u' \subterm u'' \in [U_{i,j}/(x_{i})_{\tty_{i,j}}]_{i\in\set{1,\ldots,k},j\in\set{1,\ldots,m_i}} v\\
\p s_i: \tty_{i,j}\tr U_{i,j}
\end{array}
\]
By the first condition, we have:
\[
\begin{array}{l}
\p (A\,x_1\,\cdots\,x_k \to t) \tr (A_{\tty'}\,
      \mylongseq{(x_{i})_{\tty_{i,j}}}{i\in\set{1,\ldots,k},j\in\set{1,\ldots,m_i}} \to v)
\\
\p A : \tty' \tr A_{\tty'}
\end{array}
\]
where \(\tty' = \IT_{j=1,\ldots,m_1}\tty_{1,j}\to\cdots\to 
\IT_{j=1,\ldots,m_k}\tty_{k,j} \to \tty\).
Thus, the required result holds for
\(u = A_{\tty'}\,
\mylongseq{U_{i,j}}{i\in\set{1,\ldots,k},j\in\set{1,\ldots,m_i}}\)
and the above \(u''\). %
\item Case where \(s\) is of the form \(\TT{br}\,s_0\,s_1\):
In this case, we have:
\[
\begin{array}{l}
s' = \TT{br}\,s_0'\,s'_1\\
s_i\redwith{\GRAM} s_i'\qquad s_{1-i}=s'_{1-i}\mbox{ for some \(i\in\set{0,1}\)}\\
\p s'_j\COL \tty_j\tr u'_j \mbox{ for each $j\in\set{0,1}$}\\
\end{array}
\]
By the induction hypothesis, we also have:
\( \p s_i: \tty_i \tr u_i\) with \(u_i\redwith{} u''_i\) and \(u'_i\subterm u''_i\).
Let \(u_{1-i} = u'_{1-i}=u''_{1-i}\).
The required condition holds for 
\[ (u,u'') = \left\{
  \begin{array}{ll}
    (\TT{br}\,u_0\,u_1, \TT{br}\,u''_0\,u''_1) & \mbox{ if \(\tty_0=\tty_1=\Tplus\)}\\
    (u_j,u_j'') & \mbox{ if \(\tty_j=\Tplus\) and \(\tty_{1-j}=\Tempty\)}\\
    (\Teps,\Teps) & \mbox{ if \(\tty_0=\tty_1=\Tempty\)}\\
  \end{array}\right.
\]
\end{itemize}
\end{proof}
\begin{lemma}
\label{lem:subterm-sim}
For applicative terms \(u\) and \(v\), if \(u\subterm v\) and \(u\redwith{\GRAM'} u'\), then
\(v\redswith{\GRAM'} v'\) and \(u'\subterm v'\) for some \(v'\).
\end{lemma}
\begin{proof}
This follows by straightforward induction on term \(u\)
and case analysis on the shape of \(u\).
\begin{itemize}
\item
Case \(u=A\,U_1\,\cdots\,U_k\): In this case \(v=A\,V_1\,\cdots\,V_k\) and
\(U_i \subterm V_i\) for each \(i\in\set{1,\ldots,k}\).
Since \(u=A\,U_1\,\cdots\,U_k \redwith{\GRAM'} u'\), we have
\((A\,x_1\,\cdots\,x_k \red u_0) \in \GRAM'\) such that 
\(u' = [U_i/x_i]_{i\in\set{1,\ldots,k}}u_0\).
We can show that
\(U \subterm V\) and \(u \subterm v\) implies \([U/x]u \subterm [V/x]v\), by induction on \(u\subterm v\).
Hence, as required,
\[
\begin{array}{l}
v=A\,V_1\,\cdots\,V_k \redwith{\GRAM'} [V_i/x_i]_{i\in\set{1,\ldots,k}}u_0
\\ {}
[U_i/x_i]_{i\in\set{1,\ldots,k}}u_0 \subterm [V_i/x_i]_{i\in\set{1,\ldots,k}}u_0.
\end{array}
\]
\item
Case \(u=\br\,U_1\,U_2\):
We consider only the case where \(U_1=\set{u_1}\) is reduced; the other cases are similar or clear.
Thus, \(u_1 \redwith{\GRAM'} u'_1\) and \(u'=\br\,u'_1\,U_2\) for some \(u'_1\).
Since \(u \subterm v\), \(v=\br\,V_1\,V_2\) for some \(V_1\) and \(V_2\) such that
\(U_i \subterm V_i\) for \(i\in\{1,2\}\).
As \(U_1 \subterm V_1\), \(u_1 \subterm v_1\) for some \(v_1 \in V_1\).
Hence by the induction hypothesis, there exists \(v'_1\) such that
\(v_1 \redwith{\GRAM'} v'_1\) and \(u'_1 \subterm v'_1\).
Then the required result holds for \(v' = \br\,v'_1\,V_2\).
\end{itemize}
\end{proof}
\begin{lemma}
\label{lem:tr2-ground-tree-fwd}
Suppose \(\GRAM\tr \GRAM'\).
{For any \(\pi\), there exist \(\tty\DCOL\T\) and \(\pi'\) such that
\(\p \pi:\tty\tr \pi'\) and
\begin{align*}
\leaves(\pi')&=\remeps{\leaves(\pi)} \mspace{-100mu} & \tty&=\Tplus \mspace{-100mu}&& (\text{if \(\remeps{\leaves(\pi)} \neq \empword\)})
\\
\pi'&=\Te & \tty&=\Tempty \mspace{-100mu}&& (\text{if \(\remeps{\leaves(\pi)} = \empword\)}).
\end{align*}}
\end{lemma}
\begin{proof}
This follows by induction on \(\pi\):
\begin{itemize}
\item Case \(\pi=a\):
The result follows for \(\pi'\defe a\) and \(\tty\defe \Tplus\).
\item Case \(\pi=\Te\):
The result follows for \(\pi'\defe \Te\) and \(\tty\defe \Tempty\).
\item Case \(\pi = \TT{br}\,\pi_0\,\pi_1\):
By induction hypothesis, we have \(\p \pi_i:\tty_i\tr \pi'_i\) such that
\begin{align*}
\leaves(\pi_i')&=\remeps{\leaves(\pi_i)} \mspace{-100mu} & \tty_i&=\Tplus \mspace{-100mu}&& (\text{if \(\remeps{\leaves(\pi_i)} \neq \empword\)})
\\
\pi_i'&=\Te & \tty_i&=\Tempty \mspace{-100mu}&& (\text{if \(\remeps{\leaves(\pi_i)} = \empword\)})
\end{align*}
for each \(i=0,1\).
The result follows for
\[
  (\pi',\tty) \defe \left\{
  \begin{array}{ll}
    (\TT{br}\,\pi'_0\,\pi'_1, \Tplus) & \mbox{ if \(\tty_0=\tty_1=\Tplus\)}\\
    (\pi'_i, \Tplus) & \mbox{ if \(\tty_i=\Tplus\) and \(\tty_{1-i}=\Tempty\)}\\
    (\Teps, \Tempty) & \mbox{ if \(\tty_0=\tty_1=\Tempty\).}\\
  \end{array}\right.
\]
\hspace{-15pt}\qedhere
\end{itemize}
\end{proof}

The following lemma shows that the transformation is complete in the sense that
for any tree generated by the source grammar, the tree obtained by removing \(\Te\)
can be generated by the target grammar.
\begin{lemma}[completeness of the transformation]
\label{lem:tr2-completeness}
If \(\GRAM\tr \GRAM'\), then
\(\remeps{\LLang(\GRAM)}\subseteq \LLangE(\GRAM')\).
\end{lemma}
\begin{proof}
Suppose \(w\in\remeps{\LLang(\GRAM)}\), i.e.,
\(w=\remeps{\leaves(\pi)}\) for some 
tree \(\pi\) such that \(S\redswith{\GRAM}\pi\).
By Lemma~\ref{lem:tr2-ground-tree-fwd}, 
there exist \(\tty\DCOL\T\) and \(\pi'\) such that
\(\p \pi:\tty\tr \pi'\) and
\begin{align*}
\leaves(\pi')&=\remeps{\leaves(\pi)} \mspace{-200mu}&& (\text{if \(\remeps{\leaves(\pi)} \neq \empword\)})
\\
\pi'&=\Te  \mspace{-100mu}&& (\text{if \(\remeps{\leaves(\pi)} = \empword\)}).
\end{align*}
By repeated applications of Lemmas~\ref{lem:tr2-g''-simulated-by-g}
and \ref{lem:subterm-sim}, 
there exists \(u\) such that \(S_{\tty}\redswith{\GRAM'} u\) and \(\pi'\subterm u\).
We can easily show that for any \(\pi\) and \(u\) if \(\pi \subterm u\) then \(u \redswith{\GRAM'} \pi\) by induction on \(\pi \subterm u\);
thus \(S' \redwith{\GRAM'} S_{\tty}\redswith{\GRAM'} u \redswith{\GRAM'} \pi'\).
If \(w = \remeps{\leaves(\pi)} \neq \empword\), then
\((\Te \neq)\, w=\leaves(\pi')\in \LLangE(\GRAM')\) as required.
If \(w = \remeps{\leaves(\pi)} = \empword\), then
\(\Te=\leaves(\pi') \in \LLang(\GRAM')\) and hence
\(w=\empword \in \LLangE(\GRAM')\) as required.
\end{proof}

We now turn to prove the soundness of the transformation
(Lemma~\ref{lem:tr2-soundness} below).
We again prepare several lemmas.
\begin{lemma}[substitution]
\label{lem:tr2:substitution}
Suppose \(x\not\in\dom(\TTE)\).
If \(\TTE,x\COL\tty_1,\ldots,x\COL\tty_k \p t:\tty\tr u\)
and \(\TTE\p s:\tty_i \tr U_i\) for each \(i\in\set{1,\ldots,k}\), then
\(\TTE \p [s/x]t:\tty\tr [U_1/x_{\tty_1},\ldots,U_k/x_{\tty_k}]u\).
\end{lemma}
\begin{proof}
This follows by induction on the derivation of \(\TTE,x\COL\tty_1,\ldots,x\COL\tty_k \p t:\tty\tr u\),
with case analysis on the last rule used. We discuss only the case for \rname{Tr2-Var};
the other cases follow immediately from the induction hypothesis.
\begin{itemize}
\item Case \rname{Tr2-Var}: In this case,
\[
\begin{array}{l}
\TTE,x\COL\tty_1,\ldots,x\COL\tty_k = \TTE', y\COL\tty\\
t=y\qquad
u = y_{\tty}
\end{array}
\]
If \(x\neq y\), then \([U_1/x_{\tty_1},\ldots,U_k/x_{\tty_k}]u = \set{u}\).
\(\TTE \p [s/x]y:\tty\tr u\) follows immediately from 
\(\TTE,x\COL\tty_1,\ldots,x\COL\tty_k \p t:\tty\tr u\), as \([s/x]y=y=t\)
by the (derived) strengthening rule.

If \(x=y\), then \(\tty=\tty_i\) for some \(i\).
We have \([U_1/x_{\tty_1},\ldots,U_k/x_{\tty_k}]u = U_i\). 
The required result \(\TTE \p [s/x]t:\tty\tr U_i\) follows from 
\([s/x]t=s\) and \(\TTE\p s:\tty_i \tr U_i\).
\item Cases \rname{Tr2-Const0}, \rname{Tr2-Const1}, and \rname{Tr2-NT}: Similar to the previous case where \(x\neq y\).
\item Cases \rname{Tr2-Const2} and \rname{Tr2-App}: These cases are straightforward from induction hypotheses;
we consider only the case \rname{Tr2-App}.
We have:
\[
\begin{array}{l}
t = t_0\,t_1 \qquad u=u_0\,V_1\,\cdots\,V_\ell \\
\TTE,x\COL\tty_1,\ldots,x\COL\tty_k \p t_0: \tty_1\land \cdots\land \tty_\ell\to \tty\tr u_0\\
\TTE,x\COL\tty_1,\ldots,x\COL\tty_k \p t_1: \tty_j\tr V_j \text{ for each \(j\in\set{1,\ldots,\ell}\)}
\end{array}
\]
By the induction hypotheses and \rname{Tr2-Set}, we have:
\[
\begin{array}{l}
\TTE \p [s/x]t_0: \tty_1\land \cdots\land \tty_\ell\to \tty \tr [U_1/x_{\tty_1},\ldots,U_k/x_{\tty_k}]u_0
\\
\TTE \p [s/x]t_1: \tty_j \tr [U_1/x_{\tty_1},\ldots,U_k/x_{\tty_k}]V_j
\end{array}
\]
Hence, we have the required result by \rname{Tr2-App}.
\end{itemize}
\end{proof}
Next we show that the transformation is a right-to-left forward simulation.
\begin{lemma}[subject reduction]
\label{lem:tr2:fwdsim-right2left}
Suppose \(\GRAM\tr \GRAM'\).
If \(u\redwith{\GRAM'} u'\) with
\(\p t:\tty\tr u\) and \(\tty\DCOL\T\),
then there exists \(t'\) such that 
\(t\redswith{\GRAM} t'\) with
\(\p t':\tty\tr u'\).
\end{lemma}
\begin{proof}
This follows by induction on the derivation of
\(\p t:\tty\tr u\).
We perform case analysis on the shape of \(t\).
\begin{itemize}
\item Case \(t=\TT{br}\,t_0\,t_1\):
We first consider the case where
\(t_0\) or \(t_1\) has type \(\Tempty\). Since \(u\redwith{\GRAM'}u'\), we have:
\[
\begin{array}{l}
t = \TT{br}\,t_0\,t_1\\
\tty = \tty_i=\Tplus\qquad \tty_{1-i}=\Tempty\\
\p t_i:\Tplus \tr u
\end{array}
\]
for some \(i\in\set{0,1}\).
By the induction hypothesis, there exists \(t_i'\) such that
\(t_i\redswith{\GRAM}t_i'\) and \(\p t_i':\Tplus \tr u'\).
The required result holds for \(\TT{br}\,t_0'\,t_1'\), where
\(t_{1-i}' \defe t_{1-i}\).

In the other case, for some \(i\in\set{0,1}\), we have:
\[\begin{array}{l}
u=\TT{br}\,u_0\,u_1\\
u_i \redwith{\GRAM'}u_i'\qquad u'_{1-i}=u_{1-i}\\
u' = \TT{br}\,u_0'\,u'_1\\
\p t_j\COL\Tplus \tr u_j\mbox{ for each $j\in\set{0,1}$}
\end{array}
\]
By the induction hypothesis, there exists \(t_i'\) such that
\(\empty\p t_i'\COL\Tplus\tr u_i'\) and \(t_i\redswith{\GRAM}t_i'\).
Let \(t'_{1-i}=t_{1-i}\). Then,
the result holds for \(t'=\TT{br}\,t_0'\,t_1'\).

\item \(t = A\,t_1\,\cdots\,t_k\):
In this case, we have:
\[\begin{array}{l}
u = A_{\tty'}\,U_{1,1}\,\cdots\,U_{1,\ell_1}\,\cdots\,U_{k,1}\,\cdots\,U_{k,\ell_k}\\
\p A:\tty' \tr A_{\tty'}\\
\p t_i:\tty_{i,j}\tr U_{i,j} \mbox{ for each $i\in\set{1,\ldots,k},j\in\set{1,\ldots,\ell_i}$}\\
\tty' = \IT_j \tty_{1,j} \to \cdots \to \IT_j\tty_{k,j}\to \tty\\
\big(A_\tty'\,(x_{1})_{\tty_{1,1}}\cdots\,(x_{1})_{\tty_{1,\ell_1}}\cdots\,
(x_{k})_{\tty_{k,1}}\cdots\,(x_{k})_{\tty_{k,\ell_k}}\to u_0\big)\in \GRAM'\\
u' \in [U_{i,j}/(x_{i})_{\tty_{i,j}}]_{i\in\set{1,\ldots,k},j\in\set{1,\ldots,\ell_i}}u_0
\end{array}
\]
By the condition
\(\big(A_\tty'\,(x_{1})_{\tty_{1,1}}\cdots\,(x_{1})_{\tty_{1,\ell_1}}\cdots\,
(x_{k})_{\tty_{k,1}}\cdots\,(x_{k})_{\tty_{k,\ell_k}}\to u_0\big)\in \GRAM'\),
we have
\[
\begin{array}{l}
(A\,x_1\,\cdots\,x_k\to t_0) \in \GRAM\\
x_1\COL\tty_{1,1},\ldots,x_1\COL\tty_{1,\ell_1},\ldots,
x_k\COL\tty_{k,1},\ldots,x_k\COL\tty_{k,\ell_k}\p t_0:\tty \tr u_0.
\end{array}
\]
Let \(t'=[t_1/x_1,\ldots,t_k/x_k]t_0\).
By Lemma~\ref{lem:tr2:substitution},
we have \(\p t':\tty\tr u'\) and
\(t\redswith{\GRAM}t'\) as required.
\end{itemize}
\end{proof}

\begin{lemma}
\label{lem:tr2:weak-normalization}
If \(\p t:\Tempty \tr u\), then \(t\redswith{\GRAM} \pi\) for some \(\pi\) such that
\(\leaves(\pi)\in\Te^*\).
\end{lemma}
\begin{proof}
We define the unary logical relation \(\REL_\tty\) by:
\begin{itemize}
\item \(\REL_{\Tempty}(t)\) if \(t\redswith{\GRAM} \pi\) for some \(\pi\) such that \(\leaves(\pi)\in\Te^*\).
\item \(\REL_{\Tplus}(t)\) if \(t\redswith{\GRAM} \pi\) for some \(\pi\) such that 
\(\leaves(\pi)\not\in\Te^*\).
\item \(\REL_{\tty_1\land\cdots\land\tty_k\to \tty}(t)\) if,
whenever \(\REL_{\tty_i}(s)\) for every \(i\in\set{1,\ldots,k}\), \(\REL_\tty(ts)\).
\end{itemize}
For a substitution \(\theta = [s_1/x_1,\ldots,s_\ell/x_\ell]\),
we write \(\theta \models \TTE\) if
\(\TTE = \set{x_i\COL \tty_{i,j} \mid i\in\set{1,\ldots,\ell},j\in\set{1,\ldots,m_i}}\)
and \(\forall i\in\set{1,\ldots,\ell}.\,\forall j\in\set{1,\ldots,m_i}.\,\REL_{\tty_{i,j}}(s_i)\).
We write \(\TTE\models t:\tty\) 
when \(\theta\models \TTE\) implies \(\REL_{\tty}(\theta t)\).
We show that, for every applicative term \(t\),
\(\TTE \p t:\tty \tr u\) implies \(\TTE \models t:\tty\), from which the result follows.
The proof proceeds by induction on the derivation of \(\TTE \p t:\tty \tr u\), with case analysis on
the last rule used.
\begin{itemize}
\item Cases for \rname{Tr2-Var}, \rname{Tr2-Const0}, and \rname{Tr2-Const1}: trivial.
\item Case for \rname{Tr2-Const2}:
  In this case, we have:
\[
t= \br\,t_0\,t_1
\qquad
\TTE \p t_i :\tty_i \tr u_i
\qquad
\tty_i = 
\begin{cases}
\Tplus & (\exists i.\, \tty_i=\Tplus)
\\
\Tempty & (\forall i.\, \tty_i=\Tempty)
\end{cases}
\]
Then, \(\TTE \models \br\,t_0\,t_1 :\tty\) follows from the induction hypotheses
\(\TTE \models t_i : \tty_i\).
\item Case for \rname{Tr2-NT}:
  In this case, we have:
\[ 
\begin{array}{l}
t = A\qquad \tty = \IT_{j\in\set{1,\ldots,\ell_1}} \tty_{1,j}\to \cdots \to 
\IT_{j\in\set{1,\ldots,\ell_k}} \tty_{k,j}\to \tty_0
\qquad \tty_{0} = \Tplus \text{ or } \Tempty \\
A\,x_1\,\cdots\,x_k\to t_0 \in \GRAM \\
\set{x_i\COL \tty_{i,j}\mid i\in\set{1,\ldots,k},j\in\set{1,\ldots,\ell_i}}\p t_0:\tty_0 \tr u
\end{array}
\]
To show \(\REL_{\tty}(A)\),
suppose \(\REL_{\tty_{i,j}}(s_i)\) for every \(i\in\set{1,\ldots,k}\) and \(j\in\set{1,\ldots,\ell_i}\).
By applying the induction hypothesis to the last condition, we have
\[\set{x_i\COL \tty_{i,j}\mid i\in\set{1,\ldots,k},j\in\set{1,\ldots,\ell_i}}\models t_0:\tty_0.\]
(Note that 
the derivation for
\(\set{x_i\COL \tty_{i,j}\mid i\in\set{1,\ldots,k},j\in\set{1,\ldots,\ell_i}}\p t_0:\tty_0 \tr u\)
is a sub-derivation for \(\p A:\tty\tr u\); this is the reason why
we have included the rightmost premise in the rule \rname{Tr2-NT}.)
Therefore, we have 
\(\REL_{\tty_0}([s_1/x_1,\ldots,s_k/x_k]t_0)\), which implies
\(\REL_{\tty_0}(A\,s_1\,\cdots\,s_k)\) as required.
\item Case for \rname{Tr2-App}:
In this case, we have:
\[
\begin{array}{l}
t = t_0\,t_1\\
\TTE \p t_0: \tty_1\land \cdots\land \tty_k\to \tty\tr u_0\\
\TTE \p t_1: \tty_i\tr U_i
\end{array}
\]
Suppose \(\theta \models \TTE\). We need to show \(\REL_{\tty}(\theta t)\).
By applying the induction hypothesis to the transformation judgments for \(t_0\) and \(t_1\), we have
\(\TTE\models t_0:\tty_1\land \cdots\land \tty_k\to \tty\)
and \(\TTE\models t_1:\tty_i\) for every \(i\in\set{1,\ldots,k}\).
Thus, we have also \(\REL_{\tty_1\land \cdots\land \tty_k\to \tty}(\theta t_0)\)
and \(\REL_{\tty_i}(\theta t_1)\). Therefore, we have \(\REL_{\tty}(\theta t)\) as required.
\end{itemize}
\end{proof}

\begin{lemma}
\label{lem:tr2:fwd-sim-ground-tree}
If \(\p t:\tty\tr \pi'\), {then there exists \(\pi\) such that
\(t\redswith{\GRAM} \pi\) and 
\[
\begin{aligned}
\pi' &\neq \Te \qquad \remeps{\leaves(\pi)}=\leaves(\pi') \quad&& (\text{if \(\tty=\Tplus\)})
\\
\pi' &= \Te \qquad \remeps{\leaves(\pi)}=\empword && (\text{if \(\tty=\Tempty\)}).
\end{aligned}
\]}
\end{lemma}
\begin{proof}
This follows by induction on the derivation of \(\p t:\tty\tr \pi'\).
Since the output of transformation is a tree, 
the last rule used to derive \(\p t:\tty\tr \pi'\) must be 
\rname{Tr2-Const0}, \rname{Tr2-Const1}, or \rname{Tr2-Const2}.
The cases for \rname{Tr2-Const0} and \rname{Tr2-Const1} are trivial.
If the last rule is \rname{Tr2-Const2}, we have:
\(t=\TT{br}\,t_0\,t_1\) with \(\p t_i:\tty_i\tr u_i\) for \(i\in\set{0,1}\).
We perform case analysis on \(\tty_0\) and \(\tty_1\).
\begin{itemize}
\item Case \(\tty_0=\tty_1=\Tplus\):
In this case, \(\pi'=\TT{br}\,u_0\,u_1\) and \(\tty=\Tplus\). 
For each \(i\in\set{0,1}\), by the induction hypothesis
there exists \(\pi_i\) such that
\(t_i\redswith{\GRAM} \pi_i\) %
and \(\remeps{\leaves(\pi_i)}=\leaves(u_i)\).
Thus, we have 
\(t\redswith{\GRAM} \pi\), \(\pi'\neq\Te\), and \(\remeps{\leaves(\pi)}=\leaves(\pi')\) for \(\pi=\TT{br}\,\pi_0\,\pi_1\).
\item Case \(\tty_i=\Tplus\) and \(\tty_{1-i}=\Tempty\) for some \(i\in\set{0,1}\).
In this case, \(\pi'= u_i\) and \(\tty=\Tplus\). 
By the induction hypothesis, there exists \(\pi_i\) such that
\(t_i\redswith{\GRAM} \pi_i\), \(\pi' \neq \Te\), and 
\(\remeps{\leaves(\pi_i)}=\leaves(\pi')\).
By Lemma~\ref{lem:tr2:weak-normalization}, there exists \(\pi_{1-i}\) such that \(t_{1-i}\redswith{\GRAM}
      \pi_{1-i}\) and \(\leaves(\pi_{1-i}) \in \Te^*\).
Thus, we have 
\(t\redswith{\GRAM} \pi\), \(\pi' \neq \Te\), and \(\remeps{\leaves(\pi)}=\leaves(\pi')\) for \(\pi=\TT{br}\,\pi_0\,\pi_1\).
\item Case \(\tty_0=\tty_1=\Tempty\). 
In this case, \(\pi'=\Te\) and \(\tty=\Tempty\).
By Lemma~\ref{lem:tr2:weak-normalization}, for each \(i\in\set{0,1}\)
there exists \(\pi_{i}\) such that \(t_{i}\redswith{\GRAM}
      \pi_{i}\) and \(\leaves(\pi_{i}) \in \Te^*\).
Thus, we have 
\(t\redswith{\GRAM} \pi\), \(\pi'=\Te\), and \(\remeps{\leaves(\pi)}=\empword\) for \(\pi=\TT{br}\,\pi_0\,\pi_1\).
\end{itemize}

\end{proof}

We are now ready to prove soundness of the transformation.
\begin{lemma}[soundness]
\label{lem:tr2-soundness}
If \(\GRAM\tr \GRAM'\), then
\(\remeps{\LLang(\GRAM)}\supseteq \LLangE(\GRAM')\).
\end{lemma}

\begin{proof}
Suppose \(w\in \LLangE(\GRAM')\). 
Then there exist \(\tty\in\set{\Tplus,\Tempty}\) and \(\pi'\) such that \(S_\tty\redswith{\GRAM'} \pi'\) and
\[
\begin{array}{ll}
w=\leaves(\pi') &\quad (\text{if \(\pi' \neq \Te\)})
\\
w=\empword &\quad (\text{if \(\pi' = \Te\)}).
\end{array}
\]
Now \(S_{\tty} \redwith{\GRAM'} u \redswith{\GRAM'} \pi'\) for some \(u\),
and so \((S_{\tty} \red u) \in \GRAM'\).
By \rname{Tr2-Rule}, there exists \(t\) such that \((S \red t) \in \GRAM\) and \(\p t : \tty \tr u\).
Hence, by \rname{Tr2-NT}, we have \(\p S : \tty \tr S_{\tty}\).
By repeated applications of Lemma~\ref{lem:tr2:fwdsim-right2left},
we have \(S\redswith{\GRAM} t\) and \(\p t:\tty\tr \pi'\).
By Lemma~\ref{lem:tr2:fwd-sim-ground-tree},
\(\tty=\Tempty\) iff \(\pi' = \Te\),
and we have \(t\redswith{\GRAM}\pi\) and \(\remeps{\leaves(\pi)}=w\) for some \(\pi\).
Thus, we have \(S\redswith{\GRAM} \pi\) and 
\(\remeps{\leaves(\pi)}=w\) as required.
\end{proof}

\begin{pfof}{Theorem~\ref{th:tr2-correctness}}
The fact that \(\GRAM'\) is a tree grammar of order at most \(n\) follows immediately from Lemma~\ref{lem:tr2-preserves-types}. 
\(\remeps{\LLang(\GRAM)}=\LLangE(\GRAM')\) 
follows immediately from Lemmas~\ref{lem:tr2-completeness} and \ref{lem:tr2-soundness}.
\qed
\end{pfof}

\section{Proof of Theorem~\ref{th:main-converse}}
\label{sec:ProofConverseMainTheorem}

First we give a formal definition of the construction for Theorem~\ref{th:main-converse}.
Let \(\GRAM'=(\TERMS',\NONTERMS',\RULES',S')\) be an order-\(n\) tree grammar 
such that no word in \(\LLangE(\GRAM')\) contains \(\Te\).
We define a grammar \(\GRAM=(\TERMS,\NONTERMS,\RULES,S)\) as:
\begin{align*}
\TERMS\defe&\,\set{a\mapsto 1 \mid \TERMS'(a)=0, a\neq \Te} \cup \set{\Te \mapsto 0}
\\
\NONTERMS\defe&\, \set{A\COL\conv{\kappa} \mid (A\COL\kappa) \in \NONTERMS'}
\cup\set{E \COL \T\to\T,\ \Br \COL (\T\to\T)\to(\T\to\T)\to(\T\to\T),\ S\COL \T}
\\
\RULES\defe&\, \set{A \, x_1 \,\cdots \, x_\ell \, x \to \conv{t} x \mid (A \, x_1 \,\cdots \, x_\ell \to t) \in \RULES'}
\\ \cup&\, 
\set{E\,x\to x,\ \Br\,f\,g\,x\to f(g\, x),\ S \to S' \Te}
\end{align*}
where \(E\), \(\Br\), and \(S\) are fresh non-terminals, and \(\conv{(-)}\) is defined as follows.
\begin{align*}
&\conv{\T}\defe \T\to\T \qquad \conv{(\kappa_1 \to \kappa_2)} \defe \conv{\kappa_1} \to \conv{\kappa_2}
\\
&\conv{x} \defe x \qquad
 \conv{\Te} \defe E \qquad \conv{\br} \defe \Br \qquad
\conv{a} \defe a \ \ (\TERMS'(a)=0, a\neq \Te) 
\\&
\conv{A} \defe A \qquad
\conv{(s\,t)} \defe \conv{s}\,\conv{t} %
\end{align*}

If \(n=0\), then the above \(\GRAM\) is an order-2 grammar.
In this case, any occurrence of \(\br\) in \(\GRAM'\) must be fully applied, 
hence we replace \(\conv{\br} \defe \Br\) in the above definition with
\[
\conv{(\br\,s\,t)} \defe A_{s,t}
\]
and add a fresh non-terminal \(A_{s,t}\COL\T\to\T\) and a rule
\[
A_{s,t}\,x \to \conv{s}\,(\conv{t}\,x)
\]
for each \(s\) and \(t\) such that \(\br\,s\,t\) occurs in the right hand side of some rule of \(\GRAM'\).
We write this modified grammar as \(\GRAM_0\).

\begin{lemma}
For \((A \, x_1 \,\cdots \, x_\ell \to t) \in \RULES'\)
where \(\NONTERMS'(A) = \kappa_1\ra\cdots\ra\kappa_\ell\ra\T\),
\[
\NONTERMS,x_1\COL\conv{\kappa_1},\ldots,x_\ell\COL\conv{\kappa_\ell},x\COL\T \pK \conv{t}\,x:\T.
\]
\end{lemma}
\begin{proof}
By definition we have \(\NONTERMS',x_1\COL\kappa_1,\ldots,x_\ell\COL\kappa_\ell\pK t:\T\),
and we can show by induction on \(t\) that
\(\NONTERMS',x_1\COL\kappa_1,\ldots,x_\ell\COL\kappa_\ell\pK t:{\kappa}\) implies
\(\NONTERMS,x_1\COL\conv{\kappa_1},\ldots,x_\ell\COL\conv{\kappa_\ell} \pK \conv{t}:\conv{{\kappa}}\).
\end{proof}
By the above lemma, we can show that \(\GRAM\) is well-defined as a grammar, and
clearly \(\odr{\conv{\kappa}} = \odr{\kappa}+1\); thus \(\GRAM\) is an order-\((n+1)\) grammar if \(n\ge 1\).
When \(n=0\), similarly we can show that \(\GRAM_0\) is an order-\(1\) grammar.

In the rest of this section, we show \(\Wlang(\GRAM) = \LLangE(\GRAM')\) for \(n\ge 0\); 
note that it is clear that \(\Wlang(\GRAM) = \Wlang(\GRAM_0)\) 
since the above modification for defining \(\GRAM_0\) is just unfolding of the (unique) rule for \(\Br\).

To show the goal, we use a logical relation for cpo semantics (cf. e.g., \cite{DBLP:books/daglib/0018087}
and the proof of Lemma~\ref{lem:context} in Appendix~\ref{sec:app1-basic}).
First, recall from Appendix~\ref{sec:app1-basic} that we can embed a grammar \(\GRAM\) into a \lambday-calculus,
which we write as \(\lambdayg{\GRAM}\),
and here we use binary choice operators \(+_\kappa\COL\kappa \to \kappa \to \kappa\) to represent the finite
nondeterminism;
thus, \(\lambdayg{\GRAM}\) is the \lambday-calculus that has the terminals of \(\GRAM\) and \(+\) as constants.
We extend the definition of \(\conv{(-)}\) to \lambdayg{\GRAM'}:
\begin{align*}
\conv{(\lambda x\COL\kappa .t)} \defe \lambda x\COL\conv{\kappa} .\conv{t}
\qquad
\conv{(Y_\kappa)} \defe Y_{\conv{\kappa}}
\qquad
\conv{(+_\kappa)} \defe +_{\conv{\kappa}}
\end{align*}

For a grammar \(\GRAM\), let \(\vtg{\GRAM}\) be the set of all the trees which consist of terminals of \(\GRAM\),
and we define \(\ip{-}_{\GRAM}\) as:
\begin{align*}
\ip{\T}_\GRAM &\defe (P(\vtg{\GRAM}),\subseteq)
\\
\ip{\kappa_1\to\kappa_2}_\GRAM &\defe \ip{\kappa_1}_\GRAM \cfun \ip{\kappa_2}_\GRAM
\\
\ip{a}_\GRAM(L_1)\cdots(L_n) &\defe \set{a\,\pi_1\,\cdots\,\pi_n \mid \pi_i \in L_i} \in P(\vtg{\GRAM}) 
 \qquad (\TERMS(a)=n,\ L_i \in P(\vtg{\GRAM}))
\end{align*}
where \(A \cfun B\) is the continuous function space from \(A\) to \(B\).
The adequacy theorem, i.e., \(\ip{t}_\GRAM = \set{\pi \mid t \reds \pi}\), holds as usual.
(\(\ip{t}_\GRAM \supseteq \set{\pi \mid t \reds \pi}\) is obtained by showing that
\(t \red t'\) implies \(\ip{t}_\GRAM \supseteq \ip{t'}_\GRAM\) by induction on the derivation of \(t \red t'\).
The converse immediately follows from the abstraction theorem 
{(i.e., \(\ip{t}\,R_{\kappa}\,t\) for any \(t\in\cterm{\kappa}\))} 
of the logical relation \(R\) given in the proof of Lemma~\ref{lem:context} in Appendix~\ref{sec:app1-basic} 
(cf.~\cite[Theorem~4.6]{DBLP:books/daglib/0018087}),
{with a trivial modification for the difference that here we are considering trees up to the equality
rather than up to \(\vsim\)}.)

We consider a logical relation \((\R_\kappa)_\kappa\)
between two models for \(\lambdayg{\GRAM'}\).
One model is the model \(\ip{-}_{\GRAM'}\) and
we define another model \(\ipc{-}\) for \(\lambdayg{\GRAM'}\), a model reflecting the translation \(\conv{(-)}\).
The interpretation of types is given by:
\begin{align*}
\ipc{\T}&\defe\ip{\conv{\T}}_{\GRAM} = P(\vtg{\GRAM}) \cfun P(\vtg{\GRAM})
\\
\ipc{\kappa_1\to\kappa_2}&\defe \ipc{\kappa_1}\cfun\ipc{\kappa_2}.
\end{align*}
For a constant \(c\), the interpretation is given as \(\ipc{c} \defe \ip{\conv{c}}_{\GRAM}\); specifically, for terminals,
\begin{align*}
\ipc{\Te}&\defe \id : P(\vtg{\GRAM}) \cfun P(\vtg{\GRAM})
\\
\ipc{a}&\defe (L \mapsto \set{a\,\pi \mid \pi \in L}) : P(\vtg{\GRAM}) \cfun P(\vtg{\GRAM}) \qquad
 (\TERMS'(a)=0,a\neq \Te)
\\
\ipc{\br}&\defe (f \mapsto (g \mapsto f \circ g)) : \ipc{\T} \cfun \ipc{\T} \cfun \ipc{\T}.
\end{align*}
Then, we can easily show that \(\ip{\conv{t}}_{\GRAM} = \ipc{t}\) for any term \(t\), by induction on \(t\).

The logical relation \((\R_\kappa)_\kappa\) is determined by \(\R_\T\), which is defined as:
\begin{align*}
&\R_\T \subseteq \ip{\T}_{\GRAM'} \times \ipc{\T} = P(\vtg{\GRAM'}) \times (P(\vtg{\GRAM}) \cfun P(\vtg{\GRAM}))
\\
&L'\,\R_\T\,f \text{ if for any \(L \in P(\vtg{\GRAM})\),}
\\&
\remeps{\set{\leaves(\pi) \mid \pi\in L'}} \concat
\set{a_1\cdots a_n \mid a_1(\cdots (a_n\,\TT{e})\cdots) \in L}
\\&=
\set{a_1\cdots a_n \mid a_1(\cdots (a_n\,\TT{e})\cdots) \in f(L)}
\end{align*}
where \(A\concat B\) is the concatenation of word languages \(A\) and \(B\).
If we show the abstraction theorem 
(i.e., \(\ip{t}_{\GRAM'}\,\R_\kappa\,\ipc{t}\) for any closed term \(t\)), 
then we can show our goal as follows:
Since \(\ipc{t} = \ip{\conv{t}}_{\GRAM}\), 
\(\ipc{S'}(\set{\Te})
= \ip{\conv{S'}}_{\GRAM} (\ip{\Te}_{\GRAM}) 
= \ip{S'}_{\GRAM} (\ip{\Te}_{\GRAM}) 
= \ip{S}_{\GRAM} \).
Now we have \(\ip{S'}_{\GRAM'}\,\R_\T\,\ipc{S'}\), which means---with \(L=\set{\Te}\) and the adequacy---that
\(\remeps{\LLang(\GRAM')} \concat \set{\empword} = \Wlang(\GRAM)\).
By the assumption that no word in \(\LLangE(\GRAM')\) contains \(\Te\), we have
\(\remeps{\LLang(\GRAM')}=\LLangE(\GRAM')\).

To show the abstraction theorem, we only have to show
\(\ip{c}_{\GRAM'}\,\R\,\ipc{c}\) for each constant \(c \in \set{Y, \Te, a, \br, +}\).
The cases of \(\Te\), \(a\), and \(+\) are clear.
As usual, the case of \(Y\) follows from showing that
\(\R_\T\) is closed under the least upper bound of increasing chains and contains the bottom, which is clear.
To show the case of \(\br\), suppose that \(L'_1\,\R_\T\,f_1\) and \(L'_2\,\R_\T\,f_2\);
then, for \(L \in P(\vtg{\GRAM})\), we have:
\begin{align*}
&\remeps{\set{\leaves(\pi) \mid \pi\in \ip{\br}_{\GRAM'}(L'_1)(L'_2)}} \concat
\set{a_1\cdots a_n \mid a_1(\cdots (a_n\,\TT{e})\cdots) \in L}
\\ =\ &
\remeps{\set{\leaves(\pi) \mid \pi\in L'_1}} \concat
\remeps{\set{\leaves(\pi) \mid \pi\in L'_2}} \concat
\set{a_1\cdots a_n \mid a_1(\cdots (a_n\,\TT{e})\cdots) \in L}
\\ =\ &
\remeps{\set{\leaves(\pi) \mid \pi\in L'_1}} \concat
\set{a_1\cdots a_n \mid a_1(\cdots (a_n\,\TT{e})\cdots) \in f_2(L)}
\\ =\ &
\set{a_1\cdots a_n \mid a_1(\cdots (a_n\,\TT{e})\cdots) \in f_1(f_2(L))}
\\ =\ &
\set{a_1\cdots a_n \mid a_1(\cdots (a_n\,\TT{e})\cdots) \in \ipc{\br}(f_1)(f_2)(L)}.
\end{align*}

\section{Transformation for the assumption on the order of simple types}
\label{sec:TransAssumpOrder}

Let \(\GRAM=(\TERMS,\NONTERMS,\RULES,S)\) be a given word grammar.
In this section, we treat the non-terminals \(A\) of \(\GRAM\) differently from the variables \(x\).

We first define a transformation of simple types by induction:
\begin{align*}
\up{\T} &\defe \T
\\
\up{(\sty'\to\sty)} &\defe
\begin{cases}
\up{\sty'} \to \up{\sty} & (\sty' \neq \T \lor \odr{\sty} \le 1)
\\
(\T \to \T) \to \up{\sty} & (\sty' = \T \land \odr{\sty} > 1)
\end{cases}
\end{align*}
Note that if \(\odr{\sty} \le 1\) then
\(\up{(\sty)} = \sty\).
We extend this transformation to simple type environments \(\SE\); the extended transformation \(\eup{-}{\sty}\) is
parameterized by a simple type \(\sty\):
\begin{align*}
\eup{\nSE}{\sty} &\defe \nSE
\\
\eup{\cSE{x\COL\sty'}{\SE}}{\sty} &\defe
\begin{cases}
\cSE{x\COL\up{\sty'}}{\eup{\SE}{\sty}} & (\sty' \neq \T \lor \order{(\SE \to \sty)} \le 1)
\\
\cSE{x\COL\T\to\T}{\eup{\SE}{\sty}} & (\sty' = \T \land \order{(\SE \to \sty)} > 1)
\end{cases}
\\
\eup{\cSE{A\COL\sty'}{\SE}}{\sty} &\defe
\cSE{A\COL\up{\sty'}}{\eup{\SE}{\sty}}
\end{align*}
where \(\SE \to \sty\) is defined by induction on \(\SE\):
\(\nSE \to \sty \defe \sty\) and \((\cSE{x\COL\sty'}{\SE})\to\sty\defe \sty'\to(\SE\to\sty)\).

Let \(K\) be a fresh non-terminal of type \(\T\to\T\to\T\).
We define a transformation of derived terms in context (so that the next lemma holds):
\begin{align*}
\pup{\SE', x\COL\sty, \SE \pK x : \sty} \defe&
\begin{cases}
x & (\sty \neq \T \lor \order{(\SE \to \sty)} \le 1)
\\
x\,\Te & (\sty = \T \land \order{(\SE \to \sty)} > 1)
\end{cases}
\\
\pup{\SE', A\COL\sty, \SE \pK A : \sty} \defe&\ A
\\
\pup{\SE \pK a : \T^n\to\T} \defe&\ a
\\
\pup{\SE\pK \app{t_1}{t_2}:\sty}
\defe&
\begin{cases}
\up{t_1}\,\up{t_2}
 & (\sty' \neq \T \lor \odr{\sty}\le 1)
\\
\up{t_1}\,\big(\K\,\up{t_2}\big)
 & (\sty' =\T \land \odr{\sty}  > 1)
\end{cases}
\end{align*}
where suppose that we have \(\SE\pK t_2:\sty'\)
and 
\begin{align*}
\up{t_1} \defe&\ \pup{\SE\pK t_1:\sty'\ra\sty}
\\
\up{t_2} \defe&\ \pup{\SE\pK t_2:\sty'}.
\end{align*}

Now we define the new grammar: \(\GRAM' \defe (\TERMS,\NONTERMS',\RULES',S)\) where
\begin{align*}
\NONTERMS' \defe&\ 
\set{K\COL\T\to\T\to\T} \cup \set{A \COL \up{\sty} \mid \NONTERMS(A) = \sty}
\\
\RULES' \defe&\ 
\set{K\,x\,y \to x} 
\\ \cup&\ 
\set{ A\,x_1\,\cdots\,x_\ell \to \pup{\NONTERMS,x_1\COL\kappa_1,\ldots,x_\ell\COL\kappa_\ell\pK t:\T} 
\mid A\,x_1\,\cdots\,x_\ell \to t \in \RULES}.
\end{align*}
To show that the grammar \(\GRAM'\) is well-defined, we need the following lemma.

\begin{lemma}
Given a derived term in context:
\[
\SE \pK t :\sty
\]
we have
\[
K\COL\T\to\T\to\T, \eup{\SE}{\sty} \pK 
\pup{\SE\pK t:\sty} : \up{\sty}.
\]
\end{lemma}
\begin{proof}
The proof proceeds by induction on \(t\).

The case where \(t\) is a non-terminal or a terminal is clear.
The case \(t=x\) is also clear since for an environment of the form \(\SE', x\COL\sty, \SE\) we have:
\begin{align*}
\eup{\SE', x\COL\sty, \SE}{\sty} =
\begin{cases}
\dots,x\COL\up{\sty},\eup{\SE}{\sty}
& (\sty \neq \T \lor \order{(\SE \to \sty)} \le 1)
\\
\dots,x\COL\T\to\T,\eup{\SE}{\sty}
& (\sty = \T \land \order{(\SE \to \sty)} > 1)
\end{cases}
\end{align*}

In the case where \(t=t_1\,t_2\) and we have \(\SE\pK t_2:\sty'\),
by the induction hypotheses, 
the types of \(\pup{\SE\pK t_1:\sty'\ra\sty}\) and \(\pup{\SE\pK t_2:\sty'}\)
are \(\up{(\sty'\to\sty)}\) and \(\up{\sty'}\), respectively. Since
\begin{align*}
(\up{(\sty'\to\sty)}, \up{\sty'}) =
\begin{cases}
(\up{\sty'} \to \up{\sty}, \up{\sty'})
 & (\sty' \neq \T \lor \odr{\sty}\le 1)
\\
((\T\to\T) \to \up{\sty}, \T)
 & (\sty' =\T \land \odr{\sty}  > 1)
\end{cases}
\end{align*}
\(\pup{\SE\pK \app{t_1}{t_2}:\sty}\) has type \(\up{\sty}\).
\end{proof}

The grammar \(\GRAM'\) is well-defined:
For any rule \(A\,x_1\,\cdots\,x_\ell \to t\) in \(\RULES\) with
\begin{align*}&
\NONTERMS(A) = \kappa_1 \to \dots \to \kappa_\ell \to \T
\\&
\NONTERMS,x_1\COL\kappa_1,\ldots,x_\ell\COL\kappa_\ell\pK t:\T\,,
\end{align*}
by the above lemma, we have
\begin{align*}&
K\COL\T\to\T\to\T, \eup{\NONTERMS,x_1\COL\kappa_1,\ldots,x_\ell\COL\kappa_\ell}{\T} \pK 
\pup{\NONTERMS,x_1\COL\kappa_1,\ldots,x_\ell\COL\kappa_\ell\pK t:\T} :\T
\\&
K\COL\T\to\T\to\T, \eup{\NONTERMS,x_1\COL\kappa_1,\ldots,x_\ell\COL\kappa_\ell}{\T}
=
\NONTERMS',\eup{x_1\COL\kappa_1,\ldots,x_\ell\COL\kappa_\ell}{\T}
\end{align*}
and we can show that
\begin{align*}&
\eup{x_1\COL\kappa_1,\ldots,x_\ell\COL\kappa_\ell}{\T} \to \T
=
\up{((x_1\COL\kappa_1,\ldots,x_\ell\COL\kappa_\ell) \to \T)}
\end{align*}
by induction on the length \(\ell\).
Hence, the new rule in \(\RULES'\) is well-typed.

We can prove that \(\Lang(\GRAM) = \Lang(\GRAM')\), by showing that the graph relation
of the function \(\up{(-)}\) is a bisimulation relation
in a straightforward way.

 \else\fi


\begin{thebibliography}{10}

\bibitem{AehligSafety}
Klaus Aehlig, Jolie~G. de~Miranda, and C.-H.~Luke Ong.
\newblock Safety is not a restriction at level 2 for string languages.
\newblock In {\em Proceedings of FoSSaCS 2005}, volume 3441 of {\em LNCS},
  pages 490--504. Springer, 2005.

\bibitem{Ong09safety}
William Blum and C.{-}H.~Luke Ong.
\newblock The safe lambda calculus.
\newblock {\em Logical Methods in Computer Science}, 5(1), 2009.

\bibitem{Parys16diagonality}
Lorenzo Clemente, Pawel Parys, Sylvain Salvati, and Igor Walukiewicz.
\newblock The diagonal problem for higher-order recusion schemes is decidable.
\newblock In {\em Proceedings of LICS 2016}, 2016.

\bibitem{Czerwinski14}
Wojciech Czerwinski and Wim Martens.
\newblock A note on decidable separability by piecewise testable languages.
\newblock {\em CoRR}, abs/1410.1042, 2014.

\bibitem{DammTCS}
Werner Damm.
\newblock The {IO}- and {OI}-hierarchies.
\newblock {\em Theor. Comput. Sci.}, 20:95--207, 1982.

\bibitem{Engelfriet91}
Joost Engelfriet.
\newblock Iterated stack automata and complexity classes.
\newblock {\em Info. Comput.}, 95(1):21--75, 1991.

\bibitem{EngelfrietHTT}
Joost Engelfriet and Heiko Vogler.
\newblock High level tree transducers and iterated pushdown tree transducers.
\newblock {\em Acta Inf.}, 26(1/2):131--192, 1988.

\bibitem{Hague16POPL}
Matthew Hague, Jonathan Kochems, and C.{-}H.~Luke Ong.
\newblock Unboundedness and downward closures of higher-order pushdown
  automata.
\newblock In {\em Proceedings of POPL 2016}, pages 151--163, 2016.

\bibitem{Knapik01TLCA}
Teodor Knapik, Damian Niwinski, and Pawel Urzyczyn.
\newblock Deciding monadic theories of hyperalgebraic trees.
\newblock In {\em TLCA 2001}, volume 2044 of {\em LNCS}, pages 253--267.
  Springer, 2001.

\bibitem{Kobayashi13JACM}
Naoki Kobayashi.
\newblock Model checking higher-order programs.
\newblock {\em Journal of the ACM}, 60(3), 2013.

\bibitem{Kobayashi13LICS}
Naoki Kobayashi.
\newblock Pumping by typing.
\newblock In {\em Proceedings of LICS 2013}, pages 398--407. IEEE Computer
  Society, 2013.

\bibitem{KIT14FOSSACS}
Naoki Kobayashi, Kazuhiro Inaba, and Takeshi Tsukada.
\newblock Unsafe order-2 tree languages are context-sensitive.
\newblock In {\em Proceedings of FoSSaCS 2014}, volume 8412 of {\em LNCS},
  pages 149--163. Springer, 2014.

\bibitem{KMS13HOSC}
Naoki Kobayashi, Kazutaka Matsuda, Ayumi Shinohara, and Kazuya Yaguchi.
\newblock Functional programs as compressed data.
\newblock {\em Higher-Order and Symbolic Computation}, 2013.

\bibitem{KO09LICS}
Naoki Kobayashi and C.-H.~Luke Ong.
\newblock A type system equivalent to the modal mu-calculus model checking of
  higher-order recursion schemes.
\newblock In {\em Proceedings of LICS 2009}, pages 179--188. IEEE Computer
  Society Press, 2009.

\bibitem{Salvati15OI}
Gregory~M. Kobele and Sylvain Salvati.
\newblock The {IO} and {OI} hierarchies revisited.
\newblock {\em Inf. Comput.}, 243:205--221, 2015.

\bibitem{Ong06LICS}
C.-H.~Luke Ong.
\newblock On model-checking trees generated by higher-order recursion schemes.
\newblock In {\em LICS 2006}, pages 81--90. IEEE Computer Society Press, 2006.

\bibitem{Parys14}
Pawel Parys.
\newblock How many numbers can a lambda-term contain?
\newblock In {\em Proceedings of FLOPS 2014}, volume 8475 of {\em LNCS}, pages
  302--318. Springer, 2014.

\bibitem{SalvatiW15}
Sylvain Salvati and Igor Walukiewicz.
\newblock Typing weak {MSOL} properties.
\newblock In Andrew~M. Pitts, editor, {\em Proceedings of FoSSaCS 2015}, volume
  9034 of {\em LNCS}, pages 343--357. Springer, 2015.

\bibitem{DBLP:books/daglib/0018087}
Thomas Streicher.
\newblock {\em Domain-theoretic foundations of functional programming}.
\newblock World Scientific, 2006.

\bibitem{Tsukada14LICS}
Takeshi Tsukada and C.{-}H.~Luke Ong.
\newblock Compositional higher-order model checking via
  \emph{{\(\omega\)}}-regular games over b{\"{o}}hm trees.
\newblock In {\em Proceedings of {CSL-LICS} '14}, pages 78:1--78:10. {ACM},
  2014.

\bibitem{Zetzsche15}
Georg Zetzsche.
\newblock An approach to computing downward closures.
\newblock In {\em Proceedings of ICALP 2015}, volume 9135 of {\em LNCS}, pages
  440--451. Springer, 2015.

\end{thebibliography}
\end{document}